\documentclass[a4paper]{amsart} 
\usepackage[utf8]{inputenc}
\usepackage[english]{babel}
\usepackage[T1]{fontenc}

\usepackage{latexsym,amsmath,amsfonts,amssymb,amsthm,amstext,textcomp,mathabx,relsize}
\usepackage{mathrsfs,dutchcal,bbm}
\usepackage{enumerate}
\usepackage{hyperref}

\usepackage[capitalize]{cleveref}

\makeatletter
\def\@seccntformat#1{%
	\protect\textup{\protect\@secnumfont
		\ifnum\pdfstrcmp{subsection}{#1}=0 \bfseries\fi
		\ifnum\pdfstrcmp{subsubsection}{#1}=0 \itshape\fi
		\csname the#1\endcsname
		\protect\@secnumpunct
	}%
}
\renewcommand{\@upn}{}
\DeclareRobustCommand{\crefnosort}[1]{%
	\begingroup\@cref@sortfalse\cref{#1}\endgroup
}
\makeatother

\newcommand{\EE}{\mathbb{E}}
\newcommand{\CC}{\mathbb{C}}

\newcommand{\NN}{\mathbb{N}}
\newcommand{\PP}{\mathbb{P}}

\newcommand{\RR}{\mathbb{R}}

\newcommand{\ii}{\mathrm{i}}
\newcommand{\eul}{\mathrm{e}}
\newcommand{\ev}{\mathrm{ev}}
\newcommand{\supp}{\mathrm{supp}}

\newcommand{\Id}{\mathrm{d}}

\newcommand{\LO}{\mathcal{B}}
\newcommand{\dom}{\mathcal{D}}
\newcommand{\fdom}{\mathcal{Q}}
\newcommand{\Fock}{\mathcal{F}}
\newcommand{\UV}{\Lambda}
\newcommand{\ad}{a^\dagger}
\newcommand{\id}{\mathbbm{1}}

\newcommand{\expv}[1]{\epsilon(#1)}
\newcommand{\p}{\mathrm{p}}
\newcommand{\bos}{\mathrm{b}}
\newcommand{\ren}{\mathrm{ren}}
\newcommand{\ve}{\varepsilon}
\newcommand{\vp}{\varphi}

\newcommand{\vt}{\vartheta}

\newcommand{\vo}{\varpi}
\newcommand{\wt}[1]{\widetilde{#1}}
\newcommand{\ol}[1]{\overline{#1}} 

\newcommand{\wh}[1]{\widehat{#1}}  
 
\newcommand{\fr}[1]{\mathfrak{#1}}
\newcommand{\mc}[1]{\mathcal{#1}}
\newcommand{\scr}[1]{\mathscr{#1}}
\renewcommand{\Im}{\mathrm{Im}}
\renewcommand{\Re}{\mathrm{Re}}
\renewcommand{\le}{\leqslant}
\renewcommand{\ge}{\geqslant}
\theoremstyle{plain}
\newtheorem{thm}{Theorem}[section]

\newtheorem{lem}[thm]{Lemma}
\newtheorem{cor}[thm]{Corollary}
\newtheorem{prop}[thm]{Proposition}
\theoremstyle{definition}
\newtheorem{defn}[thm]{Definition}

\theoremstyle{remark}
 
\newtheorem{rem}[thm]{Remark}
\numberwithin{equation}{section}
\crefname{equation}{}{}
\Crefname{equation}{}{}
\crefname{enumi}{}{}
\Crefname{enumi}{}{}
\crefname{lem}{Lemma}{Lemmas}
\Crefname{lem}{Lemma}{Lemmas}
\crefname{thm}{Theorem}{Theorems}
\Crefname{thm}{Theorem}{Theorems}
\crefname{prop}{Proposition}{Propositions}
\Crefname{prop}{Proposition}{Propositions}
\crefname{defn}{Definition}{Definitions}
\Crefname{defn}{Definition}{Definitions}

\title[Feynman--Kac formula for the relativistic Nelson model]{Feynman--Kac formula and asymptotic behavior of the minimal energy 
for the relativistic Nelson model in two spatial dimensions}
\author{Benjamin Hinrichs}
\address{Benjamin Hinrichs, Friedrich-Schiller-Universit\"at Jena, Institut f\"ur Mathematik, Ernst-Abbe-Platz 2, 07743 Jena, Germany // 
	Present Address: Universit\"at Paderborn, Institut f\"ur Mathematik, Warburger Str. 100, 33098 Paderborn, Germany}
\email{benjamin.hinrichs@math.upb.de}

\author{Oliver Matte}
\address{Oliver Matte, Aalborg Universitet, Institut for Matematiske Fag, Skjernvej 4a, 9220 Aalborg, Denmark}
\email{oliver@math.aau.dk}
\usepackage{xcolor}\usepackage{cancel}

\begin{document}

\begin{abstract} 
	\noindent   We consider the renormalized relativistic Nelson model in two spatial dimensions
	for a finite number of spinless, relativistic quantum mechanical matter particles in interaction with
	a massive scalar quantized radiation field. We find a Feynman--Kac formula for the corresponding semigroup
	and discuss some implications such as ergodicity and weighted $L^p$ to $L^q$ bounds, for external potentials 
	that are Kato decomposable in the suitable relativistic sense.
	Furthermore, our analysis entails upper and lower bounds on the minimal energy for all values
	of the involved physical parameters when the Pauli principle for the matter particles is ignored.
	In the translation invariant case (no external potential)
	these bounds permit to compute the leading asymptotics of the minimal energy in the three
	regimes where the number of matter particles goes to infinity, the coupling constant for
	the matter-radiation interaction goes to infinity and the boson mass goes to zero.
\end{abstract}

\maketitle

\section{Introduction}

\subsection{General introduction}

\noindent
The Nelson model, describing a fixed number of spinless non-relativistic quantum particles linearly coupled to a field of spinless bosons, is an important testing ground for various mathematical aspects of the ultraviolet problem. Originally proposed by Nelson \cite{Nelson.1964,Nelson.1964c} as an example of an explicitly ultraviolet-renormalizable theory, it is the subject of ongoing research up to date. 

We investigate an adaption of the original Nelson model in which the particles have a relativistic dispersion relation, from now on referred to as the {\em relativistic Nelson model}, and focus on its probabilistic analysis. From a physical point of view, 
such an ultraviolet-renormalizable relativistic model is of special interest, because one expects the particles to obey the relativistic dispersion relation if the system can also describe very high energies. Due to a stronger ultraviolet divergence in three space dimensions, we restrict our attention to the case of two space dimensions. We assume the bosons to have a mass, for otherwise the model would be unstable in two dimensions.

The renormalization problem for the relativistic Nelson model has been treated in the articles \cite{Sloan.1974} and \cite{Gross.1973,DeckertPizzo.2014} in two and three spatial dimensions, respectively. Recently, Schmidt \cite{Schmidt.2019} renormalized the two-dimensional model using the technique of interior boundary conditions, which allowed him to prove stronger convergence results and to give an explicit expression for the renormalized operator. A related model, in which a fermionic scalar field is linearly coupled to a bosonic scalar field has recently been renormalized via a resummation scheme technique in \cite{AlvarezMoller.2022}, where the physical example can also only be treated in two space dimensions. 

Motivated by the recently revived interest in the model, we derive a Feynman--Kac formula for the renormalized relativistic Nelson Hamiltonian and discuss some of its applications. Similar formulas for the {\em non-relativistic} model have been derived in \cite{Nelson.1964c,GubinelliHiroshimaLorinczi.2014,MatteMoller.2018}. The technique we use is built upon that of \cite{MatteMoller.2018}. Explicitly, we derive an expression for the Feynman--Kac semigroup with ultraviolet cutoff and a cutoff dependent energy counter term, which still is mathematically meaningful when the cutoff is dropped. We then prove that the generator of this semigroup without cutoff is the norm resolvent limit of ultraviolet regularized Nelson Hamiltonians with added energy counter terms as the cutoff goes to infinity. In particular, the generator of our limiting semigroup agrees with the renormalized Hamiltonians constructed in \cite{Sloan.1974,Schmidt.2019}.
Compared to \cite{MatteMoller.2018}, we need to replace the Brownian motion as generator of the particle dynamics by a suitable L\'evy process.
Further, in comparison to \cite{Sloan.1974,Schmidt.2019}, the probabilistic technique easily allows us to treat a broad class of external potentials.

Along similar lines to \cite{MatteMoller.2018}, our
probabilistic analysis also entails a variety of lower bounds on the minimal energy of the relativistic Nelson model, with only a minor amount of extra work. Combined with more well-known trial function arguments, it further yields upper bounds,
which asymptotically match with the lower ones in the translation invariant case (no external potential). More precisely, we are able to compute the leading asymptotics of the minimal energy in the three regimes where the number of matter particles goes to infinity, the coupling constant for the matter-radiation interaction goes to infinity and the boson mass goes to zero. Here, we are always discussing the minimal energy in our model without any symmetry restrictions imposed on the matter particles (no Pauli principle).

Further results of our analysis are weighted $L^p$ to $L^q$ bounds on the semigroup, which also is shown to map
bounded sets in $L^p$ onto equicontinuous sets of functions. As observed in \cite{Matte.2016,MatteMoller.2018}
the general structure of our Feynman--Kac formula easily entails ergodicity of the semigroup, 
a result that has been obtained in \cite{Sloan.1974b} for the fiber Hamiltonian at total momentum zero
in the translation invariant case but seems to be new for the full Hamiltonian. 
A well-known implication of our continuity 
and ergodicity results is the continuity, non-degeneracy and strict positivity (for the appropriate phase factor) of ground
state eigenvectors (if any). The weighted $L^p$ to $L^q$ bounds can be used to turn $L^2$-bounds on the
exponential localization of ground state eigenvectors into pointwise decay estimates; 
see~\cite{HiroshimaMatte.2019}
for such results in the original Nelson model. A Markov property satisfied by our Feynman--Kac integrands
and the ergodicity are the crucial prerequisites for the construction of path measures associated with ground states; 
see again~\cite{HiroshimaMatte.2019}.

\subsection*{Structure and Notation}

The article is structured as follows: In the remainder of this introduction we define the relativistic 
Nelson model studied here and take a first glance at our Feynman--Kac formula, as its thorough presentation is 
somewhat lengthy. Precise statements of our results on the Feynman--Kac formula, ergodicity,
as well as bounds on and asymptotics of the minimal energy can be found in \cref{secmainres}.
In \cref{sec:FK} we prove the Feynman--Kac formula, relying on a number of properties of our Feynman--Kac semigroups
as well as an integral equation satisfied by the Fock space operator part of the Feynman--Kac integrands with ultraviolet cutoff.
It is not difficult to verify the latter integral equation, whence this is done first in \cref{sec:WUVinteq}.
Establishing all required results on the semigroups is the objective of the successional three sections:
In \cref{sec:basicproc} we provide technical results on the two main contributions to the 
 analogue of Feynman's complex action in our model, an effective interaction potential and a martingale part.
 The full complex action is studied in \cref{sec:complex}; the bounds derived therein will eventually imply our lower
 bounds on the minimal energy.
 In \cref{sec:FKInt} we then treat the full integrands in the Feynman--Kac formula and prove
 Markov, semigroup and continuity properties as well as weighted $L^p$ to $L^q$ bounds.
The final \cref{sec:upbound} is devoted to proving upper bounds on the minimal energy. 
The main text is followed by four appendices respectively dealing with the relativistic Kato class,
basic integral processes attaining values in the one-boson Hilbert space,
 a corollary of an exponential tail estimate for L\'{e}vy type stochastic integrals
 due to Applebaum and Siakalli \cite{Applebaum.2009,Siakalli.2019} and bounds on the free relativistic semigroup.

Let us fix some general notational conventions:
\begin{itemize}
	\item The characteristic function of a set $M\subset \RR^n$, $n\in\NN$, is denoted by $\chi_M$.
	\item The two-dimensional open ball of radius $r>0$ about $0$ is denoted as 
	\begin{align}\label{2Dball}
	B_r \coloneq \{y\in\RR^2|\,|y|<r\};\qquad B_0\coloneq\emptyset.
	\end{align}
	\item The open unit ball about the origin in $\RR^{2N}$ is denoted as 
	\begin{align}\label{2NDball}
	B_1' \coloneq \{x\in\RR^{2N}|\,|x|<1\}.
	\end{align}
	\item For $a,b\in\RR$, we write $a\wedge b = \min \{a,b\}$ and $a\vee b = \max\{a,b\}$.
	\item The set of bounded operators on some Banach space $\mc{X}$ is denoted by $\LO(\mc{X})$.
	\item The Borel $\sigma$-algebra of some topological space $S$ is denoted by $\fr{B}(S)$.
	\item The symbol $\dom(\cdot)$ denotes domains of definition. If $A$ is  a selfadjoint operator in some Hilbert space,
then 	$\dom(A)$ is equipped with the graph norm of $A$.
\end{itemize}

\subsection{The relativistic Nelson model in two spatial dimensions}

\noindent
Let us now define the model studied in this paper. For brevity we shall freely use some standard notation
for operators acting in the bosonic Fock space $\Fock$, which is modeled over the one-boson space $L^2(\RR^2)$; the definitions of $\Fock$ and all relevant operators acting in it are recalled in \cref{subsec:Fock} below.

In the translation invariant case our model is determined by four parameters:
The number $N\in\NN$ and mass $m_{\p}\ge0$ of the matter particles, the boson mass $m_{\bos}>0$ 
and the coupling constant $g\in\RR\setminus\{0\}$ for the matter-radiation interaction. The dispersion relations for
a single matter particle and a single boson are given, respectively, by
\begin{align}\label{defpsi}
\psi(\eta)&\coloneq(|\eta|^2+m_{\p}^2)^{1/2}-m_{\p},\quad\eta\in\RR^2,
\\\label{defomega}
\omega(k)&\coloneq (|k|^2+m_{\bos}^2)^{1/2},\quad k\in\RR^2.
\end{align}
The coupling function for the matter-radiation interaction is given by
\begin{align}\label{defv}
v&\coloneq g\omega^{-1/2}.
\end{align}
We shall use the capital letter $K_i$ to denote multiplication by $k_i$, i.e.,
\begin{align}\label{eq:defK}
(K_if)(k)&\coloneq k_if(k),\quad k\in \RR^2,\quad\text{for every $f:\RR^2\to\CC$ and $i\in\{1,2\}$.}
\end{align}
Since $v\notin L^2(\RR^2)$, the bosonic field operators $\vp(\eul^{-\ii K\cdot y} v)$, $y\in\RR^2$, are ill-defined.
This makes the energy renormalization procedure necessary that we employ in the construction of the
Hamiltonian $H$ hereinafter: 

Let us write
\begin{align}\label{defxxj}
x&=(x_1,\ldots,x_N),\quad\text{with $x_1,\ldots,x_N\in\RR^2$,}
\end{align}
and introduce the ultraviolet cutoff coupling functions
\begin{align}\label{defvUV}
v_{\UV}\coloneq \chi_{B_{\UV}}v,
\quad \UV\in[0,\infty).
\end{align}
Since $v_{\UV}\in L^2(\RR^2)$ for all $\UV\in[0,\infty)$, 
the field operators $\vp(\eul^{-\ii K\cdot x_j}v_\UV)$ are now well-defined.
Furthermore, they are infinitesimally bounded, uniformly in $x$, with respect to 
the radiation field energy $\Id\Gamma(\omega)$.

Finally, we add an external potential 
\begin{align*}
V&=V_+-V_-:\RR^{2N}\longrightarrow \RR,\quad\text{with $V_+\coloneq V\vee0$.}
\end{align*} 
Throughout this article $V$ is assumed to be Kato decomposable with respect to the L\'{e}vy process
describing the trajectories of the $N$ matter particles. 
The latter notion will be explained in the beginning of \cref{subsec:FK}. For the moment
it is sufficient to know that our assumptions on $V$ imply its local integrability and infinitesimal form boundedness
of its negative part $V_-$  with respect to the Hamiltonian for the free matter particles 
$\sum_{j=1}^N \psi(-\ii\nabla_{x_j})$. 

We can now define the regularized relativistic Nelson Hamiltonian with ultraviolet cutoff $\UV\in[0,\infty)$, denoted $H_{\UV}$.
It acts in the Hilbert space $L^2(\RR^{2N},\Fock)$. Analogously to \eqref{defxxj} we write
\begin{align}\label{defxixij}
\xi&=(\xi_1,\ldots,\xi_N),\quad \text{with $\xi_1,\ldots,\xi_N\in\RR^2$,}
\end{align}
and denote by $\wh{\Psi}$ the $\Fock$-valued Fourier transform of $\Psi\in L^2(\RR^{2N},\Fock)$. 
Then the form domain $\fdom(H_{\UV})$ of $H_{\UV}$ is the set of all $\Psi\in L^2(\RR^{2N},\Fock)$ such that
$\Psi(x)$ is in the form domain of the radiation field energy $\Id\Gamma(\omega)$ for a.e. $x$ and such that
\begin{align*}
\mathfrak{q}[\Psi]&\coloneq  \sum_{j=1}^N\int_{\RR^{2N}}
\psi(\xi_j)\|\wh{\Psi}(\xi)\|_{\Fock}^2\Id\xi+\int_{\RR^{2N}} \big(V_+(x)\|\Psi(x)\|_{\Fock}^2
+\|\Id\Gamma(\omega)^{1/2}\Psi(x)\|^2_{\Fock}\big)\Id x
\end{align*}
is finite. By the above considerations the quadratic form
$\mathfrak{h}_{\UV}:\fdom(H_{\UV})\to\RR$ given by
\begin{align}\nonumber
\mathfrak{h}_{\UV}[\Psi]&\coloneq  \mathfrak{q}[\Psi]-\int_{\RR^{2N}} V_-(x)\|\Psi(x)\|_{\Fock}^2\Id x
\\\label{def:HUV}
&\quad
+\sum_{j=1}^N\int_{\RR^{2N}}\langle\Psi(x)|\vp(\eul^{-\ii K\cdot x_j}v_\UV)\Psi(x)\rangle_{\Fock}\Id x
+NE_{\UV}^{\ren}\|\Psi\|^2,
\end{align}
for all $\Psi\in \fdom(H_{\UV})$, is closed and lower semibounded.
In its definition we already added the energy counter term $NE_{\UV}^{\ren}$ with
\begin{align}\label{def:EUVren}
E_{\UV}^{\ren}&\coloneq \int_{B_{\UV}}\frac{g^2}{\omega(k)(\omega(k)+\psi(k))}\Id k,\quad \UV\in[0,\infty).
\end{align}
By definition $H_{\UV}$ is the unique selfadjoint operator in $L^2(\RR^{2N},\Fock)$ representing $\mathfrak{h}_{\UV}$.

The renormalized relativistic Nelson Hamiltion in two space dimensions is now given by
\begin{align}\label{reslimHUV}
H&\coloneq \underset{\UV\to\infty}{\textrm{norm-resolvent-lim}} \ H_\UV.
\end{align}
As alluded to above, existence of
the norm resolvent limit has been established in \cite{Sloan.1974,Schmidt.2019}, at least for the case $V=0$.

Our Feynman--Kac formula for $H$ (\cref{thm:FK})
can be stated as follows: For all  $t>0$ and $\Psi\in L^2(\RR^{2N},\Fock)$, 
there is a unique continuous representative of $e^{-tH}\Psi$ satisfying
\begin{align*}
	\begin{aligned}
			(e^{-tH}\Psi)(x) &= 
			\EE\left[e^{u_t^N(x)-\int_0^tV(X_s^x)\Id s}F_{t/2}(-U_t^{N,-}(x))F_{t/2}(-U_t^{N,+}(x))^*\Psi(X_t^x)\right],
	\end{aligned}
\end{align*}
for all $x\in\RR^{2N}$.
Here $X$ is a vector of $N$ independent identically distributed $\RR^2$-valued L\'evy processes with L\'{e}vy symbol $-\psi$,
and $X^x=x+X$. The real-valued continuous and adapted stochastic process $u^N(x)$ is the complex action and 
$U^{N,\pm}(x)$ are continuous adapted stochastic processes taking values in $L^2(\RR^2)$. 
Furthermore, for fixed $h\in L^2(\RR^2)$, $F_t(h)$ is an operator on $\Fock$ 
explicitly given as an exponential series of creation operators $\ad(h)$
controlled by the semigroup at time $t$
of the free radiation field; see \cref{defFtg}. Its adjoint $F_t(h)^*$ contains annihilation terms.
The processes $u^N(x)$ and $U^{N,\pm}(x)$ are given by explicit formulas involving (stochastic) integrals
that are mathematically meaningful directly without any cutoff. Nevertheless, 
we will actually re-prove the existence of \eqref{reslimHUV}, mainly to verify that $H$ generates the
semigroup defined by the right hand side of the Feynman--Kac formula. Here we shall
make use of our assumption that $V$ be Kato decomposable in a relativistic $N$-particle sense. 

If the boson mass was zero, there would still exist a canonical selfadjoint realization of $H_{\UV}$ for suitable $V$, 
whose spectrum was, however, unbounded from below; see, e.g., Lemma~\ref{lemesaNC} and Corollary~\ref{corubGauss1}
below. Hence, an energy renormalization is impossible
for zero boson mass. This is why we assume $m_{\bos}$ to be strictly positive, while $m_{\p}$ is only
required to be non-negative.

\section{Main results}\label{secmainres}

\noindent
In this section we present our main results in detail. To this end we first recall some necessary Fock space calculus
in \cref{subsec:Fock} as well as relevant facts on the L\'{e}vy process $X$ 
and related probabilistic objects in \cref{ssecLevyX}.
In \cref{subsec:FK} we will step by step introduce all processes appearing in our Feynman--Kac formula
and state the latter in \cref{thm:FK}. Applications of the Feynman--Kac formula are described in 
\cref{subsecerg} (ergodicity), \cref{ssectwosided} (two-sided bounds on the minimal energy) 
and \cref{ssecasymp} (asymptotics of the minimal energy).

\subsection{Important operators on Fock space}\label{subsec:Fock}
Here we shall briefly recall the definitions from Fock space theory necessary for an understanding of this article. For a more detailed introduction, we refer the reader to \cite{Arai.2018,Parthasarathy.1992}.

The bosonic Fock space over $L^2(\RR^2)$ is the direct sum of Hilbert spaces
\begin{align}\label{defFock}
\Fock&\coloneq \CC\oplus\bigoplus_{n=1}^\infty L_{\mathrm{sym}}^2(\RR^{2n}),
\end{align}
with $L_{\mathrm{sym}}^2(\RR^{2n})$ denoting the closed subspace in $L^2(\RR^{2n})$ of all its elements
$\phi_n$ satisfying $\phi_n(k_{\pi(1)},\ldots,k_{\pi(n)})=\phi_n(k_1,\ldots,k_n)$,
a.e. for every permutation $\pi$ of $\{1,\ldots,n\}$; here $k_1,\ldots,k_n\in\RR^2$. 

The first operator in $\Fock$ we
introduce is the selfadjoint radiation field energy denoted~$\Id\Gamma(\omega)$. Each subspace in \eqref{defFock}
is reducing for $\Id\Gamma(\omega)$. Its action on $\phi=(\phi_n)_{n=0}^\infty\in\dom(\Id\Gamma(\omega))$ is given by
\begin{align}\label{defdGammaomega}
(\Id\Gamma(\omega)\phi)_0=0\quad\text{and}\quad
(\Id\Gamma(\omega)\phi)_n(k_1,\ldots,k_n)&=\sum_{j=1}^n\omega(k_j)\phi_n(k_1,\ldots,k_n),
\end{align}
a.e. for every $n\in\NN$. Here $\phi\in\Fock$ is in the domain of $\Id\Gamma(\omega)$, if and only if the
expressions on the right hand sides in \eqref{defdGammaomega} define a new vector in $\Fock$.

Next, we introduce the annihilation operator, $a(f)$, associated with $f\in L^2(\RR^2)$. 
It is the closed operator whose action on $\phi=(\phi_n)_{n=0}^\infty\in\dom(a(f))$ reads
\begin{align}\label{defaf}
(a(f)\phi)_n(k_1,\ldots,k_n)&=(n+1)^{1/2}\int_{\RR^2}\ol{f(k)}\phi_{n+1}(k,k_1,\ldots,k_n)\Id k,
\end{align}
a.e. for every $n\in\NN$, and $(a(f)\phi)_0=\langle f|\phi_1\rangle$. 
Again $\phi\in\Fock$ is in $\dom(a(f))$, if and only if the previous right hand sides yield a vector in $\Fock$.
Finally, the creation operator, $\ad(f)$, and the selfadjoint field operator, $\vp(f)$, associated with $f$ are given by
\begin{align*}
\ad(f)&\coloneq a(f)^*,\qquad \vp(f)\coloneq(\ad(f)+a(f))^{**}.
\end{align*}

As is well-known, $a(f)$, $\ad(f)$ and $\vp(f)$ are relatively $\Id\Gamma(\omega)^{1/2}$-bounded
for all $f\in L^2(\RR^2)$, because $m_{\bos}>0$; see, e.g., \cite[Theorem~5.16]{Arai.2018}.
More generally, the $n$-th powers of these operators are relatively $\Id\Gamma(\omega)^{n/2}$-bounded
for all $n\in\NN$. 
Employing suitable relative bounds we can in fact verify operator norm convergence of the following series, 
that we shall encounter in our Feynman--Kac formula, 
\begin{align}\label{defFtg}
F_t(h)&\coloneq \sum_{n=0}^\infty \frac{1}{n!}\ad(h)^n\eul^{-t\Id\Gamma(\omega)},\quad t>0,\,h\in L^2(\RR^2).
\end{align}
These series have been introduced in \cite{GueneysuMatteMoller.2017} and details can be found in Appendix~6 of that paper.
For instance,
 \begin{align}\label{normFt}
\|F_t(h)\|&\le \mc{S}(\|h\|_t),\quad t>0,\,h\in L^2(\RR^2),
\end{align}
where the time-dependent norm $\|\cdot\|_t$ on $L^2(\RR^2)$ is given by
\begin{align}\label{deftnorm}
\|h\|_t^2&\coloneq\int_{\RR^2}\bigg(1+\frac{1}{t\omega(k)}\bigg)|h(k)|^2\Id k,
\end{align}
and
\begin{align}\label{def:S}
\mc{S}(z)&\coloneq \sum_{n=1}^\infty\frac{2^nz^n}{(n!)^{1/2}},\quad z\in\CC.
\end{align}

\subsection{L\'{e}vy processes for particle paths}\label{ssecLevyX}

\noindent
Next, we introduce in detail the L\'{e}vy process $X$ describing (fictitious) paths of the $N$ relativistic matter particles.
Furthermore, we collect some remarks on associated Poisson point processes and stochastic integrals which are
necessary to understand the stochastic calculus applied later on. The intention is to clarify notation and make this
article more accessible for readers who might not be well-acquainted with L\'{e}vy processes. Relevant textbooks
are, e.g., \cite{Applebaum.2009,IkedaWatanabe.1981,Metivier.1982}.

Throughout the article, let us assume that $(\Omega,\fr{F},(\fr{F}_t)_{t\ge0},\PP)$ is a filtered probability space satisfying the 
usual hypotheses of completeness and right-continuity, i.e., $(\Omega,\fr{F},\PP)$ is complete, $\fr F_0$ contains all sets of $\PP$-measure zero and $\fr F_t = \bigcap_{s> t} \fr F_s$ for all $t\ge0$. 
The letter $\EE$ denotes expectations with respect to $\PP$.

\subsubsection{L\'{e}vy processes for relativistic particles}

\noindent
For every $j\in\{1,\ldots,N\}$,
we let $X_j=(X_{j,t})_{t\ge0}$ denote an $(\fr{F}_t)_{t\ge0}$-L\'{e}vy process 
associated with the L\'{e}vy symbol $-\psi$ given by \eqref{defpsi}. 
Hence, each $X_j$ is an adapted stochastically continuous $\RR^2$-valued process with stationary increments
such that $X_{j,0}=0$, $\PP$-a.s., $X_{j,t}-X_{j,s}$ and $\fr{F}_s$ are independent for all $t>s\ge0$
and the law of $X_{j,t}$ has the density $\rho_{m_{\p},t}\coloneq(2\pi)^{-1}(\eul^{-t\psi})^{\vee}$ for all $t>0$.
Here ${}^\vee$ denotes inverse Fourier transformation.
The well-known explicit expressions for $\rho_{m_{\p},t}$ are recalled (in arbitrary dimension) in \cref{lawXd0,forrhomp}.
It might be helpful for some readers to mention that each $X_j$
is an example of an inverse Gaussian process; see, e.g., \cite[Example~1.3.32]{Applebaum.2009}.

The processes $X_1, \ldots, X_N$ model the individual relativistic matter particles and
we assume these processes to be independent. We gather them in the $\RR^{2N}$-valued process
$X\coloneq(X_1,\ldots,X_N)$ which again is $(\fr{F}_t)_{t\ge0}$-L\'{e}vy.
In this article we stick to the convention that Banach space-valued stochastic processes with time horizon $[0,\infty)$ 
are called c\`{a}dl\`{a}g, iff {\em all} their paths are right-continuous on $[0,\infty)$ and have left-sided limits 
at every point in $(0,\infty)$. Since we are working under the usual hypotheses, we may and shall
assume that $X$ is c\`{a}dl\`{a}g in this sense. Then $X_{t-}$ stands for the pointwise defined
left limit $\lim_{s\uparrow t}X$ for all $t>0$. 

For every $x\in\RR^{2N}$, we further abbreviate $X^x\coloneq x+X$ and $X_j^{x}\coloneq x_j+X_j$, so that for instance
$X_{j,t-}^x=x_j+\lim_{s\uparrow t}X_{j,s}$ for $t>0$; recall \eqref{defxxj}.

\subsubsection{L\'{e}vy measures}

\noindent
To be able to work with the pure jump processes $X$ and $X_j$ we have to know their L\'{e}vy measures. Recall that a L\'evy measure is a Borel measure $\lambda$ on $\RR^d$ such that $\lambda(\{0\})=0$ and $1\wedge |\cdot|\in L^2(\lambda)$. In our case they describe the distribution of the jumps of $X$ and $X_j$, respectively; see \cite[\textsection 1.2.4]{Applebaum.2009} for a detailed introduction.

The L\'evy measure of $X_j$, which does not depend on $j$, is denoted by $\nu:\fr{B}(\RR^2)\to[0,\infty]$. 
As is well known, $\nu$ has the density $\rho$ 
with respect to the two-dimensional Lebesgue-Borel measure, where
\begin{align}\label{def:nu}
	\forall y\in\RR^2\setminus\{0\}:\quad
	\rho(y)&\coloneq
	\begin{cases}
		\displaystyle \frac{m_{\p}^{3/2}}{2^{1/2}(\pi|y|)^{3/2}}K_{3/2}(m_{\p}|y|),&\text{if $m_{\p}>0$},
		\\[1em]
		\displaystyle \frac{1}{2\pi|y|^3},&\text{if $m_{\p}=0$,}
	\end{cases}
\end{align}
and, say, $\rho(0)=0$.
Here, $K_{3/2}$ is a modified Bessel function of the third kind; see \cref{app:Bessel} for more details.
In fact we can directly obtain $\rho(y)$ from the law $\rho_{m_p,t}(y)$ of $X_j$ by calculating the limit $\lim_{t\downarrow 0}\frac 1t\rho_{m_p,t}(y)$ for $y\in\RR^2\setminus\{0\}$.
Thus, $\Id\nu(y)=\rho(y)\Id y$ and, by \cref{bdBesselK}, 
\begin{align}\label{intcondnu}
\int_{\RR^2}|y|^{1+\ve}\wedge|y|^{1-\ve}\Id\nu(y)<\infty,\quad \ve>0;\qquad 
\int_{B_1(0)}|y|\Id\nu(y)=\infty.
\end{align}
The former integrability property reveals that $\nu$ is indeed a L\'evy measure. 
It is the one associated with each process $X_j$, as the L\'{e}vy--Khintchine formula
for the logarithm of the characteristic function of $X_{j}$ reads
\begin{align}\label{LKpsi}
-\psi(\eta)&=-\psi(-\eta)=\int_{\RR^2}(\eul^{-\ii \eta\cdot y}-1+\ii\chi_{B_1}(y)\eta\cdot y)\Id\nu(y),\quad \eta\in\RR^2,
\end{align}
with the two-dimensional ball $B_1$ as defined in \cref{2Dball}; see, e.g., \cite[\textsection 1.2.6]{Applebaum.2009}.
This equality also reveals that $X_j$ is a pure jump process, since the L\'{e}vy--Khintchine formula for the characteristic function directly corresponds to  the L\'evy--It\^o decomposition of the L\'evy process, cf. \cite[\textsection 1.2.4]{Applebaum.2009} for details.

Since $X_1,\ldots,X_N$ are independent L\'{e}vy processes 
with common L\'{e}vy symbol $-\psi$, the symbol of the 
$2N$-dimensional process $X$, call it $-\psi_X$, satisfies
\begin{align*}
\eul^{-t\psi_X(\xi)}&=\EE[\eul^{\ii(\xi_1\cdot X_{1,t}+\dots+\xi_N\cdot X_{N,t})}]=\prod_{j=1}^N\EE[\eul^{\ii\xi_j\cdot X_{j,t}}]
=\prod_{j=1}^N\eul^{-t\psi(\xi_j)},\quad\xi\in\RR^{2N},
\end{align*}
for every $t\ge0$, where we use the notation \eqref{defxixij}. Thus,
\begin{align}\label{LKPsiX}
\psi_X(\xi)&=\sum_{j=1}^N\psi(\xi_j)
=-\int_{\RR^{2N}}(\eul^{-\ii\xi\cdot z}-1+\ii\chi_{B_1'}(z)\xi\cdot z)\Id \nu_X(z),\quad\xi\in\RR^{2N}.
\end{align}
Here $B_1'$ is the $2N$-dimensional open unit ball defined in \cref{2NDball}, and the L\'{e}vy measure 
$\nu_X:\fr{B}(\RR^{2N})\to[0,\infty]$ of $X$ is
\begin{align}\label{eq:nuX}
\nu_X&\coloneq\sum_{j=1}^N\delta_0\otimes\dots\otimes\delta_0\otimes\underbrace{\nu}_{\text{$j$'th factor}}
\otimes\delta_0\otimes\dots\otimes\delta_0,
\end{align}
with $\delta_0$ denoting the Dirac measure on the Borel sets of $\RR^2$ concentrated in $0$.
To read off the formula for $\nu_X$, we made use of \eqref{LKpsi} and observed that $\chi_{B_1'}(z)=\chi_{B_1}(z_\ell)$
whenever $z_j=0$ for all $j\not=\ell$.

\subsubsection{Infinitesimal generator}

\noindent
For every bounded $f\in C^2(\RR^{2N})$ with bounded first and second order derivatives, we set
\begin{align}\label{def:psiXnabla}
(\psi_X(-\ii\nabla)f)(x)&\coloneq -\int_{\RR^{2N}}(f(x+z)-f(x)-\chi_{B_1'}(z)\nabla f(x)\cdot z)\Id\nu_X(z),
\end{align}
for every $x\in\RR^{2N}$.
Here the integrals on the right hand side exist due to the assumptions on $f$, Taylor's formula, \cref{intcondnu,eq:nuX}.

For instance, assume that $f=\check{g}$ is the inverse Fourier transform of some Borel measurable function 
$g:\RR^{2N}\to\CC$ such that $\RR^{2N}\ni\xi\mapsto(1+|\xi|^2)g(\xi)$ is integrable. 
Then, in view of  \eqref{LKPsiX} and Fubini's theorem, the above definition of $\psi_X(-\ii\nabla)$ 
agrees with the usual one based on Fourier transformation, i.e.,
\begin{align}\label{psiXFourier}
(\psi_X(-\ii\nabla)f)(x)&= \frac{1}{(2\pi)^N}\int_{\RR^{2N}}\eul^{\ii \xi\cdot x}\psi_X(\xi)g(\xi)\Id\xi,\quad x\in\RR^{2N}.
\end{align}
Let us for orientation mention that
functions of the form $f=\check{g}$ with $g$ as above also belong to the domain of definition of
the infinitesimal generator of the Feller semigroup associated with $X$; applying this generator to $f$ yields
$-\psi_X(-\ii\nabla)f$.

\subsubsection{Associated Poisson point processes of jumps and stochastic integrals}\label{ssecPPP}

\noindent
Twice in this article, we shall make crucial use of It\^{o}'s formula for $X$:
in our derivation of the Feynman--Kac formula with ultraviolet cutoff and in order to find
a new formula for the complex action that still makes sense when the cutoff is dropped.
Therefore, we shall now briefly explain the point processes and corresponding stochastic integrals
appearing there. It\^{o}'s formula itself will be discussed when needed in \cref{ssecItoformula}.
Note that the complex action without cutoff will contain a martingale part given as a stochastic integral.

We denote the jumps of $X$ by $\Delta X_t(\gamma)\coloneq X_t(\gamma)-X_{t-}(\gamma)$, $t>0$, and notice that
the set of non-zero jump times
$D(\gamma)\coloneq \{t>0|\,\Delta X_t(\gamma)\not=0\}$ is at most countable for every $\gamma\in\Omega$,
since $X$ is c\`{a}dl\`{a}g.
Then the Poisson point process $N_X$ associated with $X$ is the random measure given by
\begin{align*}
N_X(A)&\coloneq \sum_{t\in D}\chi_A(t,\Delta X_t),\quad A\in\fr{B}([0,\infty)\times\RR^{2N}).
\end{align*}
When reading some formulas below,
it is helpful to keep in mind that $N_X((0,t]\times(\RR^{2N}\setminus B_1'))<\infty$, $\PP$-a.s., for every $t>0$.
At every fixed elementary event, $N_X$ is an at most countable sum of Dirac measures and in particular
\begin{align*}
\int_{(0,t]\times\RR^{2N}}Z_s(z)\Id N_X(s,z)&=\sum_{s\in D\cap(0,t]}Z_s(\Delta X_s),
\end{align*}
for every product measurable integrand $Z:[0,\infty)\times\RR^{2N}\times\Omega\to[0,\infty]$.
Here and below, $Z_s(z)\coloneq Z(s,z,\cdot)$.

For every $B\in\fr{B}(\RR^{2N})$ with $\nu_X(B)<\infty$, 
$(N_X((0,t]\times B))_{t\ge0}$ is an $(\fr{F}_t)_{t\ge0}$-Poisson process with intensity $\nu_X(B)$,
i.e., satisfying $\EE[N_X((0,t]\times B)] = t\nu_X(B)$. We remark that this is the precise way in which the L\'evy measure $\nu_X$ describes the distribution of the jumps of $X$.
Further it implies that
the process $(\wt{N}_X((0,t]\times B))_{t\ge0}$ given by
\begin{align}\label{wtNXNXnu}
\wt{N}_X((0,t]\times B)&\coloneq N_X((0,t]\times B)-t\nu_X(B),\quad t\ge0,
\end{align}
is a martingale. There exists a theory of stochastic integration extending the relation \cref{wtNXNXnu}; see, e.g., \cite[\textsection2~Sec.~3]{IkedaWatanabe.1981}.
As corresponding integrands we only encounter elements of $\scr{H}_2$ in this article, where, for $p\in\{1,2\}$,
$\scr{H}_p$ denotes the set of predictable maps
$Z:[0,\infty)\times\RR^{2N}\times\Omega\to\RR$ such that $|Z|^p$ is
integrable over $[0,t]\times\RR^{2N}\times\Omega$ with respect to
$\Id s\otimes\nu_X\otimes\PP$ for all $t\ge0$. 
For every $Z\in\scr{H}_2$ the stochastic integral process $(\int_{(0,t]\times\RR^{2N}}Z_s(z)\Id\wt{N}_X(s,z))_{t\ge0}$
is a c\`{a}dl\`{a}g $L^2$-martingale satisfying It\^{o}'s isometry
\begin{align}\label{ItoIso}
\EE\bigg[\bigg(\int_{(0,t]\times\RR^{2N}}Z_s(z)\Id\wt{N}_X(s,z)\bigg)^2\bigg]&=
\EE\bigg[\int_0^t\int_{\RR^{2N}}Z_s(z)^2\Id\nu_X(z)\,\Id s\bigg],
\end{align}
for all $t\ge0$.
In the case $Z\in\scr{H}_1\cap\scr{H}_2$ we $\PP$-a.s. know that, for all $t\ge0$,
$Z$ is integrable over $(0,t]\times\RR^{2N}$ with respect to $N_X$ and
\begin{align}\nonumber
&\int_{(0,t]\times\RR^{2N}}Z_s(z)\Id\wt{N}_X(s,z)
\\\label{wtNXintH1}
&=\int_{(0,t]\times\RR^{2N}}Z_s(z)\Id N_X(s,z)
-\int_0^t\int_{\RR^{2N}}Z_s(z)\Id\nu_X(z)\,\Id s.
\end{align}

The above integrals show up in the L\'{e}vy--It\^{o} decomposition of $X$,
\begin{align}\label{LevyItoX}
X_t&=\int_{(0,t]\times(\RR^{2N}\setminus B_1')}z\Id N_X(s,z)
+\int_{(0,t]\times B_1'}z\Id \wt{N}_X(s,z),\quad t\ge0.
\end{align}
Here, the first member on the right hand side is the $\PP$-a.s. finite sum of all jumps of $X$
of size larger than $1$. The second member is referred to as a compensated sum of all remaining, small jumps;
notice that, by \cref{intcondnu,eq:nuX}, its integrand is in $\scr{H}_2$ but not in $\scr{H}_1$.

\subsection{The Feynman--Kac formula}\label{subsec:FK}

\noindent
We can now move to the presentation of our Feynman--Kac formula. We start by properly introducing the class
of external potentials considered in this article:

\subsubsection{Relativistic Kato class for $N$ particles}\label{sssecKato}

\noindent
Following \cite{CarmonaMastersSimon.1990}
we introduce the Kato class corresponding to the process $X$ as
\begin{align}\label{def:Kato}
	\mc K_X \coloneq \left\{f:\RR^{2N}\to [0,\infty) \ \mbox{ Borel meas.}\,\middle| \
	\lim_{t\downarrow0}\sup_{x\in\RR^{2N}}\EE\left[ \int_0^t f(X^x_s)\Id s\right] = 0 \right\}.
\end{align}
As indicated earlier the external potential $V$ is always assumed to be Kato decomposable in the sense of $\mc K_X$, i.e., 
$\chi_CV_+,V_-\in \mc K_X$ for all compact $C\subset\RR^{2N}$. As a consequence, 
$V$ is locally integrable and $V_-$ is infinitesimally form 
bounded with respect to the free particle Hamiltonian $\sum_{j=1}^N \psi(-\ii\nabla_{x_j})$; 
cf. \cite[Theorem~III.1 and the paragraph following its proof]{CarmonaMastersSimon.1990}. We also know that 
\begin{align}\label{defNVx}
\scr{N}_V(x)&\coloneq
\big\{\gamma\in\Omega\,\big|\: V(X^x_\bullet(\gamma)):[0,\infty)\to\RR\;\text{is not locally integrable}\big\}
\end{align}
is a $\PP$-zero set for every $x\in\RR^{2N}$. Once and for all we fix the convention that the integral paths
$\int_0^\bullet V(X_s^x)\Id s$ equal $0$ on $\scr{N}_V(x)$.
In \cref{app:Kato} we discuss some further properties of Kato decomposable potentials to be used later on.

For example, $V$ is Kato decomposable in the sense of $\mc K_X$, if it has the form
\begin{align*}
V(x) &= \sum_{j=1}^N V_{\{j\}}(x_j)+\sum_{1\le i<j\le N}V_{\{i,j\}}(x_i-x_j),\quad x\in\RR^{2N},
\end{align*} 
where $\chi_C(V_{A})_+,(V_{A})_-\in \mc K_{X_1}$ 
for all compact $C\subset\RR^2$ and all subsets $A\subset\{1,\ldots,N\}$ of cardinality one or two.
Here $\mc K_{X_1}$ is defined as in \cref{def:Kato} with $\RR^{2N}$ replaced by $\RR^2$ and $X$ replaced by $X_1$.
This follows from the independence of $X_1,\ldots,X_N$ and since, for $i<j$ and $s\ge0$, the laws of 
$X_{i,s}$ and $X_{j,s}$ have the same, rotationally symmetric probability density $\rho_{m_{\p},s}$ satisfying
$\rho_{m_{\p},s}*\rho_{m_{\p},s}=\rho_{m_{\p},2s}$. In view of \cref{correlLpLinfty}, we know that
$L^p(\RR^2)\cap L^{2p}(\RR^2)\subset \mc K_{X_1}$ for every $p\in(1,\infty]$.
If $m_{\p}=0$, the corollary implies $L^p(\RR^2)\subset \mc K_{X_1}$ for all $p\in (2,\infty]$.

\subsubsection{The processes appearing in the creation/annihilation terms}
Next, we discuss the two processes to be substituted into the creation term $F_{t/2}(h)$ 
and the annihilation term $F_{t/2}(h)^*$ in our Feynman--Kac formula. The bound
\begin{align}\label{eq:bounde-somv}
	\|\eul^{-s\omega}v\|\le|g|(\pi/s)^{1/2}
	, \quad s>0,
\end{align}
and Bochner's theorem ensure existence of the following $L^2(\RR^2)$-valued Bochner--Lebesgue integrals,
\begin{align}\label{defUplusminus}
U_t^-[\gamma]&\coloneq\int_0^t\eul^{-s\omega-\ii K\cdot\gamma_s}v\Id s,\quad
U_t^+[\gamma]\coloneq\int_0^t\eul^{-(t-s)\omega-\ii K\cdot\gamma_s}v\Id s,\quad t\ge0.
\end{align}
Here $\gamma:[0,\infty)\to\RR^2$ is assumed to be Borel measurable, $\gamma_s\coloneq\gamma(s)$ as usual,
and we used the notation \eqref{eq:defK}. The regularizing effect of the damping terms $\eul^{-r\omega}$, $r\ge0$,
in \cref{defUplusminus} is illustrated by computing the integrals when $\gamma$ is zero, 
which gives $(1-\eul^{-t\omega})v/\omega\in L^2(\RR^2)$.
With this we define
\begin{align}
U_t^{N,\pm}(x)&\coloneq \sum_{j=1}^N\eul^{-\ii K\cdot x_j}U_t^{\pm}[X_{j,\bullet}],\quad x\in\RR^{2N},\,t\ge0.
\label{def:UtN}
\end{align}
For each choice of the sign the $L^2(\RR^2)$-valued process $(U_t^{N,\pm}(x))_{t\ge0}$ 
is adapted and all its paths are continuous; see \cref{remUpmcont,lem:Upmadapted} for details.

\subsubsection{The complex action}\label{sseccomplexactioninfty}

\noindent
We shall next introduce the analogue of Feynman's complex action in our Feynman--Kac formula, 
which is the sum of three contributions coming from the radiation field.

Let us abbreviate
\begin{align}\label{def:beta}
\beta&\coloneq(\omega+\psi)^{-1}v=g\omega^{-1/2}(\omega+\psi)^{-1}\in L^2(\RR^2).
\end{align}
Then the simplest of the three contributions to the complex action is
\begin{align}\label{defcinfty}
c^N_{t}(x)&\coloneq \sum_{\ell=1}^N
\langle U_{t}^{N,+}(x)|\eul^{-\ii K\cdot X_{\ell,t}^x}\beta\rangle,
\quad x\in\RR^{2N},\,t\ge0.
\end{align}
This expression is in fact real-valued and uniformly bounded in $x$ and $t$ and on $\Omega$; see \cref{rem:bdcN}.
The process $(c_t^N(x))_{t\ge0}$ is manifestly adapted and c\`{a}dl\`{a}g.

Secondly, the radiation field mediates an effective, translation and rotation invariant pair interaction potential 
between the matter particles that we call $w:\RR^2\to\RR$. The potential $w$ comes into the picture as the distributional 
Fourier transform of $v\beta$ which is an element of $L^p(\RR^2)$ for every $p>1$.
More precisely, $w$ denotes a representative of that Fourier transform that is continuous on $\RR^2\setminus\{0\}$
and given by the following formulas:
With $J_0$ denoting the Bessel function of the first kind and order $0$, cf. \cref{app:Bessel}, we set
\begin{align}\label{defvt}
\vt(r,s)&\coloneq \frac{2\pi rJ_0(rs)}{\omega(r)(\omega(r)+\psi(r))},
\quad r,s\ge0,
\end{align}
where we from now on introduce the slightly abusive notation
\begin{align*}
	\omega(r)\coloneq (r^2+m_\bos^2)^{1/2}, \quad \psi(r) \coloneq (r^2+m_\p^2)^{1/2}-m_\p, \quad r\ge 0.
\end{align*}
Then $w$ is given by the following well-defined Lebesgue integrals
\begin{align}\label{def:wy}
w(y)&\coloneq g^2\int_0^{\infty} \vt(r,|y|)\Id r, \quad y\in\RR^2\setminus\{0\},
\end{align}
and we complement this by the (immaterial) convention $w(0)\coloneq0$.
For $N\ge2$, the total effective interaction term appearing in the complex action is
\begin{align}\label{def:wtN}
w_{t}^N(x)&\coloneq \sum_{1\le j<\ell\le N}\int_0^t2w(X^x_{j,s}-X^x_{\ell,s})\Id s,\quad x\in\RR^{2N},\,t\ge0.
\end{align}
For $N=1$, there is no effective interaction and we set $w_t^1(x)\coloneq0$. 
As we shall see in \cref{remdefw}, $w\in L^p(\RR^2)$ for all $p\in[2,\infty)$, whence $w^N(x)$ is a well-defined
path integral process according to the conventions and examples in \cref{sssecKato}.

Most work is needed to obtain control on the third contribution to the complex action, 
namely the real-valued c\`{a}dl\`{a}g square-integrable martingale $(m_t^N(x))_{t\ge0}$ with
\begin{align}\label{def:martingale}
m^N_{t}(x)&\coloneq\int_{(0,t]\times\RR^{2N}}\sum_{j=1}^N\langle U_s^{N,+}(x)|
\eul^{-\ii K\cdot X_{j,s-}^x}(\eul^{-\ii K\cdot z_j}-1)\beta\rangle\Id\wt{N}_X(s,z).
\end{align}
Here  $z_j$ is the $j$'th two-dimensional component of $z$ similarly as in \eqref{defxxj}.
As we shall see in \cref{lemexpmomentm}, the integrand in \cref{def:martingale} belongs to the
space $\scr{H}_2$ introduced in \cref{ssecPPP} and in particular $m_t^N(x)$ is a well-defined
isometric stochastic integral.

Combined, our formula for the complex action reads
\begin{align}\label{def:action}
u_t^N(x)&\coloneq w_t^N(x)-c_t^N(x)+m_t^N(x),\quad x\in\RR^{2N},\, t\ge0.
\end{align}
As it turns out in \cref{cor:pathsactioncont}, the paths of $u^N(x)$ are $\PP$-a.s. continuous, although the last
two members of the right hand side of \cref{def:action} are merely c\`{a}dl\`{a}g.

\subsubsection{Feynman--Kac integrand and formula}

\noindent
We are now in a position to state our Feynman--Kac formula where we abbreviate
\begin{align}\label{def:Wt}
W_t(x)&\coloneq \eul^{u_t^N(x)-\int_0^tV(X^x_s)\Id s}F_{t/2}(-U_t^{N,+}(x))F_{t/2}(-U_t^{N,-}(x))^*.
\end{align}
For every $t>0$ the $\LO(\Fock)$-valued function $W_t(x)$ on $\Omega$ is $\fr{F}_t$-measurable and separably valued
and its operator norm $\|W_t(x)\|$ has moments of all orders; see \cref{rem:Wadapt} and \cref{remmomentbdW}. 

\begin{thm}\label{thm:FK}
Let $\Psi\in L^2(\RR^{2N},\Fock)$ and $t>0$. Then $\eul^{-tH}\Psi$ has a unique continuous representative
which for every $x\in\RR^{2N}$ is given by
\begin{align}\label{eq:FK}
(\eul^{-tH}\Psi)(x)&=\EE\left[W_{t}(x)^*\Psi(X_t^x)\right].
\end{align}
\end{thm}

\begin{proof}
The continuity of the right hand side of \eqref{eq:FK} in $x\in\RR^{2N}$ follows from the 
stronger continuity result \cref{thm:cont}.
Equation \eqref{eq:FK} itself is established at the end of \cref{sec:FK.ren}. 
As a byproduct we shall in fact obtain a new proof of the existence of
the norm resolvent limit \crefnosort{reslimHUV}.
\end{proof}

In the following subsections we present some applications of \cref{thm:FK}.

\subsection{Ergodicity}\label{subsecerg}

Let $\mc{U}:\Fock\to L^2(\mc{g})\coloneq L^2(\mc{Q},\mc{g})$ be a unitary map from the
Fock space to a $\mc{Q}$-space, so that $\mc{g}$ is a probability measure on the set $\mc{Q}$
and all field operators $\vp(h)$ with
\begin{align}\label{defmcRR}
h\in\mc{h}_{\RR}&\coloneq\{f\in L^2(\RR^{2})|\,\ol{f}(k)=f(-k),\,\text{a.e. $k$}\}
\end{align}
become Gaussian random variables on $\mc{Q}$ after conjugation with $\mc{U}$. Confer, e.g., \cite{Parthasarathy.1992} for constructions of $\mc{Q}$-space. 
We denote by the same symbol the constant direct integral of $\mc U$, i.e., $(\mc{U}\Psi)(x)=\mc{U}\Psi(x)$, 
a.e. $x\in\RR^{2N}$, $\Psi\in L^2(\RR^{2N},\Fock)$.
Of course we have a canonical notion of positivity for functions in
$L^2(\RR^{2N},L^2(\mc{g}))=L^2(\RR^{2N}\times\mc{Q},\Id x\otimes\mc{g})$, that we shall employ in the following theorem.
It is a consequence of \cref{thm:FK} and its proof is a word by word  transcription of the proof of \cite[Theorem~8.3]{MatteMoller.2018}. The crucial point is that the operators
$\mc{U}F_{t/2}(h)F_{t/2}(h)^*\mc{U}^*$ with $t>0$ and $h\in \mc{h}_{\RR}$ improve positivity on $L^2(\mc{g})$ as observed in
\cite{Matte.2016,MatteMoller.2018}, and $U^{N,\pm}_{t}(x)$ are in fact $\mc{h}_{\RR}$-valued in view of \eqref{defUplusminus}.

\begin{thm}
Let $t>0$. Then $\mc{U}\eul^{-tH}\mc{U}^*$ improves positivity.
\end{thm}

	Sloan \cite{Sloan.1974} has shown an analogous result for the fiber Hamiltonian at total momentum zero in the
	renormalized translation-invariant relativistic Nelson model.
	In the non-relativistic renormalized Nelson model the ergodicity of fiber Hamiltonians attached to arbitrary
	total momenta -- with respect to a different positive cone in $\Fock$ -- has been shown in
	\cite{Lampart.2020,Miyao.2018}.

A standard conclusion (see, e.g., \cite[Corollaries~8.5\,\&\,8.6]{MatteMoller.2018}) is the following:

\begin{cor}\label{corforspecu}\ 
	\begin{enumerate}
		\item[{\rm(i)}] Let $\Psi\in L^2(\RR^{2N},\Fock)\setminus\{0\}$ be such that $\mc{U}\Psi$ is non-negative. Then
		\begin{align*}
			\inf\sigma(H)&=-\lim_{t\to\infty}\frac{1}{t}\ln\langle\Psi|\eul^{-tH}\Psi\rangle.
		\end{align*}
		\item[{\rm{(ii)}}]
		Let $f\in L^2(\RR^{2N})\setminus\{0\}$ be non-negative. Then
		\begin{align*}
			\inf\sigma(H)&=-\lim_{t\to\infty}\frac{1}{t}\ln\bigg(\int_{\RR^{2N}}f(x)
			\EE[\eul^{u^N_t(x)-\int_0^tV(X^x_s)\Id s}f(X_t^x)]\Id x\bigg).
		\end{align*}
	\end{enumerate}
\end{cor}

\subsection{Two-sided bounds on the minimal energy}\label{ssectwosided}

\noindent
With the aid of \cref{corforspecu} and exponential moment bounds on the complex action, 
we obtain the estimates on the minimal energy presented in the following. We only consider the
translation invariant case. Non-zero external potentials $V$ can be incorporated by using form bounds
relative to $H$ and variational arguments.

\subsubsection{Lower bounds}

\noindent
The first lower bound in the next theorem seems to be reasonably good for small
coupling constants $g$. In this regime it provides a lower bound proportional to $-g^4N^3$,
which is the behavior observed for all values of $g$ and $N$ in the non-relativistic Nelson
model \cite{MatteMoller.2018}. (The minimal energy depends on the choice of the energy counter terms, of course.
In \cite{MatteMoller.2018} the authors work in 3D, set $m_{\p}=1$ and use \eqref{def:EUVren} 
with $|k|^2/2$ put in place of $\psi(k)$.)
The second and third bounds in the next theorem show the correct behavior for large $g$ and/or $N$
as we shall see in \cref{ssecasymp}.

\begin{thm}\label{thm:lowbound}
Consider the translation invariant case $V=0$.
There exist constants $b,b'>0$ such that, for all coupling constants~$g$ and particle numbers~$N$,
\begin{align*}
\inf\sigma(H)&\ge -\frac{bg^4N^3}{m_{\bos}}\cdot\eul^{8\pi g^2N/m_{\bos}}
-\frac{b'g^4N^2(N-1)}{m_{\bos}}\bigg(1+\frac{m_{\p}}{m_{\bos}}\bigg).
\end{align*}
Furthermore, there exist constants $c,c'>0$ such that
\begin{align*}
\inf\sigma(H)&\ge -\bigg\{\pi g^2(2N^2-N)\ln\bigg(\frac{g^2N}{m_{\bos}}\bigg)\bigg\}-cg^2N^2-c'(N-1)m_{\p},
\ \ \text{if $g^2N\ge m_{\bos}$.}
\end{align*}
For $m_{\p}>0$, the term in curly brackets $\{\cdots\}$ can be replaced by
\begin{align*}
2\pi g^2N(N-1)\ln\bigg(\frac{g^2N}{m_{\bos}}\bigg)
+\pi g^2N\ln(1\vee(g^2N))+\frac{\pi g^2N}{m_{\p}}[(g^2N)\wedge1].
\end{align*}
\end{thm}

\begin{proof}
To prove these bounds we choose a bounded, integrable, non-negative and non-zero $f:\RR^{2N}\to\RR$ in
\cref{corforspecu}(ii) and estimate $f(X_t^x)\le\|f\|_\infty$. After that we apply
the exponential moment bounds on the complex phase
in \cref{thmexpmonentu1,thmexpmonentu2} with $p=1$. Since the parameter $\alpha>1$ 
appearing in these theorems is arbitrary, we obtain the asserted lower bounds in this way.
\end{proof}

\subsubsection{Upper bound}

\noindent
We complement the lower bounds from the previous theorem with an upper bound, which follows from a combination
of \cref{corforspecu} and our exponential moment bounds on the complex action with a more well-known
trial state argument. 

\begin{thm}\label{thm:upbound}
Assume that $V=0$
and let $\theta\in(0,1)$. Then there exists $C(\theta,m_{\p},1\vee m_{\bos})\in(0,\infty)$, depending solely on the
quantities displayed in its argument, as well as a universal constant $c>0$ such that, whenever $g^2N>2m_{\bos}$,
\begin{align*}
\inf\sigma(H)
&\le c(N-1)m_{\p}+C(\theta,m_{\p},1\vee m_{\bos})g^2N^2
\\
&\quad-2\pi \theta g^2N(N-1)\ln\bigg(\frac{g^2N}{m_{\bos}}\bigg)-\pi \theta g^2N\ln(1\vee (g^2N))
\\
&\quad-\chi_{\{0\}}(m_{\p})\pi \theta g^2N\ln\bigg(\frac{1\wedge(g^2N)}{m_{\bos}}\bigg).
\end{align*}
\end{thm}

\begin{proof}
The proof of this theorem can be found at the end of \cref{sec:upbound}.
\end{proof}

\subsection{Asymptotics of the minimal energy}\label{ssecasymp}

Combined the upper and lower bounds in \cref{ssectwosided} reveal the asymptotics of the minimal energy in the three regimes where the particle number $N$ and the coupling constant $g$ are large and the boson mass is small.

For convenience we do typically not display the dependence of $H$ and related objects on
the parameters $N$, $m_{\p}$, $m_{\bos}$ and $g$ in our notation. An exception will be
the next theorem and the subsequent remark where the minimal energy of $H$ in the translation invariant case 
is denoted by
\begin{align*}
\mc{E}(g,N,m_{\p},m_{\bos})&\coloneq\inf\sigma(H)\quad\text{when $V=0$.}
\end{align*}

\begin{thm}\label{thm:asymp}
The following relations hold for $V=0$,
\begin{align}
\lim_{N\to\infty}\frac{\mc{E}(g,N,m_{\p},m_{\bos})}{N^2\ln(N)}&=-2\pi g^2,
\\
\lim_{|g|\to\infty}\frac{\mc{E}(g,N,m_{\p},m_{\bos})}{g^2\ln(g^2)}&=-\pi(2N^2-N),
\\\label{asympmb}
\lim_{m_{\bos}\downarrow0}\frac{\mc{E}(g,N,m_{\p},m_{\bos})}{\ln(m_{\bos})}&=2\pi g^2N(N-1),\quad m_{\p}>0,
\\
\lim_{m_{\bos}\downarrow0}\frac{\mc{E}(g,N,0,m_{\bos})}{\ln(m_{\bos})}&=\pi g^2(2N^2-N).
\end{align}
\end{thm}

\begin{proof}
All limit relations follow directly from \cref{thm:lowbound,thm:upbound}, if we take into account that the
parameter $\theta\in(0,1)$ appearing in \cref{thm:upbound} can be chosen arbitrarily close to $1$.
\end{proof}

\begin{rem}
In the case $N=1$ with $m_{\p}>0$, where \eqref{asympmb} does not reveal any leading asymptotic
behavior as $m_{\bos}\downarrow0$, the last statement in \cref{thm:lowbound} implies the uniform lower bound
\begin{align}
\mc{E}(g,1,m_{\p},m_{\bos})\ge -\pi g^2\ln(g^2\vee1)-\frac{\pi g^2}{m_{\p}}[g^2\wedge1]-cg^2.
\end{align}
That is, singular contributions to the minimal energy coming from $m_{\bos}$ have been compensated for
by the renormalization energies $E_{\UV}^{\ren}$ in this case.
\end{rem}


\section{Proof of the Feynman--Kac formula}\label{sec:FK}

\noindent
In this section we present the proof of our Feynman--Kac formula. In doing so we employ some results of 
later sections as starting points. We proceed in this way so that the reader can quickly grasp the general proof
strategy leading to the Feynman--Kac formula and see which ingredients are needed.
The latter are formulated precisely in \cref{ssecdefprelFK} after the necessary notation has been set up.
We shall in particular introduce a Feynman-Kac semigroup
$(T_{\UV,t})_{t\ge 0}$ for every $\UV\in(0,\infty]$. When $\UV<\infty$, the interaction terms in this semigroup are
ultraviolet cutoff at $\UV$. Proving the Feynman--Kac formula for the ultraviolet regularized Hamiltonian $H_{\UV}$
then amounts to identifying $H_{\UV}$ as the generator of $(T_{\UV,t})_{t\ge 0}$.
This is done in \cref{sec:FK-UV}, essentially by combining an integral equation for the Fock space operator 
part of the ultraviolet cutoff Feynman--Kac integrands, stated as a substitution rule in \cref{ssecdefprelFK}, 
with an It\^{o} formula discussed in \cref{ssecItoformula}.
In the final \cref{sec:FK.ren} we establish the Feynman--Kac formula for $H$, exploiting that
$T_{\UV,t}\to T_{\infty,t}$, $\UV\to\infty$, in operator norm for every $t\ge0$
as shown in \cref{ssecconvTUV}.

\subsection{Definitions and main ingredients for the proof}\label{ssecdefprelFK}

\noindent
We start by defining the complex action $\tilde{u}_{\UV,t}^N(x)$ for finite $\UV$ and the 
Fock space operator part $W_{\UV,t}(x)$ of the ultraviolet cutoff Feynman--Kac integrands.
The formula \eqref{defuUV0} for $\tilde{u}_{\UV,t}^N(x)$ is a direct analogue of Feynman's expression for the complex action
in the polaron model. To see the analogy, one has to pull the integrations coming from \eqref{defUplusminus}
out of the scalar product in \eqref{defuUV0} and evaluate the so-obtained new scalar products, which
amounts to computing the Fourier transform of $v_{\UV}^2$. As opposed to the polaron model, we cannot
simply drop the cutoff $\UV$ in \eqref{defuUV0} after these manipulations. 
Before doing so, and to derive useful, $\UV$-uniform bounds
on $\tilde{u}_{\UV,t}^N(x)$, we have to re-write it by means of It\^{o}'s formula. This procedure
results in the definition of the limiting complex action $u_t^N(x)$ we saw in \cref{sseccomplexactioninfty};
see \cref{lem:actionito} below.

\begin{defn}\label{defn:complexphase0UV}
Let $\UV\in(0,\infty)$ and $x\in\RR^{2N}$. Then we define
\begin{align}\label{eq:defUpmUV}
U_{\UV,t}^{N,\pm}(x)&\coloneq \chi_{B_{\UV}}U^{N,\pm}_{t}(x),
\\\label{defuUV0}
\tilde{u}_{\UV,t}^N(x)&\coloneq\sum_{j=1}^N
\int_0^t\langle U^{N,+}_{\UV,s}(x)|\eul^{-\ii K\cdot X^x_{j,s}}v_{\UV}\rangle\Id s-tNE_{\UV}^{\ren},
\end{align}
for all $t\ge0$.
Furthermore, we set $W_{\UV,0}(x)\coloneq\id_{\Fock}$ and, for $t>0$,
\begin{align}\label{defWUVtx0}
W_{\UV,t}(x)&\coloneq
\eul^{\tilde{u}^N_{\UV,t}(x)-\int_0^tV(X^x_s)\Id s}F_{t/2}(-U^{N,+}_{\UV,t}(x))
F_{t/2}(-U^{N,-}_{\UV,t}(x))^*.
\end{align}
\end{defn}

To unify notation we also set
\begin{align*}
W_{\infty,t}(x)\coloneq W_t(x),\quad t\ge0,\,x\in\RR^{2N},
\end{align*}
where $W_t(x)$ is defined in \cref{def:Wt}. We shall shortly make use of the following observations:

\begin{rem}\label{rem:Wadapt}
Let $x\in\RR^{2N}$.
We know from \cite[\textsection17]{GueneysuMatteMoller.2017} that 
$(0,\infty)\times L^2(\RR^2)\ni (s,h)\mapsto F_s(h)\in\LO(\Fock)$ is continuous. 
Together with \cref{defuUV0}, \cref{lem:Upmadapted} and the way we defined
the limiting complex action $u_t^N(x)$ this implies that 
$W_{\UV,t}(x):\Omega\to\LO(\Fock)$ is $\fr{F}_t$-measurable and separably valued for all
$t\ge0$ and $\UV\in(0,\infty]$. In conjunction with \cref{defuUV0,remUpmcont} it
further follows that all paths of $(W_{\UV,t}(x))_{t\ge0}$ with $\UV\in(0,\infty)$ are continuous on
the open half-axis $(0,\infty)$. (In the case $\UV=\infty$, \cref{cor:pathsactioncont} will
imply that $(0,\infty)\ni t\mapsto W_{\infty,t}(x)\in\LO(\Fock)$ is continuous $\PP$-a.s.)
\end{rem}

In the next proposition we state the substitution rule involving
$W_{\UV,t}(x)$ with finite $\UV$ alluded to above. It can be combined with the It\^{o} formula \cref{ItoAfX} below.
For every $\UV\in(0,\infty)$, we abbreviate
\begin{align}\label{defhUVx}
h_{\UV}(x)&\coloneq \Id\Gamma(\omega)+\sum_{j=1}^N\vp(\eul^{-\ii K\cdot x_j}v_\UV)+V(x)+NE_{\UV}^{\ren},
\quad x\in\RR^{2N}.
\end{align}
Thus, each $h_{\UV}(x)$ is a selfadjoint operator in $\Fock$ with domain $\dom(\Id\Gamma(\omega))$ 
containing all terms in $H_{\UV}$ apart from the kinetic energy of the matter particles.

\begin{prop}\label{propODEWUV}
Let $\UV\in(0,\infty)$, $x\in\RR^{2N}$, $\phi_1\in\dom(\Id\Gamma(\omega))$, $\phi_2\in\Fock$ and abbreviate
\begin{align*}
A_{t}&\coloneq\langle\phi_1|W_{\UV,t}(x)\phi_2\rangle_{\Fock},\quad t\ge0.
\end{align*}
Assume that $V$ is bounded.
Then $A$ is adapted and all its paths are absolutely continuous on every compact subinterval of $[0,\infty)$.
Furthermore, $A_0=\langle\phi_1|\phi_2\rangle_{\Fock}$ and, for every 
complex-valued predictable process $(Z_s)_{s\ge0}$ with locally bounded paths,
\begin{align}\label{eq:substWUV}
\int_0^tZ_s\Id A_s&=\int_0^tZ_s\langle h_{\UV}(X_s^x)\phi_1|W_{\UV,s}(x)\phi_2\rangle_{\Fock}\Id s,\quad t\ge0.
\end{align}
\end{prop}

\begin{proof}
Adaptedness of $A$ follows from \cref{rem:Wadapt}. The remaining statements follow directly from \cref{lemODEW1} below.
\end{proof}

In view of \cref{rem:Wadapt}, our next definition is meaningful:
\begin{defn}\label{def:semigroup}
Let $\UV\in(0,\infty]$ and $t\ge0$. Assume that $\Psi:\RR^{2N}\to\Fock$ is measurable 
and such that $\|W_{\UV,t}(x)^*\Psi(X_t^x)\|_{\Fock}$ is $\PP$-integrable for all $x\in\RR^{2N}$.
Then, we define $T_{\UV,t}\Psi:\RR^{2N}\to\Fock$ by
\begin{align*}
(T_{\UV,t}\Psi)(x)&\coloneq \EE\big[W_{\UV,t}(x)^*\Psi(X_t^x)\big],\quad x\in\RR^{2N}.
\end{align*}
\end{defn}

Since $\PP(X_t^x\in\scr{N})=0$ for all zero sets $\scr{N}\in\fr{B}(\RR^{2N})$ and $t>0$,
the maps $T_{\UV,t}$, $t>0$, are also well-defined on equivalence classes of functions. 
We will often and tacitly make use of this. The symbol $T_{\UV,0}$ always stands for the
identity mapping, whether acting on measurable functions or their equivalence classes.

In \cref{sec:FKInt} we shall study the maps $T_{\UV,t}$ when applied to functions in $L^p(\RR^{2N},\Fock)$
with $p\in(1,\infty]$. We summarize all results of that section needed to derive the Feynman--Kac formula in the 
next proposition:

\begin{prop}\label{prop:TUVonL2}
For all $t>0$, $x\in\RR^{2N}$ and $\Psi\in L^2(\RR^{2N},\Fock)$, 
the random variable $\|W_{\UV,t}(x)^*\Psi(X_t^x)\|_{\Fock}$ is $\PP$-integrable and
$T_{\UV,t}\Psi$ is square-integrable.
Seen as a family of operators on $L^2(\RR^{2N},\Fock)$,
$(T_{\UV,t})_{t\ge0}$ is a strongly continuous semigroup of bounded operators for every $\UV\in(0,\infty]$.
Furthermore, there exists $c\in(0,\infty)$ such that
$\|T_{\UV,t}\|\le \eul^{c(1+t)}$ for all $t\ge0$ and $\UV\in(0,\infty]$.
Finally, $T_{\UV,t}\to T_{\infty,t}$ as $\UV\to\infty$ in operator norm for every $t\ge0$.
\end{prop}

\begin{proof}
	This proposition summarizes assertions proven in \cref{sec:FKInt}. More precisely,
	the first integrability statement can be found in \cref{prop:LpLq}(i),
	the stated $L^2$-properties are a special case of \cref{thm:semigroup,lem:strongcont},
	and the limit $\UV\to\infty$ is treated in \cref{prop:Tconv}.
\end{proof}


\subsection{A special case of It\^{o}'s formula}\label{ssecItoformula}

\noindent
To streamline later proofs and clarify where and how the operator $\psi_X(-\ii\nabla)$
emerges in our computations, we note the following special case of the standard It\^{o} formula for $X$.
It will be applied in \cref{sec:FK-UV,lem:actionito}.

\begin{lem}
Assume that $f\in C^2(\RR^{2N})$ is bounded with bounded first and second order derivatives. 
Let $(A_t)_{t\ge0}$ be an adapted c\`{a}dl\`{a}g real-valued process all whose paths have locally finite variation. 
Assume that $\sup_{s\in[0,t]}|A_s|\in L^2(\PP)$ for all $t\ge0$. Finally, let $x\in\RR^{2N}$. Then, $\PP$-a.s.,
\begin{align}\nonumber
A_tf(X_t^x)&=A_0f(x)+\int_0^tf(X_{s-}^x)\Id A_s-\int_0^tA_{s}(\psi_X(-\ii\nabla)f)(X_s^x)\Id s
\\\label{ItoAfX}
&\quad+\int_{(0,t]\times\RR^{2N}}A_{s-}(f(X_{s-}^x+z)-f(X_{s-}^x))\Id \wt{N}_X(s,z),\quad t\ge0.
\end{align}
The process comprised of the stochastic integrals in the second line is a martingale.
\end{lem}

\begin{proof}
Since $f$ is bounded with a bounded derivative,
\begin{align*}
\sup_{s\in[0,t]}|A_{s-}(f(X_{s-}^x+z)-f(X_{s-}^x))|\le (|z|\wedge1)(\|\nabla f\|_\infty\vee2\|f\|_\infty) \sup_{s\in[0,t]}|A_{s}|, 
\end{align*}
for all $z\in\RR^{2N}$ and $t\ge0$.
Thus, $(s,z,\gamma)\mapsto A_{s-}(\gamma)(f(X_{s-}^x(\gamma)+z)-f(X_{s-}^x(\gamma)))$ belongs to $\scr{H}_2$
in view of \cref{intcondnu,eq:nuX}.
In particular, recalling the discussion before \cref{ItoIso}, its $\Id\wt{N}_X$-integral is a martingale. Furthermore,
since $(a,z)\mapsto af(z)$ is in $C^2(\RR^{1+2N})$, the standard textbook version
of It\^{o}'s formula (see, e.g.,  \cite[Theorem~4.4.7]{Applebaum.2009}
or \cite[\textsection2~Theorem~5.1]{IkedaWatanabe.1981}) in conjunction with the
L\'{e}vy--It\^{o} decomposition \eqref{LevyItoX} of $X$ directly implies
\begin{align}\nonumber
A_tf(X_t^x)&=A_0f(x)+\int_0^tf(X_{s-}^x)\Id A_s
\\\nonumber
&\quad+\int_{(0,t]\times(\RR^{2N}\setminus B_1')}A_{s-}(f(X_{s-}^x+z)-f(X_{s-}^x)) \Id N_X(s,z)
\\\nonumber
&\quad+\int_{(0,t]\times B_1'}A_{s-}(f(X_{s-}^x+z)-f(X_{s-}^x))\Id \wt{N}_X(s,z)
\\\label{Itoallg}
&\quad+\int_0^t\int_{B_1'}A_s\big(f(X_{s}^x+z)-f(X_{s}^x)-z\cdot\nabla f(X_s^x)\big)\Id\nu_X(z)\,\Id s,
\end{align}
for all $t\ge0$, $\PP$-a.s.; recall the definition \cref{2NDball} of $B_1'$.
Since $f$ is bounded and $\sup_{s\in[0,t]}|A_s|\in L^1(\PP)$, by assumption, we further see that the map
\begin{align*}
\text{$(s,z,\gamma)\longmapsto
\chi_{\RR^{2N}\setminus B_1'}(z)A_{s-}(\gamma)(f(X_{s-}^x(\gamma)+z)-f(X_{s-}^x(\gamma)))$
belongs to $\scr{H}_1$.}
\end{align*} 
Hence, we can subtract the double integral
\begin{align}\label{doubleint}
\int_0^t\int_{\RR^{2N}\setminus B_1'}A_{s-}(f(X_{s-}^x+z)-f(X_{s-}^x))\Id\nu_X(z)\,\Id s
\end{align} 
in the second line of \cref{Itoallg} and add it to the last one at the same time.
Taking \eqref{wtNXintH1} into account, we then find a $\Id\wt{N}_X$-integral over
$(0,t]\times(\RR^{2N}\setminus B_1')$ in the second line of \cref{Itoallg}, that we can combine with the one over 
$(0,t]\times B_1'$ in the third line. Further, the c\`{a}dl\`{a}g paths of $A$ and $X$ have at most countably many 
discontinuities. Thus, $A_{s-}$ and $X_{s-}^x$ can be replaced by 
$A_s$ and $X_s^x$, respectively, in \cref{doubleint}
without changing the value of the double integral. We thus $\PP$-a.s. arrive at
\begin{align*}
A_tf(X_t^x)&=A_0f(x)+\int_0^tf(X_{s-}^x)\Id A_s
\\
&\quad+\int_{(0,t]\times\RR^{2N}}A_{s-}(f(X_{s-}^x+z)-f(X_{s-}^x))\Id \wt{N}_X(s,z)
\\
&\quad+\int_0^t\int_{\RR^{2N}}A_s\big(f(X_{s}^x+z)-f(X_{s}^x)
-\chi_{B_1'}(z)z\cdot\nabla_zf(X_s^x)\big)\Id\nu_X(z)\,\Id s,
\end{align*}
for all $t\ge0$. On account of \cref{def:psiXnabla} this proves \cref{ItoAfX}.
\end{proof}


\subsection{Feynman--Kac formula with ultraviolet cutoff}\label{sec:FK-UV}

\noindent
We can now prove the Feynman--Kac formula in presence of a finite cutoff.
\begin{thm}\label{thm:FK-UV}
Let $\UV\in(0,\infty)$ and $t\ge0$. Consider $T_{\UV,t}$ as an operator on $L^2(\RR^{2N},\Fock)$.
Then $\eul^{-tH_\UV}=T_{\UV,t}$ and in particular $T_{\UV,t}$ is selfadjoint.
\end{thm}

\begin{proof}
To start with we assume that $V$ is bounded. By \cref{prop:TUVonL2} and the Hille--Yosida theorem,
the semigroup $(T_{\UV,t})_{t\ge0}$ has a closed generator, call it $G$, whose spectrum
is contained in $\{\zeta\in\CC|\,\Re[\zeta]\ge -c\}$ for some $c\ge0$.
We shall show that $H_{\UV}=G$, which is equivalent to $T_{\UV,t}=\eul^{-tH_{\UV}}$, $t\ge0$.

Let $\phi_1\in\dom(\Id\Gamma(\omega))$ and $f\in C_0^\infty(\RR^{2N})$. 
Since $V$ is bounded, we can infer from \cref{def:HUV,psiXFourier,defhUVx} 
that (a representative of) $H_{\UV}f\phi_1$ can be written as
\begin{align}\label{forHUVfphi}
(H_{\UV}f\phi_1)(x)&\coloneq (\psi_X(-\ii\nabla)f)(x)\phi_1+ f(x)h_{\UV}(x)\phi_1,\quad x\in\RR^{2N}.
\end{align}
Let also $\phi_2\in\Fock$ and fix $x\in\RR^{2N}$ for the moment. 
Thanks to \cref{lemODEW1} we know that all paths of the process given by
$A_t\coloneq\langle\phi_1|W_{\UV,t}(x)\phi_2\rangle_{\Fock}$
are absolutely continuous on every compact subinterval of $[0,\infty)$.
Employing \cref{normFt,tnormU} (with $\ve=0$) and a trivial
estimation of $\tilde{u}_{\UV,t}^N(x)$ (see \cref{trivialbdu}) we find some $b\in(0,\infty)$
(depending on $\UV$ and $\|V\|_\infty$ besides other model parameters) such that
$|A_t|\le \eul^{b(1+t)}\|\phi_1\|_{\Fock}\|\phi_2\|_{\Fock}$, $t\ge0$. 
Therefore, the It\^{o} formula \cref{ItoAfX} is available, 
and in conjunction with the substitution rule \cref{eq:substWUV} it $\PP$-a.s. yields
\begin{align}\nonumber
&\langle\phi_1|W_{\UV,t}(x)\phi_2\rangle_{\Fock}\ol{f}(X_t^x)-\langle\phi_1|\phi_2\rangle_{\Fock} \ol{f}(x)
\\\nonumber
&=-\int_0^t\langle\phi_1|W_{\UV,s}(x)\phi_2\rangle_{\Fock}(\psi_X(-\ii\nabla)\ol{f})(X_s^x)\Id s
\\\nonumber
&\quad-\int_0^t\ol{f}(X_{s-}^x)\langle h_{\UV}(X_s^x)\phi_1|W_{\UV,s}(x)\phi_2\rangle_{\Fock}\Id s
\\\nonumber
&\quad-\int_{(0,t]\times\RR^{2N}}\langle 
\phi_1| W_{\UV,s}(x)\phi_2\rangle_{\Fock}(\ol{f}(X_{s-}^x+z)-\ol{f}(X_{s-}^x))\Id \wt{N}_X(s,z)
\\\label{eq:ItoAWUV}
&=-\int_0^t\langle W_{\UV,s}(x)^*(H_{\UV}f\phi_1)(X_s^x)|\phi_2\rangle_{\Fock}\Id s
-\text{[a martingale starting at $0$]}_t,
\end{align}
for all $t\ge0$. Here we used \cref{forHUVfphi} in the second step
as well as that the term in the fourth line defines a martingale, by \cref{lemODEW1}, see the discussion before \cref{ItoIso}.
Upon taking expectations in \cref{eq:ItoAWUV}, the martingale drops out and we find
\begin{align*}
\langle (T_{\UV,t}f\phi_1)(x)|\phi_2\rangle_{\Fock} -\langle f(x)\phi_1|\phi_2\rangle_{\Fock}
&=-\int_0^t\langle (T_{\UV,s}H_{\UV}f\phi_1)(x)|\phi_2\rangle_{\Fock}\Id s,\quad t\ge0.
\end{align*}
After choosing $\phi_2=\Phi(x)$ for every $x$ and some $\Phi\in L^2(\RR^{2N},\Fock)$
and integrating with respect to $x$, 
we can use Fubini's theorem to interchange the order of the $\Id x$- and $\Id s$-integration.
Since $s\mapsto T_{\UV,s}H_{\UV}f\phi_1\in L^2(\RR^{2N},\Fock)$ is continuous and in particular
Bochner-Lebesgue integrable over $[0,t]$, this reveals that
\begin{align}\label{HgenT1}
T_{\UV,t}f\phi_1-f\phi_1&=-\int_0^t T_{\UV,s}H_\UV f\phi_1\Id s,\quad t\ge0.
\end{align}
Again by strong continuity of the semigroup, $T_{\UV,s}H_\UV f\phi_1\to H_\UV f\phi_1$, $s\downarrow0$,
 in $L^2(\RR^{2N},\Fock)$, whence \eqref{HgenT1} implies
 \begin{align*}
\lim_{t\downarrow0} \frac{1}{t}(T_{\UV,t}f\phi_1-f\phi_1)=-H_\UV f\phi_1\quad\text{in $L^2(\RR^{2N},\Fock)$.}
 \end{align*}
This is equivalent to saying that $f\phi_1$ lies in the domain of $G$ with $Gf\phi_1=H_{\UV}f\phi_1$. 
Again since $\UV$ is finite and $V$ bounded, we know that the linear hull of all vectors of the form
$f\phi_1$ with $f$ and $\phi_1$ as above is a core for $H_{\UV}$; see, e.g., \cref{lemesaNC} below. Since $G$ is closed,
we conclude that $H_\UV\subset G$. As noted above, there is a uniform lower bound on the real parts
of all points in the spectrum of $G$. Since $H_{\UV}$ is
lower semibounded, we clearly find some $\zeta\in\CC$ belonging to the resolvent sets of both $G$ and $H_{\UV}$.
Then $H_\UV\subset G$ and the second resolvent identity imply $(G-\zeta)^{-1}=(H_{\UV}-\zeta)^{-1}$, thus $G=H_{\UV}$.

To generalize this result we follow a standard approximation procedure: 
Let $V$ be Kato decomposable and bounded from below. 
Set $V_n \coloneq n\wedge V$ and denote by $T^{n}_{\UV,t}$, $H^{n}_\UV$ and $\mathfrak{h}_{\UV}^n$ 
the corresponding semigroup members, Hamiltonian and quadratic form 
for $n\in\NN$. Let $\Phi\in L^2(\RR^{2N},\Fock)$. By the above result, $\eul^{-tH_{\UV}^n}\Phi=T^n_{\UV,t}\Phi$ 
for all $n\in\NN$.
By dominated convergence and \cref{prop:TUVonL2}, $(T^{n}_{\UV,t}\Phi)(x)$ converges to $(T_{\UV,t}\Phi)(x)$ 
as $n\to\infty$ for all $x\in\RR^{2N}$. Further, $\mathfrak{h}_{\UV}^n\uparrow\mathfrak{h}_{\UV}$ 
on $\dom(\mathfrak{h}_{\UV})$ as $n\to\infty$, which by a convergence theorem for an increasing sequence of
quadratic forms entails strong resolvent convergence of $H^n_{\UV}$ to $H_{\UV}$; see, e.g.,
\cite[Theorems~VIII.20(b) and S.14]{ReedSimon.1980}.
Thus, $\eul^{-tH^{n_\ell}_\UV}\Phi\to \eul^{-tH_\UV}\Phi$, $\ell\to\infty$, a.e. for some strictly increasing sequence 
$(n_\ell)_{\ell\in\NN}$ in $\NN$. This proves the theorem when $V$ is bounded from below.

For general Kato decomposable $V$, we consider $(-n)\vee V$ and argue analogously,
employing a convergence theorem for a decreasing sequence of quadratic forms; see, e.g.,
\cite[Theorems~VIII.20(b) and S.16]{ReedSimon.1980}. (When applying Theorem~S.16 in \cite{ReedSimon.1980}
we take into account that the quadratic forms corresponding to $V$ and to every $(-n)\vee V$ with $n\in\NN$ 
all have the same domain, namely $\fdom(H_{\UV})$ described prior to \cref{def:HUV}.)
\end{proof}


\subsection{Feynman--Kac formula for the renormalized Hamiltonian}\label{sec:FK.ren}

\noindent
After the next corollary we can finally prove the Feynman--Kac formula for $H$.

\begin{cor}\label{cor:semigroup-ren}
When its members are considered as operators on $L^2(\RR^{2N},\Fock)$, 
the semigroup $(T_{\infty,t})_{t\ge0}$ is strongly continuous and every $T_{\infty,t}$ is bounded and selfadjoint 
with $\|T_{\infty,t}\|\le \eul^{c(1+t)}$, $t\ge0$, for some $c>0$.
Furthermore, $H_{\UV}$ converges in norm resolvent sense to the selfadjoint, lower semibounded generator
of $(T_{\infty,t})_{t\ge0}$, as $\UV\to\infty$.	
\end{cor}

\begin{proof}
Apart from the claimed selfadjointness, the first statement just repeats parts of \cref{prop:TUVonL2}.
Given $t\ge0$, selfadjointness of $T_{\infty,t}$ follows, however, from the operator norm convergence
$T_{\UV,t}\to T_{\infty,t}$, $\UV\to\infty$, asserted in \cref{prop:TUVonL2} and the selfadjointness of 
each $T_{\UV,t}=\eul^{-tH_{\UV}}$, $\UV\in(0,\infty)$, established in \cref{thm:FK-UV}.
This implies existence of a selfadjoint, lower semibounded generator of $(T_{\infty,t})_{t\ge0}$.
The equivalence of the operator norm convergence of all semigroup members $T_{\UV,t}\to T_{\infty,t}$
and  norm resolvent convergence of the selfadjoint generators is well-known
and follows directly from the representation of the 
resolvents of the generator as Bochner-Lebesgue integrals over semigroup elements
$(H_{\UV}+\lambda)^{-1} = \int_0^\infty \eul^{-t(H_{\UV}+\lambda)}\Id t$, which holds for any $\lambda< c$ with $c$ as in the first statement.
\end{proof}
\Cref{cor:semigroup-ren} especially provides the
\begin{proof}[\textbf{Proof of the Feynman--Kac formula \cref{eq:FK}}]
\cref{cor:semigroup-ren} extends the known existence results \cite{Sloan.1974,Schmidt.2019}
for the norm resolvent limit $H\coloneq\lim_{\UV\to\infty} H_{\UV}$ to general
Kato decomposable $V$. Furthermore, the corollary ensures that $\eul^{-tH}=T_{\infty,t}$, $t\ge0$, 
which is equivalent to saying that \cref{eq:FK} holds for a.e. $x\in\RR^{2N}$, 
given any $t>0$ and $\Psi\in L^2(\RR^{2N},\Fock)$. 
\end{proof}


\section{Integral equation for UV cutoff Feynman--Kac integrands}\label{sec:WUVinteq}

\noindent
In this section we only consider finite $\UV$. Our aim is to derive a pathwise integral equation for matrix elements of 
$W_{\UV,t}(x)$ as defined in \cref{defWUVtx0}, that will complete the proof of \cref{propODEWUV}.
To this end we need the integral equation for $U^{N,+}_{\UV,t}(x)$ observed in the next lemma;
recall the definitions \cref{defUplusminus,def:UtN,eq:defUpmUV}.

\begin{lem}\label{lemIBPUplus}
Let $\UV\in(0,\infty)$ and $x\in\RR^{2N}$. Abbreviate
\begin{align}\label{def:vUVs}
v_{\UV,s}^N(x)\coloneq\sum_{j=1}^N\eul^{-\ii K\cdot X^x_{j,s}}v_\UV,\quad s\ge0.
\end{align}
Then all paths of $(U^{N,+}_{\UV,t}(x))_{t\ge0}$ are absolutely continuous on every compact subinterval of $[0,\infty)$ and
\begin{align}\label{IBPUplus}
U_{\UV,t}^{N,+}(x)
&=\int_0^t(v_{\UV,s}^N(x)-\omega U_{\UV,s}^{N,+}(x))\Id s,\quad t\ge0.
\end{align}
\end{lem}

\begin{proof}
Let $\gamma:[0,\infty)\to\RR^2$ be c\`{a}dl\`{a}g and $g\in L^2(\RR^2)$. 
In view of \cref{def:UtN} it suffices to show the second relation in
\begin{align}\label{inteqUgamma}
\theta(t)&\coloneq\langle g|\chi_{B_{\UV}}U_{t}^+[\gamma]\rangle
=\int_0^t\langle g|\eul^{-\ii K\cdot\gamma_s}v_{\UV}-\omega \chi_{B_{\UV}} U_{s}^{+}[\gamma]\rangle\Id s,
\end{align}
for all $t\ge0$. Let $n\in\NN$ and $0\le t_0<t_1<\dots<t_n\le t$. Employing \cref{defUplusminus} and 
straightforward estimations, we then find
\begin{align*}
\sum_{i=1}^n|\theta(t_i)-\theta(t_{i-1})|&\le \|g\|\|v_{\UV}\|(1+t\|\omega\chi_{B_\UV}\|_\infty)(t_n-t_0).
\end{align*}
Hence, $\theta$ is absolutely continuous on every compact subinterval of $[0,\infty)$.
Furthermore, $\theta'(t)=\langle g|\eul^{-\ii K\cdot\gamma_t}v_{\UV}-\omega \chi_{B_{\UV}} U_{t}^{+}[\gamma]\rangle$,
whenever $\gamma$ is continuous at $t>0$. Since the set of discontinuities of $\gamma$ is countable, the second 
relation in \cref{inteqUgamma} now follows from the fundamental theorem of calculus for the Lebesgue integral.
\end{proof}

When dealing with $W_{\UV,t}(x)$,
it will be convenient to do the computations for exponential vectors, i.e., Fock space vectors of the form
\begin{align}\label{def:expvec}
	\expv{h} \coloneq (1,h,\ldots,(n!)^{-1/2}h^{\otimes n},\ldots\:)\in \Fock, \quad h\in L^2(\RR^2),
\end{align}
where $h^{\otimes n}(k_1,\ldots,k_n) \coloneq h(k_1)\cdots h(k_n)$, 
a.e. $(k_1,\ldots,k_n)\in\RR^{2N}$ with $k_j\in\RR^2$.
As explained in \cite[Remark~5.2]{MatteMoller.2018},
\begin{align}\label{eq:Fstarexpvec}
F_{t}(g)\expv{h}&=\expv{g+\eul^{-t\omega}h},\quad
F_{t}(g)^*\expv{h}=\eul^{\langle g|h\rangle}\expv{\eul^{-t\omega}h},
\end{align}
for all $t>0$ and $g,h\in L^2(\RR^2)$. Hence,
the action of $W_{\UV,t}(x)$ with $\UV\in(0,\infty)$ on an exponential vector reads
\begin{align}\label{Wexpvec}
W_{\UV,t}(x)\expv{h}
&=\eul^{\tilde{u}_{\UV,t}^N(x)-\int_0^tV(X^x_s)\Id s-\langle U^{N,-}_{\UV,t}(x)|h\rangle}
\expv{\eul^{-t\omega}h-U_{\UV,t}^{N,+}(x)},
\end{align}
for all $t\ge0$ and $x\in\RR^{2N}$. 

\begin{lem}\label{lemODEW1}
Assume that $V$ is bounded and 
let $\UV\in(0,\infty)$, $x\in\RR^{2N}$, $\phi_1\in\dom(\Id\Gamma(\omega))$ and $\phi_2\in\Fock$. Then
\begin{align}\label{eq:ODEW1}
		\langle\phi_1|W_{\UV,t}(x)\phi_2\rangle_{\Fock}-\langle\phi_1|\phi_2\rangle_{\Fock}
		&=-\int_0^t\langle h_{\UV}(X^x_s)\phi_1|
		W_{\UV,s}(x)\phi_2\rangle_{\Fock}\Id s,\quad t\ge0.
\end{align}
\end{lem}

\begin{proof}
To start with recall that the linear hull of all exponential vectors $\expv{h}$ with $h\in\dom(\omega)$
is a core for $\Id\Gamma(\omega)$. We also recall that $\vp(\eul^{-\ii K\cdot y}v_{\UV})$ is relatively
$\Id\Gamma(\omega)$-bounded uniformly in $y\in\RR^2$. 
Finally, $\|U_t^{N,\pm}(x)\|_{t/2}^2\le 6\pi g^2N^2/m_{\bos}$, $t>0$, according to \cref{tnormU},
which together with \cref{normFt} implies that $\sup_{s\in[0,t]}\|W_{\UV,s}(x)\|$ is finite
for every $t\ge0$ and pointwise on $\Omega$. Thus,
by sesquilinearity and dominated convergence, it suffices to prove the assertion for 
$\phi_i=\expv{h_i}$, $i\in\{1,2\}$, with $h_1,h_2\in\dom(\omega)$.
Employing \cref{Wexpvec} and using that $\langle \expv{g}|\expv{h}\rangle_{\Fock} = \eul^{\langle g | h\rangle}$,
for all $g,h\in L^2(\RR^2)$,
we find $\langle\expv{h_1}|W_{\UV,t}(x)\expv{h_2}\rangle_{\Fock}=\eul^{\Theta(t)}$, $t\ge0$, with
\begin{align*}
\Theta(t)&\coloneq 
\tilde{u}^N_{\UV,t}(x) - \int_0^t V(X_s^x)\Id s
-\langle U^{N,-}_{\UV,t}(x)|h_2\rangle
-\langle h_1|U_{\UV,t}^{N,+}(x)\rangle+\langle h_1|\eul^{-t\omega}h_2\rangle
\\
&=\int_0^t\vt(s)\Id s+\langle h_1|h_2\rangle,
\end{align*}
where, in view of \cref{IBPUplus,defUplusminus,defn:complexphase0UV},
\begin{align}\nonumber
	\vt(s)&\coloneq
	\sum_{j=1}^N\langle U_{\UV,s}^{N,+}(x)|
	\eul^{-\ii K\cdot X_{j,s}^x}v_{\UV}\rangle -NE_{\UV}^{\ren}-V(X_s^x)
	\\\nonumber
	&\quad -\sum_{j=1}^N\langle\eul^{-s\omega-\ii K\cdot X_{j,s}^x}v_{\UV}|h_2\rangle  
	-\sum_{j=1}^N\langle h_1|\eul^{-\ii K\cdot X_{j,s}^x}v_{\UV}\rangle
	\\&\quad\label{ODEW2}
	+\langle h_1|\omega(U^{N,+}_{\UV,s}(x)-\eul^{-s\omega}h_2)\rangle.
\end{align}
Since $v(k)=v(-k) = \overline{v(-k)}$, the scalar product in the first expression is real:
\begin{align}\label{ODEW3}
\langle U_{\UV,s}^{N,+}(x)|\eul^{-\ii K\cdot X_{j,s}^x}v_{\UV}\rangle
&=\langle \eul^{-\ii K\cdot X_{j,s}^x}v_{\UV}|U_{\UV,s}^{N,+}(x)\rangle.
\end{align}
Further, since $\vt:[0,\infty)\to\RR$ is locally integrable, $\Theta$ is absolutely continuous on every compact interval
in $[0,\infty)$, and
\begin{align}\label{ODEW4}
\eul^{\Theta(t)}-\eul^{\Theta(0)}
&=\int_0^t\vt(s)\eul^{\Theta(s)}\Id s,\quad t\ge0.
\end{align}
On the other hand, using
\begin{align*}
	&a(g)\expv h = \langle g|h\rangle\expv h, \quad g,h \in L^2(\RR^2),\\
	& \langle \expv g|\Id\Gamma(\omega) \expv h\rangle_{\Fock} 
	= \langle g|\omega h\rangle \eul^{\langle g|h\rangle},\quad g\in L^2(\RR^2),\ h\in \dom(\omega),
\end{align*}
 and taking \cref{Wexpvec} into account, we find
\begin{align*}
\langle \eul^{-\ii K\cdot X_{j,s}^x}v_{\UV}|U_{\UV,s}^{N,+}(x)-\eul^{-s\omega}h_2\rangle\eul^{\Theta(s)} 
& =-\langle\expv{h_1}|a(\eul^{-\ii K\cdot X_{j,s}^x}v_{\UV})W_{\UV,s}(x)\expv{h_2}\rangle_{\Fock},
\\
\langle h_1|\omega(U^{N,+}_{\UV,s}(x)-\eul^{-s\omega}h_2)\rangle\eul^{\Theta(s)}
&=-\langle\expv{h_1}|\Id\Gamma(\omega)W_{\UV,s}(x)\expv{h_2}\rangle_{\Fock},
\\
\langle h_1|\eul^{-\ii K\cdot X_{j,s}^x}v_{\UV}\rangle\eul^{\Theta(s)}
&=\langle\expv{h_1}|\ad(\eul^{-\ii K\cdot X_{j,s}^x}v_{\UV})W_{\UV,s}(x)\expv{h_2}\rangle_{\Fock}.
\end{align*}
Combining the previous three relations with \eqref{ODEW2} and \eqref{ODEW3}, we see that \eqref{ODEW4}
is equivalent to the statement with $\phi_i=\expv{h_i}$.
\end{proof}


\section{Main contributions to the complex action}\label{sec:basicproc}

\noindent
With this \lcnamecref{sec:basicproc} we start our investigation of the Feynman--Kac semigroups,
that eventually will result in a proof of \cref{prop:TUVonL2} and generalizations thereof in \cref{sec:FKInt}.
As mentioned earlier, a crucial step will be to derive a more regular expression for the complex action $\tilde{u}_{\UV,t}^N(x)$
permitting to drop the ultraviolet cutoff. Before we do so in the succeeding \cref{sec:complex}, we shall
study the two main contributions to the new expression for the complex action in this \lcnamecref{sec:basicproc},
namely a martingale contribution and the effective interaction potential mediated by the radiation field.

For technical purposes, we introduce versions of these with infrared and ultraviolet cutoffs, 
usually denoted by $\sigma$ and $\UV$, respectively. The purpose of $\sigma$ is twofold: 
Firstly, when studying convergence properties as the ultraviolet cutoff goes to infinity, 
we encounter difference terms with an effective, large infrared cutoff $\sigma$ and with $\UV=\infty$.
Secondly, to obtain good exponential moment bounds on the complex action, we have to split up some
of its contributions into two parts, comprising boson momenta $|k|<\sigma$ and $\sigma\le|k|<\UV$,
respectively, and treat both parts differently.

\subsection{The martingale contribution to the complex action}\label{subsec:m}

We first want to discuss the martingale introduced in \cref{def:martingale} and certain infrared and/or
ultraviolet cutoff versions of it. These are stochastic
integrals with respect to $\wt{N}_X$ as described in \cref{ssecPPP}.
The integrands we are interested in are given by
\begin{align*}
h_{\sigma,\UV}(x,s,z)&\coloneq \sum_{\ell=1}^N\langle U^{N,+}_{\sigma,\UV,s}(x)|
\eul^{-\ii K\cdot X_{\ell,s-}^x}(\eul^{-\ii K\cdot z_\ell}-1)\beta\rangle,
\quad s\ge0,\,x,z\in\RR^{2N},
\end{align*}
where $0\le\sigma <\UV\le\infty$ and
\begin{align}\label{def:Ucutoff}
U_{\sigma,\UV,s}^{N,\pm}(x)&\coloneq \chi_{B_\UV\setminus B_\sigma}U_{s}^{N,\pm}(x)
=\chi_{B_\sigma^c}U_{\UV,s}^{N,\pm}(x).
\end{align}
As a consequence of \cref{remUpmcont,lem:Upmadapted}, $(U^{N,+}_{\sigma,\UV,s}(x))_{s\ge0}$ is adapted and continuous. 
Thus, we can read off from the above formula that all paths $s\mapsto h_{\sigma,\UV}(x,s,z)$ are left-continuous and,
for fixed $s\ge0$, the map $(z,\omega)\mapsto (h_{\sigma,\UV}(x,s,z))(\omega)$ is 
$\fr{B}(\RR^{2N})\otimes\fr{F}_t$-measurable.
In particular, $h_{\sigma,\UV}(x,\cdot,\cdot)$ is predictable.
Recall that the space of integrands $\scr{H}_2$ has been introduced prior to \cref{ItoIso}.

\begin{lem}\label{lemexpmomentm}
Let $0\le\sigma<\UV\le\infty$, $x\in\RR^{2N}$ and consider
\begin{align}\label{mstochint}
m_{\sigma,\UV,t}^N(x)&\coloneq \int_{(0,t]\times\RR^{2N}}h_{\sigma,\UV}(x,s,z)\Id\wt{N}_X(s,z),\quad t\ge0.
\end{align}
Then the following holds:
\begin{enumerate}
\item[{\rm(i)}] $h_{\sigma,\UV}\in\scr{H}_2$ and in particular the integrals in \eqref{mstochint}
are well-defined isometric stochastic integrals. 
Up to indistinguishability they define a unique c\`{a}dl\`{a}g square-integrable martingale
$m_{\sigma,\UV}^N(x)=(m_{\sigma,\UV,t}^N(x))_{t\ge0}$. 
\item[{\rm(ii)}]
For every $p\ge2$ there exists $c_p>0$, depending solely on $p$, such that
\begin{align*}
\EE\bigg[\sup_{s\in[0,t]}|m_{\sigma,\UV,s}^N(x)|^p\bigg]
&\le c_pg^{2p}\bigg(\frac{N^{3p/2}t^{p/2}}{\omega(\sigma)^{p/2}}
+\frac{N^{p+1}t}{\omega(\sigma)^{p-1}}\bigg),\quad t\ge0.
\end{align*}
\item[{\rm(iii)}]
There exists a universal constant $c>0$ such that
\begin{align*}
\EE\bigg[\sup_{s\in[0,t]}e^{\pm pm_{\sigma,\UV,s}^N(x)}\bigg]
&\le \frac{\alpha}{\alpha-1}\exp\left(\frac{c\alpha p^2g^4N^3t}{\omega(\sigma)}\cdot
\eul^{4\pi\alpha pg^2N/\omega(\sigma)}\right),
\end{align*}
for all $\alpha>1$, $p>0$ and $t\ge0$.
\end{enumerate}
\end{lem}

\begin{proof}
We split the proof into four steps.

\smallskip\noindent
{\em Step~1: Bounds on $h_{\sigma,\UV}$.}
Since $U_s^+[X_{j,\bullet}]$ is,  pointwise on $\Omega$, given by a Bochner-Lebesgue integral in $L^2(\RR^2)$,
we can insert the definitions \cref{def:UtN,def:Ucutoff} into the formula for $h_{\sigma,\UV}$
and switch the order of integration and taking scalar products. Applying Fubini's theorem afterwards we thus get
\begin{align}\label{eq:hstep1}
	\begin{aligned}
		&h_{\sigma,\UV}(x,s,z)\\&
		= g^2\sum_{j,\ell=1}^N\int_{B_{\UV}\setminus B_\sigma}
		\int_0^s\frac{\eul^{-(s-r)\omega(k)+\ii k\cdot (X^x_{j,r}-X^x_{\ell,s-})}
			(\eul^{-\ii k\cdot z_\ell}-1)}{\omega(k)(\omega(k)+\psi(k))}\Id r\,\Id k,
	\end{aligned}
\end{align}
for all $s\ge0$ and $z\in\RR^{2N}$. Under the double integral we can apply the bounds 
$|\eul^{-\ii k\cdot z_\ell}-1|\le 2^{1-a_\ell}|k|^{a_\ell}|z|^{a_\ell}$ as long as $a_\ell\in[0,1)$ to retain integrability.
Thus, we find, for all $s\ge0$, $z\in\RR^{2N}$ and $a_1,\ldots,a_N\in[0,1)$,
\begin{align}\nonumber
|h_{\sigma,\UV}(x,s,z)|&\le 2g^2N\sum_{\ell=1}^N|z_\ell|^{a_\ell}
\int_{B_\sigma^c}\int_0^s\frac{\eul^{-r\omega(k)}|k|^{a_\ell}}{\omega(k)(\omega(k)+\psi(k))}\Id r\,\Id k
\\\nonumber
&\le 2g^2N\sum_{\ell=1}^N|z_\ell|^{a_\ell}\int_{B_\sigma^c}\frac{|k|^{a_\ell}}{\omega(k)^3}\Id k
\\\label{mona0}
&\le 4\pi g^2N\sum_{\ell=1}^N\frac{|z_\ell|^{a_\ell}}{1-a_\ell}\cdot\frac{1}{\omega(\sigma)^{1-a_\ell}},
\end{align}
where we used $|k|^{a_\ell}\le \omega(k)^{a_\ell}$ in the last step.
Given a vector $z\in\RR^{2N}$ and $\ell\in\{1,\ldots,N\}$ such that $z_j=0$ for all $j\in\{1,\ldots,N\}\setminus\{\ell\}$, the sum in \cref{mona0} reduces to one summand when we pick $a_\ell = 0$ and $a_j\in(0,1)$ for $j\not=\ell$. 
Since, by \cref{eq:nuX}, 
at most one two-dimensional component of $z\in\RR^{2N}$ is non-zero for $\nu_X$-almost every $z\in\RR^{2N}$,
this implies
\begin{align}\label{mona1}
\sup_{s\ge0}
|h_{\sigma,\UV}(x,s,z)|&\le\frac{4\pi g^2N}{\omega(\sigma)},\quad
\text{$\nu_X$-a.e. $z\in\RR^{2N}$.}
\end{align}
Given $a\in(0,1)$, we apply \eqref{mona1} whenever $|z_\ell|\omega(\sigma)\ge(1-a)^{1/a}$
for some $\ell\in\{1,\ldots,N\}$ and choose $a_1=\ldots=a_N=a$ in \eqref{mona0} otherwise. 
For any $p>0$ and $\nu_X$-almost every $z\in\RR^{2N}$, this yields
\begin{align}\label{mona2}
	\sup_{s\ge0}|h_{\sigma,\UV}(x,s,z)|^p
	\le \frac{(4\pi g^2N)^p}{\omega(\sigma)^{p}}
	\sum_{\ell=1}^N1\wedge\left(\frac{\omega(\sigma)^{pa}}{(1-a)^{p}}|z_\ell|^{pa}\right).
\end{align}
The Bessel function asymptotics \crefnosort{bdBesselK,def:nu} imply the estimate
\begin{align}\nonumber
&\int_{\RR^2}1\wedge\left(\frac{\omega(\sigma)^{pa}}{(1-a-\ve)^{p}}|z_\ell|^{pa}\right)\Id\nu(z_\ell)
\\&\label{mona17}
\le \int_{\RR^2}\bigg\{1\wedge\left(\frac{\omega(\sigma)^{pa}}{(1-a-\ve)^p}|z_\ell|^{pa}\right)\bigg\}
\frac{c}{|z_\ell|^3}\Id z_\ell
=\frac{2\pi cpa}{pa-1}\cdot\frac{\omega(\sigma)}{(1-a-\ve)^{1/a}},
\end{align}
which applies to all particle masses $m_{\p}\ge0$ and $p,a>0$, $\ve\ge0$ with $pa>1$ and $a+\ve<1$. 
(The parameter $\ve$ is needed only in the proof of \cref{lem:canm}.)
Choosing $\ve=0$ and $a=2/3$, say, we arrive at
\begin{align}\label{mona3}
\int_{0}^t\int_{\RR^{2N}}|h_{\sigma,\UV}(x,s,z)|^p\Id\nu_X(z)\,\Id s&\le 
\frac{c_pg^{2p}N^{p+1}t}{\omega(\sigma)^{p-1}},\quad t\ge0,\,p\ge2.
\end{align}

\smallskip\noindent
{\em Step~2: Proof of {\rm(i)}.} 
By \cref{mona3}, $h_{\sigma,\UV}\in\scr{H}_2$, whence
 \eqref{mstochint} are well-defined isometric stochastic integrals and the corresponding
integral process is a square-integrable martingale. 
Since we are working under the usual hypotheses, the latter has a c\`{a}dl\`{a}g modification that is
unique up to indistinguishability, cf. \cite[Theorem~2.1.7]{Applebaum.2009}.

\smallskip

\noindent
{\em Step~3: Proof of  {\rm(ii)}.}
Let $p\ge2$. According to Kunita's inequality (see, e.g., \cite[Theorem~4.4.23]{Applebaum.2009}),
there is a solely $p$-dependent $c_p>0$ such that
\begin{align*}
\EE\bigg[\sup_{s\in[0,t]}|m_{\sigma,\UV,s}^N(x)|^p\bigg]
\le\ & c_p\EE\bigg[\bigg(\int_0^t\int_{\RR^{2N}}|h_{\sigma,\UV}(x,s,z)|^2\Id\nu_X(z)\,\Id s\bigg)^{p/2}\bigg]
\\
&\quad+\EE\bigg[\int_0^t\int_{\RR^{2N}}|h_{\sigma,\UV}(x,s,z)|^p\Id\nu_X(z)\,\Id s\bigg],
\quad t\ge0.
\end{align*}
Inserting  \eqref{mona3}, we find the bound asserted in (ii).

\smallskip
\noindent
{\em Step~4: Proof of {\rm(iii)}.}
Let $\alpha>1$ and $p>0$. Then an elementary manipulation of the exponential series yields
\begin{align}\label{Taylorexph}
\frac{1}{\alpha}|\eul^{\pm\alpha ph_{\sigma,\UV}(x,s,z)}-1\mp\alpha ph_{\sigma,\UV}(x,s,z)|
&\le \frac{\alpha p^2}{2}|h_{\sigma,\UV}(x,s,z)|^2\eul^{\alpha p|h_{\sigma,\UV}(x,s,z)|}.
\end{align}
Employing \eqref{mona1} in the exponent of the exponential on the previous right hand side and using
\eqref{mona3} with $p=2$ we find, with some universal constant $c>0$,
\begin{align}\nonumber
\frac{1}{\alpha}\int_{0}^t\int_{\RR^{2N}}|\eul^{\pm\alpha ph_{\sigma,\UV}(x,s,z)}&-1\mp\alpha ph_{\sigma,\UV}(x,s,z)|
\Id\nu_X(z)\,\Id s
\\\label{mona4}
&\le \frac{c\alpha p^2g^4N^3t}{\omega(\sigma)}\cdot\eul^{4\pi\alpha pg^2N/\omega(\sigma)},
\quad t\ge0.
\end{align}
The asserted exponential moment bound now follows from \eqref{mona0} (with $a_\ell=0$ for all~$\ell$),
\eqref{mona3} (with $p=2$), \eqref{mona4}, an inequality due to 
Applebaum and Siakalli \cite{Applebaum.2009,Siakalli.2019} and a standard argument based on a suitable 
layer cake representation; see Lemma~\ref{lemexpbdtildeN}.
\end{proof}
The next lemma about continuous dependence on the initial positions of the matter particles is analogous
to \cite[Lemma~4.14]{MatteMoller.2018}. It is proved in a similar, standard fashion with Kunita's
inequality replacing the Burkholder inequalities used in \cite{MatteMoller.2018}.
\begin{lem}\label{lem:canm}
Let $0\le\sigma<\UV\le\infty$. Then,
for every $x\in\RR^{2N}$, the c\`{a}dl\`{a}g modification of $m_{\sigma,\UV}^N(x)$ found in the first part of
Lemma~\ref{lemexpmomentm} can be chosen such that
\begin{enumerate}
\item[{\rm(a)}]
$\RR^{2N}\ni x\mapsto (m_{\sigma,\UV,t}^N(x))(\omega)$ is continuous for all $t\ge0$ and $\omega\in\Omega$.
\item[{\rm(b)}]
For all $t\ge0$, $\omega\in\Omega$ and compact $K\subset\RR^{2N}$,
\begin{align*}
\lim_{s\downarrow t} \sup_{x\in K}|m_{\sigma,\UV,s}^N(x)-m_{\sigma,\UV,t}^N(x)|&=0.
\end{align*}
\end{enumerate} 
\end{lem}

\begin{proof}
Let $x,\tilde{x}\in\RR^{2N}$. By \cref{eq:hstep1}, 
\begin{align*}
\tilde{h}_{\sigma,\UV}(\tilde{x},x,s,z)&\coloneq h_{\sigma,\UV}(\tilde{x},s,z)-h_{\sigma,\UV}(x,s,z)
\\
&= g^2\sum_{j,\ell=1}^N\int_{B_{\UV}\setminus B_\sigma}
\int_0^s\frac{\eul^{-(s-r)\omega(k)+\ii k\cdot (X_{j,r}-X_{\ell,s-})}
}{\omega(k)(\omega(k)+\psi(k))}\cdot \theta_{j,\ell}(k,\tilde{x},x,z) \Id r\,\Id k,
\end{align*} 
where $\theta_{j,\ell}(k,\tilde{x},x,z)\coloneq (\eul^{\ii k\cdot(\tilde{x}_{j}-\tilde{x}_{\ell})}-\eul^{\ii k\cdot(x_j-x_\ell)})
(\eul^{-\ii k\cdot z_\ell}-1)$. 
We note
\begin{align*}
&|\theta_{j,\ell}(k,\tilde{x},x,z)|\le 2^{2-a_\ell}|x-\tilde{x}|^\ve|z_\ell|^{a_\ell}|k|^{\ve+a_\ell},
\end{align*}
for all $\ve,a_1,\ldots,a_\ell\in[0,1]$. Assuming $\ve+\max\{a_1,\ldots,a_N\}<1$,
replacing $2|k|^{a_\ell}$ by $4|x-\tilde{x}|^\ve|k|^{\ve+a_\ell}$ 
in the second and third members of \eqref{mona0}
and following the derivation of \cref{mona2}, we obtain
\begin{align*}
\sup_{s\ge0}|\tilde{h}_{\sigma,\UV}(\tilde{x},x,s,z)|^p
&\le \frac{(8\pi g^2N)^p|x-\tilde{x}|^{p\ve}}{\omega(\sigma)^{(1-\ve)p}}
 \sum_{\ell=1}^N1\wedge\bigg(\frac{\omega(\sigma)^{pa}}{(1-\ve-a)^p}|z_\ell|^{pa}\bigg),
\end{align*}
for all $a,p,\ve>0$ such that $a+\ve<1$. Here the right hand side is $\nu_X$-integrable provided
that $pa>1$. Again using \cref{mona17} and choosing $a=2/3$, we find
\begin{align*}
\int_{0}^t\int_{\RR^{2N}}|\tilde{h}_{\sigma,\UV}(\tilde{x},x,s,z)|^p\Id\nu_X(z)\,\Id s
&\le \frac{c_{p,\ve}g^{2p}N^{p+1}|x-\tilde{x}|^{p\ve} t}{\omega(\sigma)^{p-p\ve-1}},\quad t\ge0,
\end{align*}
for all $p\ge2$ and $\ve\in(0,1/3)$. (Here and in what follows constants depend only on the quantities 
displayed in their subscripts.)
Combining these remarks with Kunita's inequality, we arrive at
\begin{align*}
&\EE\bigg[\sup_{s\in[0,t]}|m_{\sigma,\UV,s}^N(\tilde{x})-m_{\sigma,\UV,s}^N(x)|^p\bigg]
\\
&=\EE\bigg[\sup_{s\in[0,t]}\bigg|\int_{(0,s]\times\RR^{2N}}\tilde{h}_{\sigma,\UV}(\tilde{x},x,s,z)\Id\wt{N}_X(s,z)\bigg|^p\bigg]
\\
&\le c_p\EE\bigg[\bigg(\int_0^t\int_{\RR^{2N}}|\tilde{h}_{\sigma,\UV}(\tilde{x},x,s,z)|^2\Id\nu_X(z)\,\Id s\bigg)^{p/2}\bigg]
\\
&\quad+\EE\bigg[\int_0^t\int_{\RR^{2N}}|\tilde{h}_{\sigma,\UV}(\tilde{x},x,s,z)|^p\Id\nu_X(z)\,\Id s\bigg]
\\
&\le  c_{m_{\bos},p,\ve,g,N}(t^{p/2}+t)|x-\tilde{x}|^{p\ve},
\end{align*}
for all $t\ge0$, $\ve\in(0,1/3)$ and $p\ge2$. Choosing $\ve=1/4$ and $p$ such that
$p\ve=2N+1$, we see that all assertions follow from the Kolmogorov-Neveu lemma;
see, e.g., \cite[pp. 268/9]{Metivier.1982}.
\end{proof}

\subsection{Effective interaction terms}

Next, we treat the effective interaction introduced in \cref{def:wy,def:wtN} together with variants containing cutoffs.

Since $\chi_{B_\sigma^c}v_\UV\beta$ is integrable for finite $\Lambda>\sigma\ge 0$, we can consider the Fourier integrals
\begin{align}\label{def:wIRUV}
\begin{aligned}
	w_{\sigma,\UV}(y)&\coloneq \langle\eul^{\ii K\cdot y}\chi_{B_\sigma^c}v_\UV|\beta\rangle
	\\
	&=g^2\int_{B_{\UV}\setminus B_\sigma}\frac{\eul^{-\ii k\cdot y}}{\omega(k)(\omega(k)+\psi(k))}\Id k
	=g^2\int_\sigma^{\UV} \vt(r,|y|)\Id r,
\end{aligned}
\end{align}
for $0\le\sigma<\UV<\infty$ and $y\in\RR^2$, where $\vt$ is defined in \eqref{defvt}. 
Since the Bessel function $J_0$ in \eqref{defvt} satisfies the bound \cref{eq:J0bound}, it is clear that
\begin{align*}
\lim_{\UV\to\infty}w_{\sigma,\UV}(y)&=w_{\sigma,\infty}(y)\coloneq g^2\int_\sigma^\infty\vt(r,|y|)\Id r,
\quad y\in\RR^2\setminus\{0\}.
\end{align*}
Setting $w_{\sigma,\infty}(0)\coloneq 0$, say,  $w_{\sigma,\infty}:\RR^2\to\RR$ 
is the up to its value at $0$ unique representative of the 
distributional Fourier transform of $\chi_{B_\sigma^c}v\beta$ that is continuous on $\RR^2\setminus\{0\}$.

\begin{rem}\label{remdefw}
Since $v\beta=g^2/\omega(\omega+\psi)\in L^{\tilde{p}}(\RR^2)$ for every $\tilde{p}>1$, the Hausdorff-Young inequality 
implies
$w_{\sigma,\infty}\in L^{p}(\RR^2)$
and
$\|w_{\sigma,\infty}\|_{p}\le \|\chi_{B_\sigma^c}v\beta\|_{p'}$ for all $p\in[2,\infty)$,
where $p'$ is the exponent conjugated to $p$.
For all $p\in(1,\infty)$, or equivalently $p'\in(1,\infty)$, and $\sigma\ge0$, we further have
\begin{align*}
\|\chi_{B_\sigma^c}v\beta\|_{p'}
&\le g^2\bigg(\pi\int_\sigma^\infty\frac{2r}{(r^2+m_{\bos}^2)^{p'}}\Id r\bigg)^{1/p'}
=\frac{c_pg^2}{(\sigma^2+m_{\bos}^2)^{1/p}},
\end{align*}
for some solely $p$-dependent $c_p\in(0,\infty)$.
Combined, we find the bound
\begin{align}\label{Lpvbetasigma2}
\|w_{\sigma,\UV}\|_p&\le \frac{c_pg^2}{(\sigma^2+m_{\bos}^2)^{1/p}},\quad 0\le\sigma<\UV\le\infty,\,p\in[2,\infty).
\end{align}
\end{rem}
We can now treat a (possibly) cutoff version of the full effective interaction expression defined in \cref{def:wtN}.
For all $0\le\sigma<\UV\le\infty$ and $x\in\RR^{2N}$, it is given by
\begin{align}\nonumber
w^N_{\sigma,\UV,t}(x)&\coloneq \sum_{\substack{j,\ell=1\\ j\not=\ell}}^N\int_0^tw_{\sigma,\UV}(X_{\ell,s}^x-X_{j,s}^x)\Id s
\\\label{def:wNsigmaUV}
&=\sum_{1\le j<\ell\le N}\int_0^t2w_{\sigma,\UV}(X_{\ell,s}^x-X_{j,s}^x)\Id s,\quad \mbox{whenever $N\ge2$.}
\end{align}
Notice that $w_{\sigma,\UV}(-y)=w_{\sigma,\UV}(y)$ for all $y\in\RR^2$. Also notice that we are again using
the conventions introduced in \cref{sssecKato} in the case $\UV=\infty$; i.e., the pathwise integrals in 
\eqref{def:wNsigmaUV} are $0$ by definition, unless they are well-defined for {\em all} $t\ge0$
which is the case with probability $1$.
For $N=1$, we set $w_{\sigma,\UV,t}^1(x)\coloneq 0$. 

The next \lcnamecref{lemexpmomentw} is an exponential moment bound for the effective interaction.
\begin{lem}\label{lemexpmomentw}
There exist universal constants $c,c'>0$ such that
\begin{align*}
&\sup_{x\in\RR^{2N}}\EE\left[\sup_{s\in[0,t]}\eul^{p|w_{\sigma,\UV,s}^N(x)|}\right]
\le c^N\exp\left(tc'p^2g^4N^2(N-1)
\Big(\frac{m_{\p}}{\omega(\sigma)^2}+\frac{1}{\omega(\sigma)}\Big)\right),
\end{align*}
for all $0\le\sigma<\UV\le\infty$, $t\ge0$ and $p\in[1,\infty)$.
\end{lem}

\begin{proof}
Let $t>0$ and $x\in\RR^{2N}$. Exactly as in \cite[Equation (4.26)]{MatteMoller.2018}
(where an observation from \cite{Bley.2018} is used; see also \cite[Appendix B]{MatteMoller.2018}),
we can exploit the independence of the processes $X_1,\ldots,X_N$ to obtain the bound
\begin{align*}
\EE\left[\sup_{s\in[0,t]}\eul^{p|w_{\sigma,\UV,s}^N(x)|}\right]
&\le \max_{1\le j<\ell\le N}\EE\left[\sup_{s\in[0,t]}
\eul^{2pN\int_0^t|w_{\sigma,\UV}(X_{\ell,s}^x-X_{j,s}^x)|\Id s}\right]^{(N-1)/2}.
\end{align*}
Pick two indices $j$ and $\ell$ with $1\le j<\ell\le N$. Since $X_j$ and $X_\ell$ are independent and
since $\psi(-\eta)=\psi(\eta)$ for all $\eta\in\RR^2$, 
we know that the difference $X_\ell-X_j$ is a L\'{e}vy process associated with the L\'{e}vy symbol $-2\psi$.
Therefore, we can apply the bound \eqref{relLpLinfty3} with $a=d=p=2$ to deduce that
\begin{align*}
\EE\left[\sup_{s\in[0,t]}
\eul^{2pN\int_0^t|w_{\sigma,\UV}(X_{\ell,s}^x-X_{j,s}^x)|\Id s}\right]
&\le c\exp\left(c'(m_{\p}^{1/2}\|2pNw_{\sigma,\UV}\|_2+\|2pNw_{\sigma,\UV}\|_4)^2t\right),
\end{align*}
with universal constants $c,c'>0$. Combining the above remarks with
\eqref{Lpvbetasigma2}, we arrive at the asserted inequality.
\end{proof}
We can also bound the $p$'th moments of the effective interaction.
\begin{lem}\label{lempmomentw}
Let $p\in[1,\infty)$. Then there exists $c_p>0$, depending solely on $p$, such that
for all $0\le\sigma<\UV\le\infty$, $t\ge0$ and $x\in\RR^{2N}$,
\begin{align}
&\EE\bigg[\sup_{s\in[0,t]}|w_{\sigma,\UV,s}^N(x)|^p\bigg]
\label{convw0}
\le c_pg^{2p}N^p(N-1)^{p}t^{p-1/2}\bigg(\frac{m_{\p}}{\omega(\sigma)^2}
+\frac{1}{\omega(\sigma)}\bigg)^{1/2}.
\end{align}
\end{lem}
\begin{proof}
First applying the ordinary and afterwards the generalized Minkowski inequality, we find
\begin{align*}
\EE\bigg[\sup_{s\in[0,t]}|w_{\sigma,\UV,s}^N(x)|^p\bigg]^{\frac{1}{p}}
&\le\sum_{1\le j<\ell\le N}
\EE\bigg[\bigg(\int_0^t2|w_{\sigma,\UV}(X_{\ell,s}^x-X_{j,s}^x)|\Id s\bigg)^p\bigg]^{\frac{1}{p}}
\\
&\le 2 \sum_{1\le j<\ell\le N}
\int_0^t\EE\big[|w_{\sigma,\UV}(X_{\ell,s}^x-X_{j,s}^x)|^p\big]^{\frac{1}{p}}\Id s.
\end{align*}
Here $X_{\ell,s}-X_{j,s}$ has the same law as $X_{1,2s}$. Thus, by \eqref{relLpLinfty2} (with $p=d=2$)
and~\eqref{Lpvbetasigma2},
\begin{align*}
\EE\big[|w_{\sigma,\UV}(X_{\ell,s}^x-X_{j,s}^x)|^p\big]
&\le c(\|w_{\sigma,\UV}\|_{4p}^p+m_{\p}^{1/2}\|w_{\sigma,\UV}\|_{2p}^p) s^{-1/2}
\\
&\le \frac{c_pg^{2p}}{s^{1/2}}\bigg(\frac{1}{\omega(\sigma)^{1/2}}
+\frac{m_{\p}^{1/2}}{\omega(\sigma)}\bigg), \quad s>0.
\end{align*}
Putting these remarks together we arrive at \eqref{convw0}.
\end{proof}


\section{Discussion of the complex action}\label{sec:complex}
\noindent
After the preliminary technical discussions of its main contributions in the previous \lcnamecref{sec:basicproc}, 
we now turn to the analysis of the complex action itself. In \cref{ssecdefcomplexaction} we shall rewrite
the formula for the complex action with ultraviolet cutoff, obtaining a more regular expression that stays
meaningful when the cutoff is dropped. \cref{ssecexpmbca,ssecconvca} are devoted to
exponential moment bounds and convergence properties as $\UV\to\infty$, respectively.
In most parts of this section we allow for optional infrared cutoffs for the technical reasons explained at
the beginning of  \cref{sec:basicproc}.

\subsection{Definition of the complex action}\label{ssecdefcomplexaction}

First, we introduce an abbreviation for the complex action with finite $\UV$ 
containing an optional infrared cutoff; recall the earlier notation \cref{defuUV0,def:Ucutoff,def:vUVs}.

\begin{defn}\label{def:action-process}
Let $\UV\in(0,\infty)$, $\sigma\in[0,\UV)$, $t\ge0$ and $x\in\RR^{2N}$. Then we define
\begin{align}\label{defuUV}
\tilde{u}_{\sigma,\UV,t}^N(x)&\coloneq
\int_0^t\langle U^{N,+}_{\sigma,\UV,s}(x)|v_{\UV,s}^N(x)\rangle\Id s-tNE_{\sigma,\UV}^{\ren},
\end{align}
where 
\begin{align}\label{defEren}
E_{\sigma,\UV}^{\ren}&\coloneq \langle\chi_{B_{\UV}\setminus B_\sigma}v|\beta\rangle.
\end{align}
\end{defn}
In view of \cref{defvUV,def:beta} we indeed have  $E_{0,\UV}^{\ren}=E_{\UV}^{\ren}$ with $E_{\UV}^{\ren}$ given by
\cref{def:EUVren}. To reduce cluttering of indices a bit, we shall return to the earlier notation  
$\tilde{u}_{0,\UV,t}^N(x)=\tilde{u}_{\UV,t}^N(x)$ whenever $\sigma=0$.

Using It\^o's formula, we shall rewrite the integral in \cref{defuUV} in \cref{lem:actionito} below. The alternative expression
for the complex action we find is \cref{defcomplexactionsigmaUV} in the next definition. 
As opposed to $\tilde{u}_{\sigma,\UV,t}^N(x)$, the new expression $u_{\sigma,\UV,t}^N(x)$
is meaningful for $\UV=\infty$ as well!

Recall the definition of the martingale contribution \cref{mstochint} and the effective interaction \cref{def:wNsigmaUV}.
For all $\UV\in(0,\infty]$, we also abbreviate
\begin{align*}
\beta_\UV\coloneq\chi_{B_{\UV}}\beta,\qquad
\beta_{\UV,t}^N(x)&\coloneq \sum_{j=1}^N\eul^{-\ii K\cdot X_{j,t}^x}\beta_{\UV},\quad t\ge0,\,x\in\RR^{2N}.
\end{align*}

\begin{defn}\label{defn:complexactioninfty}
Let $x\in\RR^{2N}$.  Then we define
\begin{align}\label{defcomplexactionsigmaUV}
u_{\sigma,\UV,t}^N(x)&\coloneq w^N_{\sigma,\UV,t}(x)-c^N_{\sigma,\UV,t}(x)+m^N_{\sigma,\UV,t}(x),\quad t\ge0,
\end{align}
for all $0\le\sigma<\UV\le\infty$, with
\begin{align}\label{def:cNsigmaUV}
c^N_{\sigma,\UV,t}(x)&\coloneq 
\langle U_{\sigma,\UV,t}^{N,+}(x)|\beta_{\UV,t}^N(x)\rangle.
\end{align}
\end{defn}

Recall that the three contributions to the complex action and the action itself 
for our model without any cutoffs have already been introduced in \eqref{defvt} through \eqref{def:action}. In fact,
\begin{align}\label{defcomplexaction}
u_t^N(x)&= u_{0,\infty,t}^N(x), \ \ \mbox{and analogously with $c$, $m$ or $w$ put in place of $u$.}
\end{align}

\begin{rem}\label{rem:bdcN}
We note for later use that
\begin{align}\label{bdcN}
|c_{\sigma,\UV,t}^N(x)|\le\frac{2\pi g^2N^2}{\omega(\sigma)},\quad \sigma,t\ge0,
\end{align}
for all $x\in\RR^{2N}$, by the last two relations in \eqref{mona0} with $a_\ell=0$.
\end{rem}

\begin{rem}\label{eq:u-decomp}
Let $0\le \sigma <\kappa < \UV\le \infty$ and $x\in\RR^{2N}$. Then
\cref{mstochint,def:wIRUV,def:wNsigmaUV,def:cNsigmaUV} $\PP$-a.s. imply
	\begin{align}
		u_{\sigma,\UV,t}^N(x) = u_{\sigma,\kappa,t}^N(x) + u_{\kappa,\UV,t}^N(x), \quad t\ge0.
	\end{align}
\end{rem}

\begin{lem}\label{lem:actionito}
Let $0\le\sigma<\UV<\infty$ and $x\in\RR^{2N}$. Then, $\PP$-a.s.,
\begin{align}\label{ItoNeu0}
\tilde{u}_{\sigma,\UV,t}^N(x)&=u_{\sigma,\UV,t}^N(x),\quad t\ge0.
\end{align}
\end{lem}

\begin{proof}
We pick an orthonormal basis $\{e_i:i\in\NN\}$ of the real Hilbert space $\mc{h}_{\RR}$ defined
in \eqref{defmcRR} and put $P_n\coloneq\sum_{i=1}^n|e_i\rangle\langle e_i|$ for every $n\in\NN$.
We introduce the projections $P_n$ as we wish to employ the It\^{o} formula \cref{ItoAfX}.
(Alternatively, we could also use an It\^{o} formula for infinite dimensional processes \cite[Theorem~27.2]{Metivier.1982} and
re-write the $\Id X$-integrations and summations over jumps appearing there.)

Let $i\in\{1,\ldots,n\}$.
Recalling that $U_{\sigma,\UV,t}^{N,+}$ is $\mc h_\RR$-valued, cf. \cref{subsecerg}, and in view of \eqref{IBPUplus}, 
the paths of the adapted  process given by $A_{i,t}\coloneq\langle U_{\sigma,\UV,t}^{N,+}(x)|e_i\rangle$
are real-valued and absolutely continuous on every compact subinterval of $[0,\infty)$, and it satisfies
\begin{align}\label{ZsdAis}
\int_{(0,t]}Z_s\Id A_{i,s}&=
\int_0^tZ_s\langle \chi_{B_\sigma^c}v_{\UV,s}^N(x)-\omega U_{\sigma,\UV,s}^{N,+}(x)|e_i\rangle\Id s,\quad t\ge0,
\end{align}
for every real-valued predictable process $(Z_s)_{s\ge0}$ with locally bounded paths.
Furthermore, since $\UV<\infty$, the function $f_i:\RR^{2N}\to\RR$ given by
\begin{align*}
f_i(x)&\coloneq\sum_{j=1}^N\langle e_i|\eul^{-\ii K\cdot x_j}\beta_{\UV}\rangle,\quad x\in\RR^{2N},
\end{align*}
is smooth and bounded with bounded derivatives of any order. 
Thus, $\psi_X(-\ii\nabla)f_i$ is well-defined by \cref{def:psiXnabla}.
Also taking \eqref{eq:nuX} into account, applying Fubini's theorem and using \cref{LKpsi} afterwards, we find
\begin{align*}
\psi_X(-\ii\nabla)f_i(x)&=\sum_{j=1}^N\langle e_i|\eul^{-\ii K\cdot x_j}\psi\beta_{\UV}\rangle,\quad x\in\RR^{2N}.
\end{align*}
Applying the It\^{o} formula \cref{ItoAfX} to each term under the sum over $j$ and 
using $A_{i,0}=0$ and \cref{ZsdAis} with $Z_s=f_i(X_{s-}^x)$, we $\PP$-a.s. obtain
\begin{align*}
A_{i,t}f_i(X_t^x)
&=\sum_{j=1}^N\int_0^t\langle \chi_{B_\sigma^c}v_{\UV,s}^{N}(x)|e_i\rangle
\langle e_i|\eul^{-\ii K\cdot X_{j,s}^x}\beta_{\UV}\rangle\Id s
\\
&\quad
-\sum_{j=1}^N\int_0^t\langle \omega U_{\sigma,\UV,s}^{N,+}(x)|e_i\rangle
\langle e_i|\eul^{-\ii K\cdot X_{j,s}^x}\beta_{\UV}\rangle\Id s
\\
&\quad-\sum_{j=1}^N\int_0^t\langle U_{\sigma,\UV,s}^{N,+}(x)|e_i\rangle
\langle e_i|\eul^{-\ii K\cdot X_{j,s}^x}\psi\beta_{\UV}\rangle\Id s
\\
&\quad +\sum_{j=1}^N\int_{(0,t]\times\RR^{2N}}\langle U_{\sigma,\UV,s}^{N,+}(x)|e_i\rangle
\langle e_i|\eul^{-\ii K\cdot X_{j,s-}^x}(\eul^{-\ii K\cdot z_j}-1)\beta_{\UV}\rangle\Id\wt{N}_X(s,z),
\end{align*}
for all $t\ge0$, where the first two $\Id s$-integrals are the $\Id A_{i,s}$-integral from the It\^o formula in conjunction with \cref{ZsdAis}.
Summing over $i\in\{1,\ldots,n\}$ we $\PP$-a.s. find
\begin{align}\label{ItoNeu1}
\langle P_nU_{\sigma,\UV,t}^{N,+}(x)|\beta_{\UV,t}^N(x)\rangle
&=I_1(n,t)+I_2(n,t)+I_3(n,t)+I_4(n,t), 
\end{align}
for all $t\ge0$ and $n\in\NN$, with 
\begin{align*}
I_1(n,t)&\coloneq\int_0^t\langle P_n\chi_{B_\sigma^c}v_{\UV,s}^N(x)|\beta_{\UV,s}^{N}(x)\rangle\Id s,
\\
I_2(n,t)&\coloneq
-\int_0^t\langle P_n\omega U_{\sigma,\UV,s}^{N,+}(x)|\beta_{\UV,s}^N(x)\rangle\Id s,
\\
I_3(n,t)&\coloneq-\int_0^t\langle P_nU_{\sigma,\UV,s}^{N,+}(x)|\psi\beta_{\UV,s}^N(x)\rangle\Id s,
\\
I_4(n,t)&\coloneq\int_{(0,t]\times\RR^{2N}}h_{n}(s,z)\Id \wt{N}_X(s,z),
\end{align*}
where the integrand of $I_4(n,t)$ is given by
\begin{align*}
h_{n}(s,z)&\coloneq \sum_{j=1}^N
\langle P_nU_{\sigma,\UV,s}^{N,+}(x)|\eul^{-\ii K\cdot X_{j,s-}^x}(\eul^{-\ii K\cdot z_j}-1)\beta_{\UV}\rangle,
\quad s\ge0,\,z\in\RR^{2N}.
\end{align*}
Next, we pass to the limit $n\to\infty$ for fixed $t\ge0$. Since $P_n\to\id$ strongly,
the expression on the left hand side of \eqref{ItoNeu1} converges pointwise on $\Omega$ to $c_{\sigma,\UV,t}^N(x)$.
By dominated convergence, \cref{def:wIRUV,def:wNsigmaUV,defEren} we further have, again pointwise on $\Omega$,
\begin{align*}
I_1(n,t)\xrightarrow{\;\;n\to\infty\;\;}&w_{\sigma,\UV,t}^N(x)+tNE_{\sigma,\UV}^{\ren},
\\
I_2(n,t)+I_3(n,t)\xrightarrow{\;\;n\to\infty\;\;}
&-\int_0^t\langle U_{\sigma,\UV,s}^{N,+}(x)|(\omega+\psi)\beta_{\UV,s}^N(x)\rangle\Id s
\\
&=-\tilde{u}_{\sigma,\UV,t}^N(x)-tNE_{\sigma,\UV}^{\ren}.
\end{align*}
Here we used \cref{defuUV} and $(\omega+\psi)\beta_{\UV}=v_{\UV}$ in the last step.
Finally, recall that $h_{\sigma,\UV}(x,\cdot)$, the integrand of the martingale
$m_{\sigma,\UV}^N(x)$, has been introduced prior to \cref{lemexpmomentm}.
Clearly, $h_{n}(s,z)\to h_{\sigma,\UV}(x,s,z)$, $n\to\infty$, pointwise on $\Omega$
and for all $s\ge0$ and $z\in\RR^{2N}$. Since $\UV$ is finite, we can employ
some $(g,N,m_{\bos})$-dependent constant (coming from \cref{tnormU} with $\ve=0$)
times $(|z|\|K\beta_{\UV}\|)\wedge\|2\beta\|$
as square-integrable dominating function to conclude that
$h_n\to h_{\sigma,\UV}(x,\cdot)$ in $L^2((0,t]\times\RR^{2N}\times\Omega,\Id s\otimes\nu_X\otimes\PP)$.
By means of It\^{o}'s isometry \cref{ItoIso} we deduce that
$I_4(n,t)\to m_{\sigma,\UV,t}^N(x)$, $n\to\infty$, in $L^2(\PP)$ and in particular $\PP$-a.s.
along a suitable subsequence.

All these remarks prove $\PP$-a.s. an identity that is obviously equivalent to the one in \eqref{ItoNeu0}
for fixed $t\ge0$. Since the processes on both sides in \eqref{ItoNeu0} are c\`{a}dl\`{a}g, they
consequently are indistinguishable.
\end{proof}

\subsection{Exponential moment bounds}\label{ssecexpmbca}

\noindent
Exponential moment bounds on the complex action
are the essential quantitative ingredient for the proofs of our lower bounds on the minimal energy.
The first one, stated in the next theorem, is good for small $g^2N$ and obviously very bad for large $g^2N$.

\begin{thm}\label{thmexpmonentu1}
There exist universal constants $b,c,c'>0$ such that, for all $\alpha>1$, $p>0$, $0<\UV\le\infty$,
$t\ge0$ and $x\in\RR^{2N}$,
\begin{align*}
&\EE\bigg[\sup_{s\in[0,t]}e^{pu_{0,\UV,s}^N(x)}\bigg]
\le b^N\left(\frac{\alpha}{\alpha-1}\right)^{1/2}\eul^{2\pi pN^2g^2/m_{\bos}}
\\
&\quad\cdot
\exp\left(tc'\cdot\frac{p^2g^4N^2(N-1)}{m_{\bos}}
\Big(1+\frac{m_{\p}}{m_{\bos}}\Big)
+tc\cdot\frac{\alpha p^2g^4N^3}{m_{\bos}}\cdot
\eul^{8\pi\alpha pg^2N/m_{\bos}}\right).
\end{align*}
\end{thm}

\begin{proof}
This follows from \cref{defcomplexactionsigmaUV}, \cref{bdcN}, \cref{lemexpmomentw,lemexpmomentm} with $\sigma=0$ and the Cauchy-Schwarz inequality.
(After using the Cauchy-Schwarz inequality, we have to apply \cref{lemexpmomentm,lemexpmomentm} with $2p$ put in place of $p$; hence the factor $8\pi$ in the exponent of 
$\eul^{8\pi\alpha pg^2N/m_{\bos}}$.)
\end{proof}

We infer another exponential moment bound on the complex action, exhibiting a
better behavior for large $g^2N$. For $N=1$ and $m_{\p}>0$ we also find a bound that
is uniform in $m_{\bos}>0$.

\begin{thm}\label{thmexpmonentu2}
There exist universal constants $b,c,c'>0$
such that, for all $\alpha>1$, $0<\UV\le\infty$,
$t\ge0$, $x\in\RR^{2N}$ and $p>0$ such that $pg^2N>m_{\bos}$,
\begin{align*}
\EE\bigg[\sup_{s\in[0,t]}e^{pu_{0,\UV,s}^N(x)}\bigg]
&\le b^N\left(\frac{\alpha}{\alpha-1}\right)^{1/2}
\exp\bigg(t\pi pg^2(2N^2-N)\ln\Big(\frac{pg^2N}{m_{\bos}}\Big)\bigg)
\\
&\quad\cdot\exp\Big(tc'(N-1)(m_{\p}+pg^2N)+tc\alpha pg^2N^2\eul^{8\pi\alpha}\Big).
\end{align*}
For $m_{\p}>0$, the argument of the first exponential can be replaced by
\begin{align*}
t2\pi pg^2N(N-1)\ln\Big(\frac{pg^2N}{m_{\bos}}\Big)
+t\pi pg^2N[\ln(pg^2N)\vee0]+t\frac{\pi pg^2N}{m_{\p}}[(pg^2N)\wedge1].
\end{align*}
\end{thm}

\begin{proof}
Thanks to our assumption $pg^2N>m_{\bos}$, we can fix $\sigma>0$ such that
\begin{align}\label{eq:sigma0}
\omega(\sigma) = (\sigma^2+m_{\bos}^2)^{1/2}&=pg^2N.
\end{align}
Let us further assume that $\UV>\sigma$. Then, by \eqref{defcomplexactionsigmaUV}, \cref{eq:u-decomp,lem:actionito}, 
we may $\PP$-a.s. split up the complex action as
\begin{align}\label{splitu}
u_{0,\UV,s}^N(x)&=\tilde{u}_{\sigma,s}^N(x)-c_{\sigma,\UV,s}^N(x)+w_{\sigma,\UV,s}^N(x)+m_{\sigma,\UV,s}^N(x),
\quad s\ge0.
\end{align}
The first two terms on the right hand side of \cref{splitu} can be estimated trivially using \eqref{bdcN} and
\begin{align}\nonumber
\tilde{u}_{\sigma,t}^N(x)
&\le g^2N^2\int_{B_{\sigma}}\int_0^t\int_0^s\frac{\eul^{-(s-r)\omega(k)}}{\omega(k)}\Id r\,\Id s\,\Id k-tNE_{\sigma}^{\ren}
\\\label{trivialbdu}
&\le tg^2N^2 \int_{B_{\sigma}}\frac{1}{\omega(k)^2}\Id k-tNE_{\sigma}^{\ren},\quad t\ge0.
\end{align}
Direct computation and our choice of $\sigma$ satisfying \cref{eq:sigma0} lead to
\begin{align}\label{rosa}
g^2(N^2-N)\int_{B_{\sigma}}\frac{1}{\omega(k)^2}\Id k
=2\pi g^2N(N-1) \ln\bigg(\frac{pg^2N}{m_{\bos}}\bigg).
\end{align}
Since $\omega(k)>|k|\ge\psi(k)$,
the remaining contribution proportional to $N$ can be estimated using
\begin{align}\nonumber
	g^2N\int_{B_\sigma}\frac{1}{\omega(k)^2}\Id k-NE_{\sigma}^{\ren}
	&=g^2N\int_{B_\sigma}\frac{\psi(k)}{\omega(k)^2(\omega(k)+\psi(k))}\Id k
	\\\label{eq:g2N-NE}
	&\le\frac{g^2N}{2}\int_{B_\sigma}\frac{1}{\omega(k)^2}\Id k
	=g^2N\pi \ln\bigg(\frac{pg^2N}{m_{\bos}}\bigg).
\end{align}
Combining \cref{trivialbdu,rosa,eq:g2N-NE} we find
\begin{align*}
\sup_{s\in[0,t]}p\tilde{u}_{\sigma,s}^N(x)&\le t\pi pg^2(2N^2-N)\ln\Big(\frac{pg^2N}{m_{\bos}}\Big),\quad t\ge0.
\end{align*}
Finally, we can bound the expectation of the product of
$\sup_{s\in[0,t]}\eul^{pw_{\sigma,\UV,s}^N(x)}$ and $\sup_{s\in[0,t]}\eul^{pm_{\sigma,\UV,s}^N(x)}$ 
by means of the Cauchy-Schwarz inequality and \cref{lemexpmomentm,lemexpmomentw} with the above choice for $\sigma$. 
Putting all these remarks together we obtain the asserted bound for general $m_{p}\ge0$.

If $m_{\p}>0$, then we can instead use the following estimates in \eqref{eq:g2N-NE}.
Observing that $\psi(k)(\psi(k)+2m_{\p})=|k|^2$, we find
\begin{align*}
\int_{B_{\sigma\wedge 1}}\frac{g^2\psi(k)}{\omega(k)^2(\omega(k)+\psi(k))}\Id k
&=\int_{B_{\sigma\wedge1}}\frac{|k|^2}{\omega(k)^2}\cdot\frac{g^2}{\omega(k)(\psi(k)+2m_{\p})+|k|^2}\Id k
\\
&\le \frac{g^2}{2m_{\p}}\int_{B_{\sigma\wedge1}}\frac{1}{\omega(k)}\Id k
\\
&\le\frac{\pi g^2}{m_{\p}}(\sigma\wedge1)\le\frac{\pi g^2}{m_{\p}}[(pg^2N)\wedge1],
\end{align*}
and, similar to above,
\begin{align*}
\int_{B_\sigma\setminus B_{\sigma\wedge 1}}\frac{g^2\psi(k)}{\omega(k)^2(\omega(k)+\psi(k))}\Id k
&\le \frac{g^2}{2}\int_{B_\sigma\setminus B_{\sigma\wedge 1}}\frac{1}{\omega(k)^2}\Id k
\le \pi g^2[\ln(pg^2N)\vee 0].
\end{align*}
Replacing the bound in \cref{eq:g2N-NE} by the latter two estimations proves the last statement.

Finally, if $\UV\le\sigma$, then we write $u_{0,\UV,s}^N(x)=\tilde{u}_{\UV,s}^N(x)$, $s\ge0$, $\PP$-a.s.,
and employ \eqref{trivialbdu} and the first equality in \eqref{eq:g2N-NE} both with $\UV$ put in place of $\sigma$.
Replacing $B_\UV$ by the larger ball $B_\sigma$, we then see that
$\tilde{u}_{\UV,t}^N(x)$ is again bounded from above by the expression in the last line of \eqref{trivialbdu}
as it stands there. It is now clear that all asserted bounds also hold for $\UV\le\sigma$.
\end{proof}

\subsection{Convergence properties of the complex action}\label{ssecconvca}

\noindent
Using the bounds derived in \cref{sec:basicproc}, we can easily deduce convergence statements for the complex action when removing the ultraviolet cutoff.

\begin{thm}\label{thm:uconv}
Let $\sigma,t\ge0$ and $p>0$. Then
\begin{align}\label{convu1}
\sup_{x\in\RR^{2N}}\EE\bigg[\sup_{s\in[0,t]}|u_{\sigma,\UV,s}^N(x)-u_{\sigma,\infty,s}^N(x)|^p\bigg]
&\xrightarrow{\;\;\UV\to\infty\;\;}0,
\\\label{convu2}
\sup_{x\in\RR^{2N}}\EE\bigg[\sup_{s\in[0,t]}\big|\eul^{u_{0,\UV,s}^N(x)}-\eul^{u_{0,\infty,s}^N(x)}\big|^p\bigg]
&\xrightarrow{\;\;\UV\to\infty\;\;}0.
\end{align}
\end{thm}

\begin{proof}
By \cref{defcomplexactionsigmaUV,eq:u-decomp}, for $\sigma<\UV<\infty$, we $\PP$-a.s. have
\begin{align*}
u_{\sigma,\infty,s}^N(x)-u_{\sigma,\UV,s}^N(x)&=w_{\UV,\infty,s}^N(x)-c_{\UV,\infty,s}^N(x)+m_{\UV,\infty,s}^N(x),
\quad s\ge0.
\end{align*}
Therefore, the convergence \eqref{convu1} follows directly from \cref{lemexpmomentm,lempmomentw}, 
\cref{bdcN} (all applied with $(\UV,\infty)$ put in place of $(\sigma,\UV)$)
and Jensen's and Minkowski's inequalities. Now \eqref{convu2} follows from the bound 
$|\eul^a-\eul^b|\le|a-b|\eul^a\eul^b$, $a,b\ge0$, H\"{o}lder's inequality,
the exponential moment bound in \cref{thmexpmonentu1,convu1}.
\end{proof}

The following conclusion easily follows and is a helpful technical ingredient for our proof of strong continuity of the 
Feynman--Kac semigroup.
\begin{cor}\label{cor:pathsactioncont}
	For any $\UV\in(0,\infty]$ and $x\in\RR^{2N}$, the paths 
	$[0,\infty)\ni t\mapsto u^N_{0,\UV,t}(x)$ are continuous $\PP$-a.s. 
\end{cor}
\begin{proof}
	The statement directly follows from \cref{def:action-process,lem:actionito} for $\UV<\infty$. 
	Further, by \cref{convu1}, there exists an increasing sequence of cutoffs $\UV_n\in (0,\infty)$, $n\in\NN$, 
	with $\lim_{n\to \infty}\UV_n=\infty$ such that $u_{0,\UV_n,s}^N(x)$ converges  $\PP$-a.s. to $u_{0,\infty,s}^N(x)$ uniformly in $s$ on any compact subinterval of $[0,\infty)$. This takes care of the case $\UV=\infty$.
\end{proof}


\section{Feynman--Kac integrands and semigroups}\label{sec:FKInt}
\noindent
In this \lcnamecref{sec:FKInt} we discuss the Feynman--Kac integrands and semigroups 
we introduced in \cref{ssecdefprelFK} and in particular we shall prove all assertions of \cref{prop:TUVonL2}.
In \cref{ssecpwbds} we first provide some pathwise bounds on the Fock space operator-valued
part of the Feynman--Kac integrands, which in conjuction with our results on the complex action are applied in the succeeding
subsections. There we prove weighted $L^p$ to $L^q$ bounds on and
convergence as $\UV\to\infty$ of semigroup members (\cref{subsec:LpLq,ssecconvTUV}), 
Markov and semigroup properties (\cref{subsec:Markov})
and finally strong continuity of the semigroups as well as equicontinuity in the range of semigroup elements
(\cref{sseccontprop}).


\subsection{Some pathwise bounds}\label{ssecpwbds}

\noindent
We start out by collecting some additional bounds on the Fock space operators $F_t(h)$ defined in \eqref{defFtg}.
Recall the definitions \cref{deftnorm,def:S} of the time dependent norms $\|\cdot\|_t$
and the series $\mc{S}$, respectively, as well as the bound \eqref{normFt} on the operator norm of $F_t(h)$.

\begin{lem}
Let $g,\tilde{g},h,\tilde{h}\in L^2(\RR^2)$ and $t>0$. Then
\begin{align}\label{normdiffFt}
\|F_t(g)-F_t(h)\|&\le4\|g-h\|_t\sup_{\tau\in[0,1]}\mc{S}\big(\|\tau g+(1-\tau)h\|_t\big),
\end{align}
and
\begin{align}\nonumber
&\|F_t(g)F_t(\tilde{g})^*-F_t(h)F_t(\tilde{h})^*\|
\\\label{LipFF}
&\le 4(\|g-h\|_t+\|\tilde{g}-\tilde{h}\|_t)
\max\big\{\mc{S}(\|g\|_t),\mc{S}(\|h\|_t)\big\}\max\big\{\mc{S}(\|\tilde{g}\|_t),\mc{S}(\|\tilde{h}\|_t)\big\}.
\end{align}
Furthermore, for all $\ve\in(0,1]$,
\begin{align}\label{FFbei0}
\|(F_t(g)F_t(\tilde{g})^*-\id)(1+\Id\Gamma(\omega))^{-\ve}\|
&\le 4(\|g\|_t+\|\tilde{g}\|_t)\mc{S}(\|g\|_t)\mc{S}(\|\tilde{g}\|_t)+(2t)^{\ve}.
\end{align}
Finally, if $\omega g\in L^2(\RR^{2N})$, then  $F_t(g)\Fock\subset\dom(\Id\Gamma(\omega))$ and
	\begin{align}\label{eq:FdG}
		\|\Id\Gamma(\omega)F_t(g)\| \le 4\|\omega g\|_t\mc S(\|g\|_t) + \frac1{\eul\tau}\mc S(\|g\|_{t-\tau}),\quad \tau\in(0,t)  .
	\end{align}
\end{lem}

\begin{proof}
In \cite[Lemma~17.4]{GueneysuMatteMoller.2017}, it is shown that, for all $f\in L^2(\RR^2)$,
\begin{align}\label{eq:Fderbound}
\sum_{n=0}^\infty\frac{1}{n!}\|\ad(f)\ad(g)^n\eul^{-t\Id\Gamma(\omega)}\|&\le 4\|f\|_t\mc{S}(\|g\|_t),
\end{align}
as well as that the Fr\'echet derivative of $F_t$ satisfies
\begin{align*}
	F_t'(g)f=\sum_{n=0}^\infty\frac{1}{n!}\ad(f)\ad(g)^n\eul^{-t\Id\Gamma(\omega)}.
\end{align*}
This directly implies \eqref{normdiffFt}, which in turn entails \eqref{LipFF}. 
Applying \eqref{LipFF} with $h=\tilde{h}=0$ 
and recalling that $F_{t}(0)=\eul^{-t\Id\Gamma(\omega)}$, we see that, to verify \eqref{FFbei0}, it only
remains to show that $\|(\id-\eul^{-2t\Id\Gamma(\omega)})(1+\Id\Gamma(\omega))^{-\ve}\|\le (2t)^{\ve}$.
The latter bound, however, follows directly from the spectral calculus, $\Id\Gamma(\omega)\ge 0$ 
and the simple bound $\sup_{s\ge0}(1-e^{-2ts})/(1+s)^\ve \le (2t)^\ve$.

Finally, \cref{eq:FdG} follows by using the commutation relations $[\Id\Gamma(\omega),\ad(g)]=\ad(\omega g)$
and $[\ad(g),\ad(\omega g)]=0$, which hold on the range of $\ad(g)^k\eul^{-t\Id\Gamma(\omega)}$ for every $k\in\NN_0$, 
the bounds \eqref{normFt}, \eqref{eq:Fderbound}  
and the simple estimate $\|\Id\Gamma(\omega)\eul^{-\tau \Id\Gamma(\omega)}\|\le 1/\eul \tau$.
\end{proof}
In the next lemma, we consider the Fock space operator-valued part of our Feynman--Kac integrands, i.e., 
\begin{align*}
		D_{\UV,t}(x)&\coloneq F_{t/2}(-U_{\UV,t}^{N,+}(x))F_{t/2}(-U^{N,-}_{\UV,t}(x))^*,
\end{align*}
where $t>0$, $x\in\RR^{2N}$ and $\UV\in(0,\infty]$. The aim is
to obtain norm bounds and study their convergence properties as $\UV$ goes to infinity or $t$ goes to zero. 

\begin{lem}\label{lempwbdD}
	Abbreviate $s_{*} \coloneq \mc S((12\pi/m_\bos)^{1/2}|g|N)$.
	Let $t>0$, $x\in\RR^{2N}$, $\ve\in(0,1/2)$ and $\UV\in(0,\infty]$. Then
\begin{align}
\label{bdFU}
&\|D_{\UV,t}(x)\| \le s_*^2, 
\\
&\label{bddGFU}\|\Id\Gamma(\omega)D_{\UV,t}(x)^*\| \le 4(6\pi)^{1/2} |g|N \omega(\Lambda)^{1/2}s_*^2 
+ \frac{4}{\eul t}s_*^2,\quad \mbox{if $\UV<\infty$,}
\\
\label{bddiffFFUU}
&\|D_{\UV,t}(x)-D_{\infty,t}(x)\|
\le  \frac{8(6\pi)^{1/2}|g|Ns_*^2}{\omega(\UV)^{1/2}},\qquad\:\qquad\mbox{if $\UV<\infty$,}
\\
\label{Dbei0}
&\|(D_{\UV,t}(x)^*-\id)(1+\Id\Gamma(\omega))^{-\ve}\|
 \le t^{\ve}\bigg(1+\frac{8(6\pi)^{1/2}|g|Ns_*^2}{(1-2\ve)^{1/2}m_{\bos}^{(1-2\ve)/2}}\bigg).
\end{align}
\end{lem}

\begin{proof}
From \cref{deftnorm,tnormU} with $\ve=\sigma=0$, we know that
\begin{align}\label{bdeasy}
2^{-1/2}\|U^{N,\pm}_{\UV,t}(x)\|_{t/4} \le
	\|U^{N,\pm}_{\UV,t}(x)\|_{t/2} \le (6\pi/m_{\bos})^{1/2}|g|N.
\end{align} 
Hence, \cref{bdFU} directly follows from \cref{normFt}. Similarly, \cref{bddGFU} follows from 
\cref{bdeasy} in conjunction with \cref{tnormwU,normFt,eq:FdG}, 
the latter with $(t/4,t/2)$ put in place of $(\tau,t)$.
Further, \cref{LipFF,bdeasy} imply that the left hand side of 
\eqref{bddiffFFUU} is less than or equal to
$4(\|U^{N,+}_{\UV,\infty,t}(x)\|_{t/2}+\|U^{N,-}_{\UV,\infty,t}(x)\|_{t/2})s_*^2$.
Again employing \eqref{tnormU} with $\ve=0$ and
$(\UV,\infty)$ put in place of $(\sigma,\UV)$, we see that \eqref{bddiffFFUU} holds true as well.
Finally, \eqref{Dbei0} follows from \crefnosort{FFbei0,bdeasy,tnormU} applied with $(\ve,\sigma)$ replaced by $(2\ve,0)$.
\end{proof}


\subsection{Weighted \texorpdfstring{$L^p$}{Lp} to \texorpdfstring{$L^q$}{Lq} bounds}\label{subsec:LpLq}

\noindent
We now want to derive bounds on the maps $T_{\UV,t}$ introduced in \cref{def:semigroup}
seen as operators from $L^p(\RR^{2N},\Fock)$ to $L^q(\RR^{2N},\Fock)$ for any $1< p\le q\le \infty$. 
We also cover Lipschitz-continuous weight functions, in the sense of the norm given below in \cref{1norm2}. 
Similar bounds for the free relativistic semigroup are derived in \cref{app:lplq} and applied in what follows.

As in \cref{app:lplq}, it would also be possible to include the case $p=1$ in the discussion of the families
$(T_{\UV,t})_{t\ge0}$, $\UV\in(0,\infty]$. This would, however, require additional arguments proceeding along the
lines in \cite{MatteMoller.2018} at several places in the article. As we have no actual need to include the case $p=1$
here, we refrained from doing so. (Only considering the case $p=2$ would not lead to any simplifications.)

For a start it makes sense to clarify the following product measurability:

\begin{lem}\label{rem:Wprodmeas}
Let $\UV\in(0,\infty]$ and $t\ge0$. Then $(W_{\UV,t}(x))(\gamma)$ is $\fr{B}(\RR^{2N})\otimes\fr{F}$-measurable 
and separably valued as a $\LO(\Fock)$-valued function of $(x,\gamma)\in\RR^{2N}\times\Omega$.
\end{lem}

\begin{proof}
Recall the definition \cref{defNVx} of the $\PP$-zero sets $\scr{N}_V(x)$ and 
the convention introduced in \cref{sssecKato}. In view of  the latter and
\begin{align*}
\scr{N}_V&\coloneq\big\{(x,\gamma)\in\RR^{2N}\times \Omega\,\big|\,\gamma\in\scr{N}_V(x)\big\}
\\
&=\bigcup_{n=1}^\infty \bigg\{(x,\gamma)\in\RR^{2N}\times \Omega\,\bigg|\:
\max_{\diamond\in\{+,-\}}\int_0^nV_\diamond(X_s^x(\gamma))\Id s=\infty\bigg\}
\in\fr{B}(\RR^{2N})\otimes\fr{F},
\end{align*}
the expression
$\int_0^tV(X_s^x(\gamma))\Id s$ is $\fr{B}(\RR^{2N})\otimes\fr{F}$-measurable in $(x,\gamma)$ for every $t\ge0$.
Analogous remarks apply to $w_{0,\infty,t}^N$ defined in \eqref{def:wNsigmaUV}.
The assertion now follows from the continuity of $h\mapsto F_s(h)$ (see, e.g., \cref{normdiffFt}), as well as from
\cref{remUpmcont,lem:canm,defn:complexphase0UV,defn:complexactioninfty}.
\end{proof}

\begin{rem}\label{remmomentbdW}
For all $p,t>0$ we abbreviate
\begin{align}\label{eq:defmho}
\mho_V(p,t)&\coloneq
\sup_{\UV\in(0,\infty]}\sup_{x\in\RR^{2N}}\EE\bigg[\sup_{s\in[0,t]}\|W_{\UV,s}(x)\|^p\bigg]^{1/p}.
\end{align}
Recall that $W_{\UV,s}(x)$ depends on $V$. Also observe that $\tilde{u}_{\UV,s}^N(x)$,
which appears in the definition \cref{defWUVtx0} of $W_{\UV,s}(x)$ with $\UV<\infty$, can
be replaced by $u_{0,\UV,s}^N(x)$ under the expectation in \cref{eq:defmho} because of \cref{lem:actionito}.
Thus, $\mho_0(p,t)^p$ is less than or equal to the product of $s_*^{2p}$ from \cref{lempwbdD}
with the minimum of the applicable alternative right hand sides 
in the exponential moment bounds of \cref{thmexpmonentu1,thmexpmonentu2}. 
Combining the lemma and the two theorems with H\"{o}lder's inequality and \eqref{V2} (to bound expectations of exponentials of the external potential $V$), we obtain
\begin{align*}
\mho_V(p,t)
&\le \eul^{c_{\theta,p,V_-}(1\vee t)}\mho_0(\theta p,t)^{1/\theta}\le \eul^{(c_{\theta,p,V_-}+c_*)(1\vee t)},
\end{align*}
for all $\theta\in(1,\infty)$.
Here, $c_*\in(0,\infty)$ depends on $p$, $\theta$,
the model parameters $m_{\p}$, $m_{\bos}$, $g$, $N$, the parameter $\alpha$
appearing in \cref{thmexpmonentu1,thmexpmonentu2}, but not on $t$.
Moreover, $c_{\theta,p,V_-}\in[0,\infty)$ depends solely on $p$, $\theta$ and $V_-$. If $V_-$ is bounded, then
we can also choose $\theta=1$ and $c_{\theta,p,V_-}=\|V_-\|_\infty$.
\end{rem}

For $1\le p\le q\le \infty$, we shall denote the norm on the space of all bounded operators from $L^p(\RR^{2N},\Fock)$
to $L^q(\RR^{2N},\Fock)$ by $\|\cdot\|_{p,q}$ in what follows.
In the next proposition we also set
\begin{align}\label{1norm2}
|x|_{1}&\coloneq|x_1|+\dots+|x_N|,\quad x=(x_1,\ldots,x_N)\in\RR^{2N},
\end{align}
where $|\cdot|$ is the Euclidean norm on $\RR^2$.

\begin{prop}\label{prop:LpLq}
In the case $m_{\p}>0$, we assume that $G:\RR^{2N}\to\RR$ satisfies 
\begin{align*}
|G(x)-G(x')|&\le L|x-x'|_1,\quad x,x'\in\RR^{2N},
\end{align*}
with $L\in[0,m_{\p})$. In the case $m_{\p}=0$ we set $G=L=0$. 
If $m_{\p}>L>0$, then we pick some $r>1$ such that $rL<m_{\p}$.
Otherwise we set $r=\infty$ so that $r'=1$.
Let $\UV\in(0,\infty]$, $p\in (1,\infty]$, $q\in[p,\infty]$ and $t>0$. Then the following holds:
\begin{enumerate}
\item[{\rm(i)}] 
The random variable
$\|W_{\UV,t}(x)^*(\eul^{-G}\Psi)(X_t^x)\|_{\Fock}$ is $\PP$-integrable for all $x\in\RR^{2N}$ and $\Psi\in L^p(\RR^{2N},\Fock)$.
\item[{\rm(ii)}] 
$\eul^{G}T_{\UV,t}\eul^{-G}$ maps $L^p(\RR^{2N},\Fock)$ continuously into $L^q(\RR^{2N},\Fock)$. 
\item[{\rm(iii)}] 
There exists $b_*>0$, depending solely on $m_{\p}$, $N$, $L$, $p$, $q$ and $r$, such that
\begin{align*}
\|\eul^GT_{\UV,t}\eul^{-G}\|_{p,q}
&\le b_*\bigg[\frac{1}{t^2}\vee\frac{m_{\p}-L}{t}\bigg]^{N(\frac{1}{p}-\frac{1}{q})}
\eul^{LNt}\mho_V(p'r',t).
\end{align*}
We can actually work with the supremum over $x\in\RR^{2N}$ instead of the essential supremum when $q=\infty$,
	i.e., we can use the convention
	\begin{align*}
	\|\eul^GT_{\UV,t}\eul^{-G}\|_{p,\infty}
	&=\sup_{\|\Psi\|_p\le 1}\sup_{x\in\RR^{2N}}\eul^{G(x)}\|(T_{\UV,t}\eul^{-G}\Psi)(x)\|_{\Fock}.
	\end{align*}
\end{enumerate}
\end{prop}

\begin{proof}
All assertions are direct consequences of \cref{rem:Wprodmeas,remmomentbdW,corwLpLqNW}.
\end{proof}

Our next lemma holds only for finite $\UV$. It is essential for our proof of the continuity results in \cref{thm:cont}.

\begin{lem}\label{prop:LpLqdG}
	Let $\UV\in(0,\infty)$, $p\in(1,\infty]$, $q\in[p,\infty]$ and $0<\tau_1\le\tau_2<\infty$.
	Then $T_{\UV,t}$ with $t>0$ maps
	$L^p(\RR^{2N},\Fock)$ into $L^q(\RR^{2N},\dom(\Id\Gamma(\omega)))$ and
	\begin{align*}
		\sup_{t\in[\tau_1,\tau_2]}\|\Id\Gamma(\omega)T_{\UV,t}\|_{p,q}&<\infty.
	\end{align*}
\end{lem}

\begin{proof}
In view of \cref{trivialbdu,bddGFU,V2} it is clear that 
\begin{align*}
\sup_{t\in[\tau_1,\tau_2]}\sup_{x\in\RR^{2N}}\EE[\|\Id\Gamma(\omega)W_{\UV,t}(x)^*\|^{\theta}]<\infty,
\quad \theta\in[1,\infty).
\end{align*}
In conjunction with \cref{corwLpLqNW} (there we choose $d=2$, $G=L=0$, $r'=1$) this implies all statements.
\end{proof}

\subsection{Convergence properties}\label{ssecconvTUV}

\noindent
Next, we derive the semigroup convergence when removing the ultraviolet cutoff.

\begin{prop}\label{prop:Tconv}
	Let $1<p\le q\le\infty$ and $0< \tau_1\le \tau_2<\infty$. Then
	\begin{align*}
		\sup_{t\in(0,\tau_1]} \|T_{\UV,t}-T_{\infty,t}\|_{p,p}\xrightarrow{\;\;\UV\to\infty\;\;} 0, \quad 
		\sup_{t\in[\tau_1,\tau_2]} \|T_{\UV,t}-T_{\infty,t}\|_{p,q} \xrightarrow{\;\;\UV\to\infty\;\;} 0.
	\end{align*}
	Again we can work with the actual supremum over $x\in\RR^{2N}$ when $q=0$, i.e., 
	in both convergence relations we can use the convention
	\begin{align*}
	\|T_{\UV,t}-T_{\infty,t}\|_{p,\infty}
	&=\sup_{\|\Psi\|_p\le 1}\sup_{x\in\RR^{2N}}\|(T_{\UV,t}\Psi)(x)-(T_{\infty,t}\Psi)(x)\|_{\Fock},\quad t>0.
	\end{align*}
\end{prop}
\begin{proof}
	Let $\theta\in[1,\infty)$.
	By Minkowski's inequality, i.e., the triangle inequality in $L^{\theta}(\PP)$, and \cref{lem:actionito}, we have
	\begin{align*}
		&\EE\bigg[\sup_{s\in[0,t]}\|W_{\UV,s}(x)-W_{\infty,s}(x)\|^{\theta}\bigg]^{1/\theta} 
		\\
		&\le\EE \bigg[\eul^{\int_0^t \theta V_-(X_s^x)\Id s}
		\sup_{s\in[0,t]} \big|\eul^{u_{0,\UV,s}^N(x)}-\eul^{u_{0,\infty,s}^N(x)}\big|^{\theta}
		\|D_{\UV,s}(x)\|^{\theta}  \bigg]^{1/\theta}
		\\
		&\quad  + \EE\bigg[\eul^{\int_0^t \theta V_-(X_s^x)\Id s}\sup_{s\in[0,t]}\eul^{\theta u_{0,\infty,s}^N(x)}
		\|D_{\UV,s}(x)-D_{\infty,s}(x)\|^{\theta}
		\bigg]^{1/\theta},
	\end{align*}
	for all $t>0$ and $x\in\RR^{2N}$.
	Here the first expression can be bounded by H\"{o}lder's inequality together with \cref{bdFU,thm:uconv} as well as \cref{V2}, to bound the factor involving the external potential $V$. 
	To treat the second one we apply \cref{thmexpmonentu1,bddiffFFUU,V2} instead.
	As a result we see that
	\begin{align}\label{locunivconvW}
	\sup_{x\in\RR^{2N}}\EE\bigg[\sup_{s\in[0,t]}\|W_{\UV,s}(x)-W_{\infty,s}(x)\|^{\theta}\bigg]
	\xrightarrow{\;\;\UV\to\infty\;\;}0,\quad t>0.
	\end{align}
	Now the assertion follows from \cref{corwLpLqNW} (again with $d=2$, $G=L=0$, $r'=1$).
\end{proof}

\subsection{Markov and semigroup properties}\label{subsec:Markov}

\noindent
Next, we derive a Markov property of our Feynman--Kac integrands that will entail semigroup properties
for the families $(T_{\UV,t})_{t\ge0}$. The Markov property will be a consequence
of a flow relation we prove first in \cref{lemMark1} below.

For these purposes it is convenient to observe the existence of certain canonical representatives of the 
processes $(W_{\UV,t}(x))_{t\ge0}$. With $D=D([0,\infty),\RR^{2N})$ denoting the set of all c\`adl\`ag functions 
from $[0,\infty)$ to $\RR^{2N}$, these are given by product measurable separably valued maps
\begin{align}\label{canrepW1}
\RR^{2N}\times D\ni(x,\gamma)\longmapsto W_{\UV,t}[x,\gamma]\in\LO(\Fock),
\end{align}
such that, for all $\UV\in(0,\infty]$ and $x\in\RR^{2N}$,
\begin{align}\label{canrepW2}
\text{$(W_{\UV,t}(x))_{t\ge0}$ and $(W_{\UV,t}[x,X_\bullet])_{t\ge0}$ are indistinguishable.}
\end{align}
Here $X_\bullet:\Omega\to D$ is the path map of $X$.

The canonical representatives are helpful for the following reasons: Firstly, for every $t\ge0$,
we obtain time-shifted versions 
$(W_{\UV,s}[x,X_{t+\bullet}-X_t])_{s\ge0}$ of $(W_{\UV,s}(x))_{s\ge0}$ simply by substitution of the path 
map of the time-shifted L\'{e}vy process. Moreover, since $X_\bullet$ and $X_{t+\bullet}-X_t$
have the same law, it is then obvious from \cref{canrepW2} that $W_{\UV,s}(x)$ and
$W_{\UV,s}[x,X_{t+\bullet}-X_t]$ have identical laws as well, for each $s\ge0$. 
Secondly, having the product measurable map \cref{canrepW1}
in whose arguments the independent random variables $X_t^x$ and $X_{t+\bullet}-X_t$ can be inserted,
permits to apply the ``useful rule'' (see, e.g., \cite[(6.8.14)]{HoffmannJorgensen.1994a})
for the computation of the conditional expectation $\EE^{\fr{F}_t}$ in the proof of the Markov property.

Let us give more details on \cref{canrepW1,canrepW2}: On $D$ we introduce the evaluation maps $\ev_t:D\to\RR^{2N}$, 
$\gamma\mapsto\gamma_t$, as well as the associated $\sigma$-algebra $\fr{D}'$ and filtration $(\fr{D}_t')_{t\ge0}$, i.e.,
\begin{align*}
	\fr{D}'&\coloneq \sigma\big(\ev_t :\,t\in[0,\infty)\big),\qquad
	\fr{D}_t'\coloneq \sigma\big(\ev_s:\,s\in[0,t]\big),\quad t\ge0.
\end{align*}
While the choice of our basic c\`{a}dl\`{a}g L\'{e}vy process $X$ is not unique, all possible choices have the same
law on $D$ that we denote by $\PP_{\mathrm{can}}'$. That is, $\PP_{\mathrm{can}}'=\PP\circ X_{\bullet}^{-1}$ 
where the path map $X_\bullet$ is as usual considered as a measurable map from $(\Omega,\fr{F})$ to $(D,\fr{D}')$.
The process $(\ev_t)_{t\ge0}$ on $(D,\fr{D}',\PP_{\mathrm{can}}')$ is the so-called
canonical representative of $X$ with paths in $D$.
Furthermore, we let $(D,\fr{D},(\fr{D}_t)_{t\ge0},\PP_{\mathrm{can}})$ denote the completion of 
the filtered probability space $(D,\fr{D}',(\fr{D}_t')_{t\ge0},\PP_{\mathrm{can}}')$.
Since $(\ev_t)_{t\ge0}$ is L\'{e}vy with c\`{a}dl\`{a}g paths, we know that $(\fr{D}_t)_{t\ge0}$ is
automatically right-continuous; see, e.g., \cite[Theorem~2.1.10]{Applebaum.2009}.
Thus, $(D,\fr{D},(\fr{D}_t)_{t\ge0},\PP_{\mathrm{can}})$ satisfies the usual hypotheses. Again we can write
$\PP_{\mathrm{can}}=\PP\circ X_{\bullet}^{-1}$ provided, however, that
$X_\bullet$ on the right hand side of this relation is now treated as a measurable map from $(\Omega,\fr{F})$ to $(D,\fr{D})$.

Let us now agree on the following notational convention for our stochastic processes 
parametrically depending on $x$: When the data
\begin{align}\label{candata}
(D,\fr{D},(\fr{D}_t)_{t\ge0},\PP_{\mathrm{can}},(\ev_t)_{t\ge0})
\end{align}
is put in place of $(\Omega,\fr{F},(\fr{F}_t)_{t\ge0},\PP,(X_t)_{t\ge0})$ in their constructions,
then we replace the argument $(x)$ of the processes by $[x,\cdot]$. For example,
we denote by $(m_{0,\infty,t}^N[x,\cdot])_{t\ge0}$, $x\in\RR^{2N}$,
the particular versions of the martingales we obtain by applying \cref{lem:canm} to the
data \cref{candata}, and these martingales are used to define $W_{\infty,t}[x,\cdot]$. 

The product-measurability of \cref{canrepW1} for all $\UV\in(0,\infty]$ and $t\ge0$ then follows from \cref{rem:Wprodmeas}
applied to the data \cref{candata}.
Manifestly, we further have transformation rules for the following processes which have all been constructed {\em pathwise}:
\begin{align}\nonumber
U_{\UV,t}^{N,\pm}(x)&=U_{\UV,t}^{N,\pm}[x,X_\bullet],\quad\text{for all $\UV\in(0,\infty]$;}
\\\label{tildeuuXcanindist}
\tilde{u}_{\UV,t}^N(x)&=\tilde{u}_{\UV,t}^N[x,X_\bullet],\quad
W_{\UV,t}(x)=W_{\UV,t}[x,X_\bullet], \quad\text{for finite $\UV\in(0,\infty)$;}
\\\nonumber
w_{0,\infty,t}^N(x)&=w_{0,\infty,t}^N[x,X_\bullet],\quad c_{0,\infty,t}^N(x)=c_{0,\infty,t}^N[x,X_\bullet].
\end{align} 
These relations hold on $\Omega$ and for all $t\ge0$ and $x\in\RR^{2N}$; the fourth one follows from the first three
and a similar relation involving $V$. They prove \cref{canrepW2} for finite $\UV$. When $\UV=\infty$, the complex actions
\begin{align*}
		u_{t}^N[x,\cdot]&\coloneq w_{0,\infty,t}^N[x,\cdot]-c_{0,\infty,t}^N[x,\cdot]+m_{0,\infty,t}^N[x,\cdot],
		\quad t\ge0,\,x\in\RR^{2N},
\end{align*}
contain the stochastic integrals $m_{0,\infty,t}^N[x,\cdot]$, whence \cref{canrepW2} becomes less obvious.

\begin{lem}\label{lem:uuindist}
$(u_t^N(x))_{t\ge0}$ and $(u_t^N[x,X_\bullet])_{t\ge 0}$ are indistinguishable. In particular, \cref{canrepW2}
holds for $\UV=\infty$ as well.
\end{lem}
\begin{proof}
	Denoting expectations with respect to $\PP_{\mathrm{can}}$ by $\EE_{\mathrm{can}}$, the first assertion follows from
	\begin{align*}
		\EE\bigg[\sup_{s\in[0,t]}|u_{s}^N(x)-u_s^N[x,X_\bullet]|^2\bigg]
		&=\lim_{\UV\to\infty}\EE\bigg[\sup_{s\in[0,t]}|\tilde{u}_{\UV,s}^N[x,X_\bullet]-u_s^N[x,X_\bullet]|^2\bigg]
		\\
		&=\lim_{\UV\to\infty}\EE_{\mathrm{can}}\bigg[
		\sup_{s\in[0,t]}|\tilde{u}_{\UV,s}^N[x,\cdot]-u_s^N[x,\cdot]|^2\bigg]
		\\
		&=0,\quad t\ge0.
	\end{align*}
	In the first equality we applied \eqref{convu1} (with $\sigma=0$ and filtered probability space
	$(\Omega,\fr{F},(\fr{F}_t)_{t\ge0},\PP)$) and took into account that $(u_{0,\UV,t}^N(x))_{t\ge0}$ and 
	$(\tilde{u}_{\UV,t}^N[x,X_\bullet])_{t\ge0}$ are indistinguishable for finite $\UV$ by \cref{lem:actionito,tildeuuXcanindist}. 
	The second equality follows from $\PP_{\mathrm{can}}=\PP\circ X_{\bullet}^{-1}$.
	In the last one we again applied \cref{lem:actionito,convu1}, now with filtered probability space
	$(D,\fr{D},(\fr{D}_t)_{t\ge0},\PP_{\mathrm{can}})$.
\end{proof}

We are now in a position to derive the flow relation, arguing similarly as in \cite[Proposition~5.5]{MatteMoller.2018}:

\begin{lem}\label{lemMark1}
Let $\UV\in(0,\infty]$, $t\ge0$ and $x\in\RR^{2N}$. Then, $\PP$-a.s.,
\begin{align}\label{Mark1}
W_{\UV,s}[X_t^x,X_{t+\bullet}-X_t]W_{\UV,t}[x,X_\bullet]&=W_{\UV,s+t}[x,X_\bullet],\quad s\ge0.
\end{align}
\end{lem}

\begin{proof}
To start with we consider only finite $\UV$. 
It will be convenient to verify the flow relation in this case by applying the
operator-valued processes to exponential vectors of the form \cref{def:expvec}. 
Repeatedly applying the following analogue of  \eqref{Wexpvec},
\begin{align*}
W_{\UV,\tau}[z,\gamma]\expv{h}
&=\eul^{\tilde{u}_{\UV,\tau}^N[z,\gamma]-\int_0^\tau V(z+\gamma_s)\Id s-\langle U^{N,-}_{\UV,\tau}[z,\gamma]|h\rangle}
\expv{\eul^{-\tau\omega}h-U_{\UV,\tau}^{N,+}[z,\gamma]},
\end{align*}
which holds for all $\tau\ge0$, $z\in\RR^{2N}$ and $\gamma\in D$, we see that the $\PP$-a.s. validity of 
\begin{align*}
W_{\UV,s}[X_t^x,X_{t+\bullet}-X_t]W_{\UV,t}[x,X_\bullet]\expv{h}&=W_{\UV,s+t}[x,X_\bullet]\expv{h},
\quad s\ge0,\,h\in L^2(\RR^2),
\end{align*}
is implied by the $\PP$-a.s. validity for all $s\ge0$ of the four relations
\begin{align}\label{MarkV}
\int_0^{s+t}V(X_r^x)\Id r&=\int_0^sV(z+X_{t+r}-X_t)\Id r\bigg|_{z=X_t^x}+\int_0^tV(X_r^x)\Id r,
\end{align}
as well as
\begin{align}\label{Mark2}
U^{N,-}_{\UV,s+t}[x,X_\bullet]&=\eul^{-t\omega}U_{\UV,s}^{N,-}[X_t^x,X_{t+\bullet}-X_t]+U^{N,-}_{\UV,t}[x,X_\bullet],
\\\label{Mark3}
U^{N,+}_{\UV,s+t}[x,X_\bullet]&=U_{\UV,s}^{N,+}[X_t^x,X_{t+\bullet}-X_t]+\eul^{-s\omega}U^{N,+}_{\UV,t}[x,X_\bullet],
\end{align}
and
\begin{align}\nonumber
\tilde{u}^N_{\UV,s+t}[x,X_\bullet]
&=\tilde{u}_{\UV,s}^N[X_t^x,X_{t+\bullet}-X_t]+\tilde{u}^N_{\UV,t}[x,X_\bullet]
\\\label{Mark4}
&\quad+\langle U_{\UV,s}^{N,-}[X_t^x,X_{t+\bullet}-X_t]|U^{N,+}_{\UV,t}[x,X_\bullet]\rangle.
\end{align}
Recall that the left and rightmost integrals in \eqref{MarkV} are set equal to zero
on the $\PP$-zero set $\scr{N}_V(x)$ introduced in \eqref{defNVx}, and the analogous convention
applies to the expression in the middle for each $z$. However, if $\gamma\in \Omega\setminus\scr{N}_V(x)$,
and $X_t^x(\gamma)=z$, then $V(z+X_{t+\bullet}(\gamma)-X_t(\gamma))$ is locally integrable. This implies that
\eqref{MarkV} holds for all $s\ge0$ at least on  $\Omega\setminus\scr{N}_V(x)$.
Straightforward computations further show that \eqref{Mark2} and \eqref{Mark3}
hold on $\Omega$ for all $s\ge0$ and $0<\UV\le \infty$. Again by direct computations based on \eqref{defuUV}
(virtually identical to the ones in the proof of \cite[Lemma~4.18]{MatteMoller.2018}), we can also verify \eqref{Mark4}
on $\Omega$ for all $s\ge0$ and $\UV\in(0,\infty)$. 
(Here we also use that the scalar product in the second line of \eqref{Mark4} is real.)
Since the set of exponential vectors is total in $\Fock$, 
these remarks prove the assertion for finite~$\UV$.

By approximation we can $\PP$-a.s. extend \eqref{Mark1} to $\UV=\infty$: Since $X_t^x$ and $X_{t+\bullet}-X_t$
are independent, the law of $X_t$ has the probability density $\rho_{0,m_{\p},t}^{\otimes_N}$ defined prior to
\cref{relLpLqN1} and $X_{t+\bullet}-X_t$ has the same law as $X_\bullet$, 
we infer from \cref{canrepW2,rem:Wprodmeas,locunivconvW} that
\begin{align*}
&\EE\bigg[\sup_{s\in[0,\tau]}\|W_{\UV,s}[X_t^x,X_{t+\bullet}-X_t]-W_{\infty,s}[X_t^x,X_{t+\bullet}-X_t]\|^p\bigg]
\\
&=\int_{\RR^{2N}}\EE\bigg[\sup_{s\in[0,\tau]}\|W_{\UV,s}[z,X_{t+\bullet}-X_t]-W_{\infty,s}[z,X_{t+\bullet}-X_t]\|^p\bigg]
\rho_{0,m_{\p},t}^{\otimes_N}(z-x)\Id z
\\
&\xrightarrow{\;\;\UV\to\infty\;\;}0,\quad p>1,\,\tau>0.
\end{align*} 
On account of this convergence relation and \eqref{locunivconvW} we find an increasing sequence of
cutoffs $\UV_n\in(0,\infty)$, $n\in\NN$, with $\UV_n\to\infty$, $n\to\infty$, such that, after substituting 
$\UV=\UV_n$ in \eqref{Mark1}, both sides convergence $\PP$-a.s. in operator norm and locally uniformly in $s\ge0$.
The limiting relation is \eqref{Mark1} with $\UV=\infty$.
\end{proof}

We are now in a position to derive the promised Markov property. We shall denote by $\EE^{\fr F_t}$ the conditional expectation with respect to $\fr F_t$ for $t\ge 0$.

\begin{thm}[Markov property]\label{thm:markov}
Let $\UV\in (0,\infty]$, $p\in(1,\infty]$, $s,t\ge0$, $x\in\RR^{2N}$ and $\Psi\in L^p(\RR^{2N})$. Then, $\PP$-a.s.,
\begin{align}\label{Markprop1}
\EE^{\fr{F}_t}[W_{\UV,s+t}[x,X_\bullet]^*\Psi(X_{s+t}^x)]&=W_{\UV,t}[x,X_\bullet]^*(T_{\UV,s}\Psi)(X_t^x).
\end{align}
\end{thm}

\begin{proof}
Let $x\in\RR^{2N}$. By \cref{prop:LpLq}(i) and \cref{canrepW2}, 
$W_{\UV,s+t}[x,X_\bullet]^*\Psi(X_{s+t}^x)$
is $\PP$-integrable, whence the left hand side of \eqref{Markprop1} is well-defined. 
Furthermore, we infer from \cref{rem:Wprodmeas} applied to the data \cref{candata}
that the function $\Theta$ defined by
\begin{align*}
\Theta(x',\phi',\gamma)\coloneq\langle \phi'|W_{\UV,s}[x',\gamma]^*\Psi(x'+\ev_s(\gamma))\rangle_{\Fock},
\quad x'\in\RR^{2N},\,\phi'\in\Fock,\,\gamma\in D,
\end{align*}
is $\fr{B}(\RR^{2N}\times\Fock)\otimes\fr{D}$-measurable.
Pick $x\in\RR^{2N}$ and $\phi\in\Fock$.
Then $(X_t^x,W_{\UV,t}[x,X_\bullet]\phi):\Omega\to\RR^{2N}\times\Fock$ is $\fr{F}_t$-measurable
by \cref{rem:Wadapt,canrepW2},
while the path map $X_{t+\bullet}-X_t:\Omega\to D$ is $\fr{F}_t$-independent.
On account of the first remark in this proof and \cref{lemMark1} we further know that
$\Theta(X_t^x,W_{\UV,t}[x,X_\bullet]\phi,X_{t+\bullet}-X_t)$ is $\PP$-integrable.
The ``useful rule'' for conditional expectations, cf. \cite[(6.8.14)]{HoffmannJorgensen.1994a}, thus implies the first relation in 
\begin{align*}
&\EE^{\fr{F}_t}[\Theta(X_t^x,W_{\UV,t}[x,X_\bullet]\phi,X_{t+\bullet}-X_t)]
\\
&
=\EE[\Theta(x',\phi',X_{t+\bullet}-X_t)]\Big|_{(x',\phi')=(X_t^x,W_{\UV,t}[x,X_\bullet]\phi)}
\\
&=\EE[\Theta(x',\phi',X_{\bullet})]\Big|_{(x',\phi')=(X_t^x,W_{\UV,t}[x,X_\bullet]\phi)},\quad\text{$\PP$-a.s.},
\end{align*} 
where the second equality holds because $X_{t+\bullet}-X_t$ and $X_\bullet$ have the same law.
The resulting identity is equivalent to the second one in 
\begin{align*}
&\langle\phi|\EE^{\fr{F}_t}[W_{\UV,s+t}[x,X_\bullet]^*\Psi(X_{s+t}^x)]\rangle_{\Fock}
\\
&=\EE^{\fr{F}_t}[\langle \phi|W_{\UV,t}[x,X_\bullet]^*W_{\UV,s}[X_t^x,X_{t+\bullet}-X_t]^*\Psi(X_{s+t}^x)\rangle_{\Fock}]
\\
&=\langle\phi|W_{\UV,t}[x,X_\bullet]^*(T_{\UV,s}\Psi)(X_t^x)\rangle_{\Fock},\quad\text{$\PP$-a.s.}
\end{align*}
Here the first equality follows from \eqref{Mark1} and the fact that the $\Fock$-valued conditional expectation
$\EE^{\fr{F}_t}$ can be interchanged with taking the scalar product with~$\phi$.
Applying this for every $\phi$ in a countable dense subset of $\Fock$ we arrive at \eqref{Markprop1}.
\end{proof}

Taking expectations on both sides of \eqref{Markprop1}, using $\EE\EE^{\fr{F}_t}=\EE$ and employing \cref{canrepW2}, 
we obtain the following corollary:

\begin{cor}\label{thm:semigroup}
Let $\UV\in(0,\infty]$ and $p\in(1,\infty]$. 
Then $(T_{\UV,t})_{t\ge0}$ is a semigroup of bounded operators on $L^p(\RR^{2N},\Fock)$
such that $\|T_{\UV,t}\|_{p,p}\le \eul^{c(1+t)}$, $t\ge0$, for some constant $c\in(0,\infty)$ whose dependence
on the model parameters and $p$ can be read off from \cref{remmomentbdW,prop:LpLq}.
\end{cor}

\subsection{Continuity properties}\label{sseccontprop}

\noindent
We end this section by deriving the strong continuity and regularizing property 
of the semigroup in the next two propositions.

\begin{prop}\label{lem:strongcont}
Let $p\in(1,\infty)$ and $\UV\in(0,\infty]$. Then
the semigroup $(T_{\UV,t})_{t\ge0}$ is strongly continuous on $L^p(\RR^{2N},\Fock)$.
\end{prop}

\begin{proof}
Due to the semigroup property and the bound
$\sup_{t\in[0,1]}\|T_{\UV,t}\|_{p,p}<\infty$ from \cref{thm:semigroup}, it suffices to show that
$T_{\UV,t}\Phi\to \Phi$, $t\downarrow0$,
for every $\Phi$ in a dense subset of $L^p(\RR^{2N},\Fock)$. Hence, we assume in the following that $\Phi$ is of the form
$\Phi=(1+\Id\Gamma(\omega))^{-1/4}\Psi$ for some $\Psi\in L^p(\RR^{2N},\Fock)$. 

Minkowski's inequality implies
$\|T_{\UV,t}\Phi-\Phi\|_p\le\mc{N}_1(t)+\mc{N}_2(t)+\mc{N}_3(t)$, for all $t>0$, 
where the norms $\mc{N}_j(t)$ are given by
\begin{align*}
\mc{N}_1(t)^p&\coloneq\int_{\RR^{2N}}\big\|\EE\big[W_{\UV,t}(x)^*(\Phi(X^x_t)-\Phi(x))\big]\big\|_{\Fock}^p\Id x,
\\
\mc{N}_2(t)^p&\coloneq\int_{\RR^{2N}}\big\|\EE\big[\big(\eul^{u_{0,\UV,t}^N(x)-\int_0^tV(X^x_s)\Id s}
-1\big)D_{\UV,t}(x)^*\Phi(x)\big]\big\|_{\Fock}^p\Id x,
\\
\mc{N}_3(t)^p&\coloneq\int_{\RR^{2N}}\big\|\EE\big[(D_{\UV,t}(x)^*-\id)
(1+\Id\Gamma(\omega))^{-1/4}\Psi(x)\big]\big\|_{\Fock}^p\Id x.
\end{align*}
Employing H\"{o}lder's inequality and setting $f(z)\coloneq\|\Phi(\cdot - z)-\Phi\|_p^p$, $z\in\RR^{2N}$, we find
\begin{align*}
\mc{N}_1(t)^p&\le \EE[f(X_t)]\sup_{s\in[0,1]}\sup_{x\in\RR^{2N}}\EE[\|W_{\UV,s}(x)\|^{p'}]^{p/p'},\quad t\in(0,1],
\end{align*}
where the double supremum is finite due to \cref{remmomentbdW}. Furthermore, we know that $f$ is bounded and continuous.
Since the laws of $X_t$, $t>0$, converge weakly to the Dirac measure concentrated in $0$, 
we deduce that $\EE[f(X_t)]\to f(0)=0$ as $t\downarrow0$.

Furthermore, by \eqref{bdFU}, 
\begin{align}\label{strcont1}
\mc{N}_2(t)^p&\le s_*^{2p}\int_{\RR^{2N}}\EE\big[\big|\eul^{u_{0,\UV,t}^N(x)-\int_0^tV(X^x_s)\Id s}
-1\big|\big]^p\|\Phi(x)\|_{\Fock}^p\Id x.
\end{align}
For fixed $x\in\RR^{2N}$, the integrand under the expectation in \eqref{strcont1} goes $\PP$-a.s. to $0$ as $t\downarrow0$
by \cref{cor:pathsactioncont}. It is dominated by
$G_x\coloneq1+\eul^{\int_0^1V_-(X^x_s)\Id s}\sup_{s\in[0,1]}\eul^{u_{0,\UV,s}^N(x)}$
as long as $t\in(0,1]$. The majorant $G_x$
is $\PP$-integrable by Theorem~\ref{thmexpmonentu1} and \cref{V2}, which in fact imply that
$\sup_{x\in\RR^{2N}}\EE[G_x]<\infty$. Applying the dominated convergence theorem for each $x$,
we first see that the expectation in \eqref{strcont1} goes to $0$ as $t\downarrow0$ pointwise in $x\in\RR^{2N}$.
Since $\|\Phi(\cdot)\|_{\Fock}^p$ is integrable, we can invoke dominated convergence a second time to conclude
that $\mc{N}_2(t)\to0$ as $t\downarrow0$. By virtue of \eqref{Dbei0} and the integrability of
$\|\Psi(\cdot)\|_{\Fock}^p$, it is finally clear that  $\mc{N}_3(t)\to0$ as $t\downarrow0$. 
\end{proof}
The proof of the regularizing property is based on a similar one due to Carmona \cite[Proposition~3.3]{Carmona.1979}; 
see also \cite[Theorem~4.1]{BroderixHundertmarkLeschke.2000} and \cite[Theorem~8.1]{Matte.2016}.
\begin{prop}\label{thm:cont}
	Let $p\in(1,\infty]$, $\UV\in(0,\infty]$ and $0<\tau_1\le\tau_2<\infty$. Then the family
	\begin{align*}
		\left\{ T_{\UV,t}\Psi\middle| \Psi\in L^p(\RR^{2N},\Fock),\|\Psi\|_p\le 1,t\in[\tau_1,\tau_2]\right\}
	\end{align*}
	is uniformly equicontinuous on every compact subset of $\RR^{2N}$. Further, if $V\in \mc K_X$, then uniform equicontinuity holds on all of $\RR^{2N}$.
\end{prop}

\begin{proof}
By the uniform convergence as $\UV\to\infty$ from \cref{prop:Tconv} (recall the last remark in that proposition) 
it suffices to prove the assertion when $\UV$ is finite, which we assume
in the rest of this proof. We write
	\begin{align*}
		(P_t \Psi)(x) = \EE\left[ \Psi(X_t^x) \right] = \int_{\RR^{2N}} \rho_{0,m_\p,t}^{\otimes_N}(x-y)\Psi(y) \Id y
		,\quad \Psi\in L^p(\RR^{2N},\Fock),
	\end{align*}
	where the probability density function $\rho_{0,m_{\p},t}^{\otimes_N}$ is defined as in \cref{app:lplq} with $d=2$. 
	For given $\tau>0$, it is well-known that $P_\tau M$ is uniformly equicontinuous on $\RR^{2N}$ 
	for any bounded family $M\subset L^\infty(\RR^{2N},\Fock)$. In fact,
	\begin{align*}
		\|P_\tau f(x)-P_\tau f(y)\| 
		&\le \|\rho_{0,m_\p,\tau}^{\otimes_N}(x-y-\cdot)-\rho_{0,m_\p,\tau}^{\otimes_N}\|_1\|f\|_\infty,
		\quad x,y\in\RR^{2N},\, f\in M,
	\end{align*}
	and $x\mapsto\|h(x-\cdot)-h\|_1$ is continuous for any $h\in L^1(\RR^{2N})$.
	Hence, since $\|T_{\UV,t}\|_{p,\infty}$ is bounded uniformly in $t\in[\tau_1/2,\tau_2]$ by \cref{prop:LpLq},
	the set
	\begin{align*}
		\left\{ P_\tau T_{\UV,t-\tau}\Psi\middle| \Psi\in L^p(\RR^{2N},\Fock),\|\Psi\|_p\le 1,t\in[\tau_1,\tau_2]\right\}
	\end{align*}
	is uniformly equicontinuous on $\RR^{2N}$ for every $\tau\in(0,\tau_1/2)$. 
	In view of the semigroup property from \cref{thm:semigroup}, 
	it therefore suffices to prove
	\begin{align}\label{eq:uniformconv}
		\lim_{\tau\downarrow0}\sup_{t\in[\tau_1,\tau_2]}\sup_{\|\Psi\|_p\le 1}\sup_{x\in K }
		\|((P_\tau-T_{\UV,\tau})T_{\UV,t-\tau}\Psi)(x)\|_{\Fock} = 0, 
	\end{align}
	where $K\subset \RR^{2N}$ is compact (or $K=\RR^{2N}$ in case $V\in\mc K_X$).
	
	However, by \cref{prop:LpLq,prop:LpLqdG} and since $\UV$ is assumed to be finite,
	\begin{align*}
		\mc M &\coloneq \big\{(\id + \Id\Gamma(\omega))^{1/4}T_{\UV,t-\tau}\Psi\big|
		\,\|\Psi\|_p\le 1,\,t\in[\tau_1,\tau_2],\,\tau\in[0,\tau_1/2]\big\}
	\end{align*}
		is bounded in $L^\infty(\RR^{2N},\Fock)$.
	Further, for $\Phi\in L^\infty(\RR^{2N},\Fock)$, $\tau>0$ and $x\in\RR^{2N}$, 
	\begin{align}\nonumber
		&\big\|\big((P_\tau-T_{\UV,\tau})(\id+\Id \Gamma(\omega))^{-1/4}\Phi\big)(x)\big\|_{\Fock}
		\\\nonumber
		& =
		\left\|\EE\left[ (\id-W_{\UV,\tau}(x)^*) (\id+\Id \Gamma(\omega))^{-1/4}\Phi (X_\tau^x)  \right]\right\|_{\Fock}
		\\ \nonumber
		& \le \EE\left[|1-\eul^{-\int_0^\tau V(X_s^x) \Id x}|\eul^{\tilde{u}_{\UV,\tau}^N(x)}
		\|D_{\UV,\tau}(x)\|\right] \|\Phi\|_\infty
		\\\nonumber
		&\quad +\EE\left[|\eul^{\tilde{u}_{\UV,\tau}^N(x)}-1|
		\|D_{\UV,\tau}(x)\|\right] \|\Phi\|_\infty
		\\\label{eq:uniformconv2}
			& \quad + \EE\left[ \|(\id-D_{\UV,\tau}(x)^*)(\id+\Id\Gamma(\omega))^{-1/4}\|  \right]\|\Phi\|_\infty.
	\end{align}
	Since $\UV<\infty$ by the present assumption, a trivial bound analogous to \eqref{trivialbdu} entails
	$|\tilde{u}_{\UV,\tau}^N(x)|\le2\pi\tau g^2N^2\ln(\omega(\UV)/m_{\bos})+\tau NE_{\UV}^{\ren}$. 
	Therefore, \eqref{eq:uniformconv} follows from the above remark on $\mc{M}$ and
	\eqref{eq:uniformconv2} in conjunction with \eqref{bdFU}, \eqref{Dbei0} and  \cref{eq:boundexpVcomplete,eq:boundexpVminus1}. 
	This concludes the proof.
\end{proof}
%


\section{Upper bound on the minimal energy}\label{sec:upbound}

\noindent
In this last \lcnamecref{sec:upbound}, we prove the upper bound on the minimal energy of the translation invariant Hamiltonian stated in \cref{thm:upbound}.

\subsection{A well-known trial function argument}\label{ssectrial}

In \cref{lemtrialfct} we carry through a trial function argument yielding an upper bound on the
minimal energy for finite $\UV$, which in combination with our exponential moment
bounds on the complex action will allow us to prove \cref{thm:upbound} later on. The procedure we follow in the proof
of  \cref{lemtrialfct}.
is well-known and leads to the minimization of the Pekar functional in the non-relativistic case in three dimensions;
see, e.g., \cite[\textsection8.2]{MatteMoller.2018} and the references given there.
Here we encounter relativistic analogues of Pekar's functional (see the last two displayed relations in the proof
of Lemma~\ref{lemtrialfct}), parametrically depending on $m_{\bos}/g^2N$. 
For our purposes it will, however, be sufficient to simply insert a family of Gaussians into these functionals.

In \cref{ssectrial} we do not require the boson mass to be strictly positive, and we only consider finite $\UV$. 
We pick a core $\scr{D}$ for $\Id\Gamma(1+\omega)$ and set
\begin{align*}
\scr{C}&\coloneq\mathrm{span}\big\{f\phi\,\big|\:f\in\scr{S}(\RR^{2N}),\,\phi\in\scr{D}\},
\\
(H_{\UV}'\Psi)(x)&\coloneq(\psi_X(-\ii\nabla)\Psi)(x)+\Id\Gamma(\omega)\Psi(x)
+\sum_{j=1}^N\vp(\eul^{-\ii K\cdot x_j}v_\UV)\Psi(x),\quad \Psi\in\scr{C},
\end{align*}
with $\scr{S}(\RR^{2N})$ denoting the set of Schwartz functions on $\RR^{2N}$.
Again $\psi_X(-\ii\nabla)$ is defined via Fourier tranformation analogously to \eqref{def:psiXnabla}.
Since $\omega^{-1/2}v_\UV\notin L^2(\RR^2)$ for $m_{\bos}=0$, we cannot rely on the
Kato-Rellich theorem to ensure essential selfadjointness of $H_{\UV}'$ in this case.
We can, however, still apply Nelson's commutator theorem: 

\begin{lem}\label{lemesaNC}
If $m_{\bos}\ge0$ and $\UV\in(0,\infty)$, then $H_{\UV}'$ is essentially selfadjoint on $\scr{C}$.
\end{lem}

\begin{proof}
Since $\|\vp(h)\phi\|\le 2\|h\|\|(1+\Id\Gamma(1+\omega))^{1/2}\phi\|$, $h\in L^2(\RR^2)$, $\phi\in\scr{D}$,
well-known and simple estimations reveal that, for sufficiently large $C\in(0,\infty)$, the operator
$K'\coloneq H_{\UV}'+\Id\Gamma(1)+C$
is essentially selfadjoint on $\dom(K')\coloneq\scr{C}$ by the Kato-Rellich theorem, 
its selfadjoint closure is lower bounded by $1$ and
$\|H_{\UV}'\Psi\|\le c\|K'\Psi\|$, $\Psi\in\scr{C}$, for some $c\in(0,\infty)$. Furthermore,
\begin{align*}
2\Im \langle H_{\UV}'\Psi|K'\Psi\rangle&=\sum_{j=1}^N\int_{\RR^{2N}}
\langle\Psi(x)|\vp(\ii\eul^{-\ii K\cdot x_j}v_{\UV})\Psi(x)\rangle_{\Fock}\Id x,
\end{align*}
by the commutation relation between field operators and $\Id\Gamma(1)$;
see, e.g., \cite{Arai.2018}. This implies $|\Im \langle H_{\UV}'\Psi|K'\Psi\rangle|\le d\langle\Psi|K'\Psi\rangle$,
$\Psi\in\scr{C}$, for some $d\in(0,\infty)$.
 The assertion now follows from Nelson's commutator theorem
\cite[Theorem~X.37]{ReedSimon.1975} applied with the selfadjoint closure of $K'$ as test operator.
\end{proof}

In what follows we shall denote the unique selfadjoint extension of $H_\UV'$ again by the same symbol.

\begin{lem}\label{lemtrialfct}
Let $m_{\bos}\ge0$, $0\le \sigma<\UV<\infty$ and $s>0$.
Assume in addition that $\sigma>0$ when $m_{\bos}=0$. Then
\begin{align}\label{ubGauss1}
\inf\sigma(H_{\UV}')
&\le \frac{\pi ^{1/2}s^{1/2}g^2N^2}{2^{1/2}}
- g^2N^2\int_{B_{\UV}\setminus B_{\sigma}}
\frac{\eul^{-|k|^2/4sg^4N^2}}{|k|^2+m_{\bos}^2}\Id k.
\end{align}
\end{lem}

\begin{proof}
Let $f:\RR^2\to\CC$ be a Schwartz function, normalized in $L^2(\RR^2)$, and let $h\in L^2(\RR^2)$
be such that $\supp(h)\subset B_{\UV}$. We set (recall the definition \eqref{def:expvec} of $\expv{h}$)
\begin{align*}
f_N(x)&\coloneq \prod_{j=1}^Nf(x_j),\quad x\in\RR^{2N},\quad\text{and}\quad
\zeta(h)\coloneq\eul^{-\|h\|^2/2}\expv{h},
\end{align*}
so that $f_N$ and $\zeta(h)$ are normalized in $L^2(\RR^{2N})$ and $\Fock$, respectively.
Direct computation shows that
\begin{align*}
&\langle f_N\zeta(h)|H_{\UV}' f_N\zeta(h)\rangle
\\
&=N\int_{\RR^2}\psi(\eta)|\hat{f}(\eta)|^2\Id\eta+\|\omega^{1/2}h\|^2
+2N\int_{\RR^2}\Re\langle\eul^{-\ii K\cdot y}v_\UV|h\rangle |f(y)|^2\Id y.
\end{align*}
Setting $\rho\coloneq|f|^2$ and choosing 
$h = -2\pi gN\omega^{-3/2}\chi_{B_\UV\setminus B_\sigma}\wh{\rho}$, we obtain
\begin{align*}
\inf\sigma(H_{\UV}')&\le N\int_{\RR^2}\psi(\eta)|\hat{f}(\eta)|^2\Id\eta
-g^2N^2\int_{B_\UV\setminus B_\sigma}\frac{|2\pi\wh{\rho}(k)|^2}{\omega(k)^2}\Id k.
\end{align*}
Replacing $f$ by its unitarily scaled version $f_a(y)\coloneq af(ay)$, $y\in\RR^2$, $a>0$, writing
$\rho_a\coloneq |f_a|^2$ and taking $\wh{f_a}(\eta)=\hat{f}(\eta/a)/a$ and $\wh{\rho_a}(k)=\wh{\rho}(k/a)$
into account, we deduce that
\begin{align*}
\inf\sigma(H_{\UV}')
&\le N\int_{\RR^2}\psi(a\eta)|\hat{f}(\eta)|^2\Id\eta
-g^2N^2\int_{B_{\UV}\setminus B_{\sigma}}\frac{|2\pi\wh{\rho}(k/a)|^2}{|k|^2+m_{\bos}^2}\Id k.
\end{align*}
Estimating $\psi(a\eta)\le a|\eta|$ and choosing $a=g^2N$ we find
\begin{align*}
\inf\sigma(H_{\UV}')&\le g^2N^2\int_{\RR^2}|\eta||\hat{f}(\eta)|^2\Id\eta
-g^2N^2\int_{B_{\UV}\setminus B_{\sigma}}\frac{|2\pi\wh{\rho}(k/g^2N)|^2}{|k|^2+m_{\bos}^2}\Id k.
\end{align*}
Since $2\pi\wh{\rho}=\hat{f}*\ol{\hat{f}(-\,\cdot\,)}$, it is convenient to choose
$\hat{f}(k)\coloneq (2\pi s)^{-1/2}\eul^{-|k|^2/4s}$, a Gaussian that is normalized in $L^2(\RR^2)$.
Then $2\pi\wh{\rho}(k)=(\hat{f}*\hat{f})(k)=\eul^{-|k|^2/8s}$. Taking also 
$\int_{\RR^2}|\eta|\eul^{-|\eta|^2/2s}\Id\eta/2\pi s=(\pi s/2)^{1/2}$ 
into account we arrive at \eqref{ubGauss1}.
\end{proof}

\begin{cor}\label{corubGauss1}
If $m_{\bos}=0$ and $\UV\in(0,\infty)$, then $H_{\UV}'$ is unbounded from below.
\end{cor}

\begin{proof}
In the case $m_{\bos}=0$, the right hand side of \eqref{ubGauss1} goes to $-\infty$ as $\sigma\downarrow0$.
\end{proof}

\subsection{Derivation of the upper bound}
We now add the renormalization energy to both sides of \cref{ubGauss1}
and derive bounds on the so-obtained right hand side. The  estimates found in this way
then allow us to prove \cref{thm:upbound}. 
Throughout the remainder of this section, we will again work under the assumption $m_\bos>0$.
\begin{cor}\label{corubGauss2}
Assume that $g^2N>m_{\bos}>0$ and choose $\UV_{g,N}>0$ such that $\UV_{g,N}^2+m_{\bos}^2=g^4N^2$. 
Let $s>0$. Then, for every $\ve\in(0,1)$, there exists $c_*>1$, depending only on $\ve$, $m_{\p}$ and $m_{\bos}\vee1$,
such that
\begin{align}\nonumber
\inf\sigma(H_{\UV_{g,N}}'+NE_{\UV_{g,N}}^{\ren})
&\le \frac{\pi ^{1/2}s^{1/2}g^2N^2}{2^{1/2}}-\frac{2\pi}{\eul^{1/4s}}g^2N(N-1)\ln\bigg(\frac{g^2N}{m_{\bos}}\bigg)
\\\nonumber
&\quad-\frac{\pi g^2N}{\eul^{1/4s}}\cdot\frac{2-2\ve}{2-\ve}\cdot[0\vee\ln(c_*^{-1}g^2N)]+\frac{\pi}{4s}g^2N
\\\label{ubGauss2}
&\quad-\chi_{\{0\}}(m_{\p})\frac{\pi g^2N}{\eul^{1/4s}}\ln\bigg(\frac{1\wedge(g^2N)}{\eul m_{\bos}}\bigg).
\end{align}
\end{cor}

\begin{proof}
We shall apply \eqref{ubGauss1} with $\sigma=0$ and $\UV=\UV_{g,N}$.
Observing that $\UV_{g,N}/g^2N<1$, we can estimate $N(N-1)$ copies of the integral on the right hand side
of \eqref{ubGauss1} as
\begin{align*}
\int_{B_{\UV_{g,N}}}
\frac{\eul^{-|k|^2/4sg^4N^2}}{|k|^2+m_{\bos}^2}\Id k
&\ge \frac{\pi}{\eul^{1/4s}}\int_0^{\UV_{g,N}}\frac{2r}{r^2+m_{\bos}^2}\Id r
=\frac{2\pi}{\eul^{1/4s}}\ln\bigg(\frac{g^2N}{m_{\bos}}\bigg).
\end{align*}
To deal with the $N$ remaining copies and the term $NE_{\UV_{g,N}}^{\ren}$, we use
\begin{align*}
&g^2N\int_{B_{\UV_{g,N}}}
\frac{\eul^{-|k|^2/4sg^4N^2}}{|k|^2+m_{\bos}^2}\Id k-NE_{\UV_{g,N}}^{\ren}
\\
&=g^2N\int_{B_{\UV_{g,N}}}
\frac{\eul^{-|k|^2/4sg^4N^2}\psi(k)}{\omega(k)^2(\omega(k)+\psi(k))}\Id k
-g^2N\int_{B_{\UV_{g,N}}}\frac{1-\eul^{-|k|^2/4sg^4N^2}}{\omega(k)(\omega(k)+\psi(k))}\Id k
\\
&\ge \frac{g^2N}{\eul^{1/4s}}\int_{B_{\UV_{g,N}}}
\frac{\psi(k)}{\omega(k)^2(\omega(k)+\psi(k))}\Id k
-\frac{1}{4sg^2N}\int_{B_{\UV_{g,N}}}\frac{|k|^2}{\omega(k)(\omega(k)+\psi(k))}\Id k.
\end{align*}
Since $\omega(k)\le(|k|^2+(1\vee m_{\bos})^2)^{1/2}$,
we find some $R_{*}\ge1$ depending only on $\ve$, $m_{\p}$ and $1\vee m_{\bos}$ such that 
$\psi(k)\ge (1-\ve)\omega(k)$ as soon as $|k|\ge R_{*}$. 
Since $(0,\infty)\ni a\mapsto a/(\omega(k)+a)$ is increasing and we always have the bound $\omega(k)\ge\psi(k)$,
we can thus continue the above estimation as follows,
\begin{align*}
&g^2N\int_{B_{\UV_{g,N}}}
\frac{\eul^{-|k|^2/4sg^4N^2}}{|k|^2+m_{\bos}^2}\Id k-NE_{\UV_{g,N}}^{\ren}
\\
&\ge \frac{g^2N}{2\eul^{1/4s}}\int_{B_{R_*\wedge\UV_{g,N}}}
\frac{\psi(k)}{\omega(k)^3}\Id k
+\frac{g^2N}{\eul^{1/4s}}\cdot\frac{1-\ve}{2-\ve}\int_{B_{\UV_{g,N}}\setminus B_{R_*}}\frac{1}{\omega(k)^2}\Id k
-\frac{\mathrm{vol}(B_{\UV_{g,N}})}{4sg^2N}.
\end{align*}
In the case $m_{\p}>0$, we simply estimate the integral over $B_{R_*\wedge\UV_{g,N}}$ from below by $0$.
For $m_{\p}=0$ it is equal to the expression in the first line of
\begin{align*}
&\int_{B_{R_*\wedge\UV_{g,N}}}\frac{|k|}{\omega(k)^3}\Id k
\\
&\ge \int_{B_{R_*\wedge\UV_{g,N}}}\frac{1}{\omega(k)^2}\Id k
-\int_{B_{R_*\wedge\UV_{g,N}}}\frac{m_{\bos}}{\omega(k)^3}\Id k
\\
&=\pi\ln\bigg(\frac{(R_*\wedge\UV_{g,N})^2+m_{\bos}^2}{m_{\bos}^2}\bigg)
-2\pi m_{\bos}\bigg(\frac{1}{m_{\bos}}-\frac{1}{\sqrt{(R_*\wedge\UV_{g,N})^2+m_{\bos}^2}}\bigg),
\end{align*}
where we used $|k|\ge\omega(k)-m_{\bos}$.
If $R_*\wedge\UV_{g,N}=\UV_{g,N}$, then $(R_*\wedge\UV_{g,N})^2+m_{\bos}^2=g^4N^2$.
Otherwise, $(R_*\wedge\UV_{g,N})^2+m_{\bos}^2\ge1+m_{\bos}^2>1$ because $R_*\ge1$.
Hence, the term containing the logarithm in the previous bound is $\ge2\pi\ln([1\wedge(g^2N)]/m_{\bos})$,
while the second term is obviously $\ge -2\pi$.
Finally, we calculate
\begin{align*}
\int_{B_{\UV_{g,N}}\setminus B_{R_*}}\frac{1}{\omega(k)^2}\Id k
&=\pi\ln\bigg(\frac{g^4N^2}{R_*^2+m_{\bos}^2}\bigg)\chi_{(R_*^2+m_{\bos}^2,\infty)}(g^4N^2),
\end{align*}
where we obtain a lower bound on the right hand side by putting $m_{\bos}\vee1$ in place of $m_{\bos}$ in it.
Putting all remarks above together we arrive at \eqref{ubGauss2}.
\end{proof}
\begin{proof}[\textbf{Proof of \cref{thm:upbound}}]
In \cref{corubGauss2} we proved an upper bound on the minimal energy of the relativistic Nelson operator
$H_\sigma$ with an ultraviolet cutoff $\sigma\coloneq\UV_{g,N}$ 
chosen such that $\sigma^2+m_{\bos}^2=g^4N^2$. This upper
bound already contains all leading contributions to the upper bound for the renormalized operator $H$
asserted in \cref{thm:upbound} (when $\ve$ and $s$ are chosen appropriately as in the end of this proof).
Hence, what remains to do, is to bound the minimal energy of $H_{\sigma}$ from below by the one of
$H$ modulo error terms that won't harm these leading contributions. By means of \cref{corforspecu}(ii)
and its analogue for $H_\sigma$
we shall translate this task to the comparison of integrals involving exponentials of complex actions,
which can be done thanks to our exponential moment bounds of \cref{sec:basicproc}.

Besides $\sigma>0$ satisfying $\sigma^2+m_{\bos}^2=g^4N^2$
we also pick some measurable $f:\RR^{2N}\to\RR$ with $0\le f\le1$ and $\int_{\RR^{2N}}f(x)\Id x=1$. 
Furthermore, we pick some $p\in(1,\infty)$ and let $p'$ denote its conjugated exponent.
We again employ the splitting \eqref{splitu} of $u_t^N(x)=u_{0,\infty,t}^N(x)$ with our new, 
$p$-independent choice of $\sigma$, however.
Applying H\"{o}lder's inequality and afterwards \eqref{bdcN}, Lemma~\ref{lemexpmomentm}(iii)
and Lemma~\ref{lemexpmomentw} (all with $\UV=\infty$ and $\sigma$ as chosen above) we find
\begin{align*}
\EE\big[\eul^{\tilde{u}_{\sigma,t}^N(x)}f(X_t^x)\big]
&\le\|\eul^{c_{\sigma,\infty,t}^N(x)}\|_{L^\infty(\PP)}
\EE\big[\eul^{-2pw_{\sigma,\infty,t}^N(x)}f(X_t^x)\big]^{1/2p}
\\
&\quad\cdot
\EE\big[\eul^{-2pm_{\sigma,\infty,t}^N(x)}f(X_t^x)\big]^{1/2p}
\EE\big[\eul^{p'u_{t}^N(x)}f(X_t^x)\big]^{1/p'}
\\
&\le 
b^{N}c_\alpha^{1/2p}\eul^{tc'p(N-1)(m_{\p}+g^2N)+tc\alpha p g^2N^2\eul^{8\pi\alpha p}}
\\
&\quad\cdot\EE\big[\eul^{p'u_{t}^N(x)}f(X_t^x)\big]^{1/p'},
\end{align*}
for all $\alpha>1$ with $c_\alpha\coloneq\alpha/(\alpha-1)$ and universal constants $b,c,c'>0$. This entails
\begin{align*}
&\frac{\eul^{-tc'p(N-1)(m_{\p}+g^2N)-tc\alpha pg^2N^2\eul^{8\pi\alpha p}}}{b^{N}c_\alpha^{1/2p}}
\int_{\RR^{2N}}f(x)\EE\big[\eul^{\tilde{u}_{\sigma,t}^N(x)}f(X_t^x)\big]\Id x
\\
&\le\bigg(\int_{\RR^{2N}}f(x)\EE\big[\eul^{p'u_{t}^N(x)}f(X_t^x)\big]\Id x\bigg)^{1/p'},
\end{align*}
where we used H\"{o}lders inequality and $\|f\|_1=1$. 
Let $H(\sqrt{p'}g)$ denote the renormalized relativistic Nelson operator with coupling constant $\sqrt{p'}g$.
Then the complex action corresponding to it is $p'u_{t}^N(x)$, and we infer from the previous inequality 
and \cref{corforspecu}(ii) that
\begin{align}\nonumber
\frac{1}{p'}\inf\sigma(H(\sqrt{p'}g))
&=-\lim_{t\to\infty}\frac{1}{p't}\ln \bigg(\int_{\RR^{2N}}f(x)\EE\big[\eul^{p'u_{t}^N(x)}f(X_t^x)\big]\Id x\bigg)
\\\label{upbd2000}
&\le c'p(N-1)(m_{\p}+g^2N)+c\alpha pg^2N^2\eul^{8\pi\alpha p}+\inf\sigma(H_\sigma).
\end{align}
Here $H_\sigma$ again has coupling constant $g$ as usual, and we used the relation
\begin{align*}
\inf\sigma(H_\sigma)&=-\lim_{t\to\infty}\frac{1}{t}
\ln\bigg(\int_{\RR^{2N}}f(x)\EE\big[\eul^{\tilde{u}_{\sigma,t}^N(x)}f(X_t^x)\big]\Id x\bigg),
\end{align*}
which in analogy to \cref{corforspecu}(ii) follows from the Feynman--Kac formula for $H_\sigma$
obtained in \cref{thm:FK-UV} and the proofs of \cite[Theorem~8.3 and Corollaries~8.5\,\&\,8.6]{MatteMoller.2018}.
Combining \eqref{upbd2000} with \eqref{ubGauss2} 
(recall that $\UV_{g,N}$ defined in Corollary~\ref{corubGauss2} is equal to our present choice of $\sigma$),
multiplying by $p'$ and passing to the limit $\alpha\downarrow1$ we obtain
\begin{align*}
\inf\sigma(H(\sqrt{p'}g))
&\le c'p(N-1)(p'm_{\p}+p'g^2N)+cpp'g^2N^2\eul^{8\pi p}
 +\frac{\pi ^{1/2}s^{1/2}p'g^2N^2}{2^{1/2}}
 \\
 &\quad-\frac{2\pi}{\eul^{1/4s}}p'g^2N(N-1)\ln\bigg(\frac{g^2N}{m_{\bos}}\bigg)
\\
&\quad-\frac{\pi p'g^2N}{\eul^{1/4s}}\cdot\frac{2-2\ve}{2-\ve}\cdot[0\vee\ln(c_*^{-1}g^2N)]+\frac{\pi}{4s}p'g^2N
\\
&\quad-\chi_{\{0\}}(m_{\p})\frac{\pi p'g^2N}{\eul^{1/4s}}\ln\bigg(\frac{1\wedge(g^2N)}{\eul m_{\bos}}\bigg),
\end{align*}
for all $s>0$ and $\ve\in(0,1)$ with $c_*>1$ as described in Corollary~\ref{corubGauss2}.
The constants appearing here do not depend on $g$. We have, however, used that $g^2N>m_{\bos}$.
Assuming $g^2N>m_{\bos}p'$ we may thus put $g/\sqrt{p'}$ in place of $g$ in the above bound.
Given $\theta\in(0,1)$, we conclude afterwards by choosing $p=p'=2$, $s\ge1$ large enough
as well as $\ve\in(0,1)$ small enough such that $(2-2\ve)/(2-\ve)\eul^{1/4s}\ge\theta$; we also observe the
elementary bounds $1\wedge(g^2N/2)\ge(1\wedge(g^2N))/2$ and
$0\vee\ln(g^2N/2c_*)\ge\ln(1\vee(g^2N))-\ln(2c_*)$.
\end{proof}

\appendix

\section{The relativistic Kato class}\label{app:Kato}

\noindent
In this \lcnamecref{app:Kato}, we briefly discuss some implications of our standing assumption that the external potential $V$ is Kato decomposable with respect to the process $X$, i.e. that $V_-\in \mc K_X$ and $\chi_K V_+ \in \mc K_X$ for any compact $K\subset \RR^{2N}$. Here, the Kato class $\mc K_X$ is as defined in \cref{def:Kato}.

The following statements are corollaries of a theorem originally proven in \cite{Khasminskii.1959}. Especially, we apply \cref{V2} with the choice $f=pV_-$, $p\ge 1$.
\begin{lem}[{\cite[Proposition 3.8]{ChungZhao.1995}}]
	Assume $f\in \mc K_X$. Then there exist $b,c>0$ such that
	\begin{align}\label{V2}
		\sup_{x\in\RR^{2N}}\EE\left[ \eul^{\int_0^t f(X_s^x)}\Id s\right] \le b\eul^{ct}, \quad t>0.
	\end{align}
	Further,
	\begin{align}\label{eq:KXt0}
		\lim_{t\downarrow 0} \sup_{x\in\RR^{2N}}\EE\left[\eul^{\int_0^t f(X_s^x)\Id s}\right] = 1.
	\end{align}
\end{lem}

Combining \cref{eq:KXt0} with the estimate $|\eul^z-1|^p\le \eul^{p|z|}-1$, we directly obtain the following bound.
\begin{cor}[{\cite[Proposition~3.9]{ChungZhao.1995}}]
	Let $f:\RR^{2N}\to\RR$ be measurable with $f_\pm \in \mc K_X$. Then
	\begin{align}\label{eq:boundexpVcomplete}
		\lim_{t\downarrow 0} \sup_{x\in\RR^{2N}} \EE\left[\left|\eul^{\int_0^t f(X_s^x)\Id s}-1\right|^p\right] = 0,
		\quad p\ge 1.
	\end{align}
\end{cor}

Finally, following \cite[Lemma~C.4]{BroderixHundertmarkLeschke.2000}, we prove an analog of \cref{eq:boundexpVcomplete} for Kato decomposable functions, i.e., functions in which $f_+$ merely belongs to the {\em local} Kato class.

\begin{lem}
	Let $f:\RR^{2N}\to\RR$ be measurable with $f_-\in \mc K_X$. Further, assume $\chi_K f_+\in \mc K_X$ holds for all compact $K\subset \RR^{2N}$. Then, for any compact $K\subset \RR^{2N}$,
	\begin{align}\label{eq:boundexpVminus1}
		\lim_{t\downarrow0}\sup_{x\in K}\EE\left[\left|\eul^{-\int_0^t f(X_s^x)\Id s}-1\right|^p\right] = 0, \quad p\ge 1.
	\end{align}
\end{lem}
\begin{proof}
	Given $x\in\RR^{2N}$ and $\gamma:[0,\infty)\to\RR^{2N}$ Borel measurable, let
	\begin{align*}
		\Delta_{R,t}[x,\gamma] \coloneq \begin{cases}
			1 &\mbox{if}\ |x+\gamma_s|\le R,\ s\in[0,t],\\
			0 & \mbox{else}.
		\end{cases}
	\end{align*}
	Then, by the Cauchy-Schwarz inequality, we have
	\begin{align}\label{eq:expfestimate}
		\begin{aligned}
			\EE\left[\left|\eul^{-\int_0^t f(X_s^x)\Id s}-1\right|^p\right]
			\le\, &
			\EE\left[1-\Delta_{R,t}[x,X]\right]^{1/2}\EE\left[\left|\eul^{-\int_0^t f(X_s^x)\Id s}-1\right|^{2p}\right]^{1/2}\\
			& + \EE\left[(1-\Delta_{R,t}[x,X])\left|\eul^{-\int_0^t f(X_s^x)\Id s}-1\right|^p\right].
		\end{aligned}
	\end{align}
	By the assumptions and \cref{eq:boundexpVcomplete}, we have
	\begin{align*}
		\lim_{t\downarrow 0}\sup_{x\in\RR^{2N}}\EE\left[(1-\Delta_{R,t}[x,X])\left|\eul^{-\int_0^t f(X_s^x)\Id s}-1\right|^p\right] = 0, \quad R>0.
	\end{align*}
	Further, since $f_-\in \mc K_X$, $\EE[|\eul^{-\int_0^t f(X_s^x)\Id s}-1|^{2p}]^{1/2}$ is uniformly bounded in $t\in[0,1]$ and $x\in\RR^{2N}$, by \cref{V2}.
	Finally, using L\'evy's maximal inequality similar to \cite{BroderixHundertmarkLeschke.2000}, also see \cite[Theorem~3.6.5 and \textsection~7]{Simon.1979}, we find
	\begin{align*}
		\EE\left[1-\Delta_{R,t}[x,X]\right] &= \PP\bigg( \sup_{s\in[0,t]  } |X_s^x|\ge R  \bigg) 
		\le \PP\bigg( \sup_{s\in[0,t]} |X_s|\ge R/2  \bigg) \\&\le 2 \PP\left(|X_t|\ge R/2\right), \quad |x|\le R/2.
	\end{align*}
	By the stochastic continuity of $X$, the right hand side converges to zero as $t\downarrow 0$, cf. \cite[p.~43]{Applebaum.2009}.
	Given a compact set $K\subset \RR^{2N}$, the statement now follows,  by choosing $R>0$ such that $K\subset B_{R/2}$ and combining the above observations with \cref{eq:expfestimate}.
\end{proof}

\section{Processes appearing in the creation/annihilation terms}\label{sec:basicproc.UtN}

\noindent
In this \lcnamecref{sec:basicproc.UtN} we derive some basic results on the processes $U^{N,\pm}(x)$
that directly follow from their definition in \cref{defUplusminus,def:UtN}.

In the main text we repeatedly make use of
the following bounds in the time-dependent norm $\|\cdot\|_t\ge\|\cdot\|$ on $L^2(\RR^2)$ given by \cref{deftnorm}:

\begin{lem}
Let $0\le\sigma<\UV\le\infty$, $t>0$ and $x\in\RR^{2N}$. Then 
\begin{align}\label{tnormU}
\|\chi_{B_\UV\setminus B_\sigma}U_{t}^{N,\pm}(x)\|_{t/2}^2\le 
\frac{6\pi}{1-\ve}\cdot\frac{g^2N^2t^{\ve}}{\omega(\sigma)^{1-\ve}}\qquad\mbox{for all}\ \ve\in[0,1).
\end{align}
Further, we have
\begin{align}\label{tnormwU}
		\|\omega \chi_{B_\UV\setminus B_\sigma}U_{t}^{N,\pm}(x)\|_{t/2}^2 \le 6\pi g^2 N^2 \omega(\Lambda), 
		\quad \UV<\infty.
\end{align}
\end{lem}
\begin{proof}
For $a,\ve\ge0$ such that $a+\frac\ve2\le 1$ and $\ve<1$, we find
\begin{align}
\|\omega^{-a}\chi_{B_\UV\setminus B_\sigma}U_{t}^{N,\pm}(x)\|^2
&\le g^2N^2\int_{B_\sigma^c}\int_0^t\int_0^t\frac{\eul^{-s\omega(k)-u\omega(k)}}{\omega(k)^{1+2a}}\Id u\,\Id s\,\Id k
\nonumber\\
&=g^2N^2\int_{B_\sigma^c}\frac{(1-\eul^{-t\omega(k)})^2}{\omega(k)^{3+2a}}\Id k
\nonumber\\
&\le g^2N^2\int_{B_\sigma^c}\frac{t^{2a+\ve}}{\omega(k)^{3-\ve}}\Id k
\nonumber\\
&=\frac{2\pi}{1-\ve}\cdot\frac{g^2N^2t^{2a+\ve}}{\omega(\sigma)^{1-\ve}}. \label{eq:easyboundU}
\end{align}
Applying this estimate with both $a=0$ and $a=1/2$, we arrive at \eqref{tnormU}.
By a similar argument, we have
\begin{align*}
	\|\omega^a\chi_{B_\UV\setminus B_\sigma}U_{t}^{N,\pm}(x)\|^2
	\le \frac{2\pi}{\ve-1}g^2N^2t^{\ve-2a}\omega(\UV)^{\ve-1}, \quad a\ge 0,\ve\in(1,\infty)\cap[2a,2a+2].
\end{align*}
Applying this estimate with $\ve=2$ and both $a=1/2$ and $a=1$ yields \cref{tnormwU}.
\end{proof}

The next two \lcnamecrefs{remUpmcont} ensure in particular the continuity and adaptedness of $U^{N,\pm}(x)$
for both choices of the sign.

\begin{lem}\label{remUpmcont}
Let $\UV\in(0,\infty]$. Then at every fixed elementary event, the map
\begin{align}\label{Upmcont}
[0,\infty)\times\RR^{2N}\ni(t,x)\longmapsto \chi_{B_{\UV}}U_{t}^{N,\pm}(x)\in L^2(\RR^2)\quad\text{is continuous.}
\end{align}
\end{lem}

\begin{proof}
In view of \cref{eq:bounde-somv},
the assertion is clear when the minus sign is chosen in \eqref{Upmcont}.
Let $\gamma:[0,\infty)\to\RR^2$ be Borel measurable. Then
dominated convergence ensures that $I_{j,\UV}:[0,\infty)^2\to L^2(\RR^2)$ 
is continuous for every $\UV\in(0,\infty)$, where
\begin{align*}
I_{\UV}(t,\tau)&\coloneq \int_0^t\eul^{-|\tau-s|\omega-\ii K\cdot \gamma_{s}}v_{\UV}\Id s,\quad t,\tau\ge0.
\end{align*}
Hence, $[0,\infty)\ni t\mapsto \chi_{B_{\UV}} U^+_{t}[\gamma]=I_{\UV}(t,t)$ is continuous 
and so is the map in \eqref{Upmcont} with plus sign as long as $\UV<\infty$.
By virtue of \eqref{tnormU}, with $(\UV,\infty,0)$ put in place of $(\sigma,\UV,\ve)$ there, we further know that,
pointwise on $\Omega$, the norm
$\|\chi_{B_{\UV}}U_{t}^{N,+}(x)-U_{t}^{N,+}(x)\|$ goes to $0$ as $\UV\to\infty$
uniformly in $(t,x)\in[0,\infty)\times\RR^{2N}$. Hence, the statement in \eqref{Upmcont} 
holds for $\UV=\infty$ as well.
\end{proof}

\begin{lem}\label{lem:Upmadapted}
Let $x\in\RR^{2N}$. Then $(U^{N,\pm}_{t}(x))_{t\ge0}$ is predictable.
\end{lem}

\begin{proof}
By virtue of \eqref{Upmcont} it suffices to show that both processes are adapted.
In view of \eqref{def:Ucutoff} we only have to show that $(U_t^\pm[X_{j,\bullet}])_{t\ge0}$ are adapted.

So fix $t>0$ and $j\in\{1,\ldots,N\}$. Then
$(\chi_{[0,t)}(s)\eul^{-(t-s)\omega-\ii K\cdot X_{j,s}}v)_{s\ge0}$ is adapted with
right-continuous paths, hence progressively measurable. In particular the integrand in
$U_{t}^+[X_{j,\bullet}]=\int_{[0,t)}\eul^{-(t-s)\omega-\ii K\cdot X_{j,s}}v\Id s$
is $\fr{B}([0,t])\otimes\fr{F}_t$-measurable, whence $U_{t}^+[X_{j,\bullet}]$ itself
is $\fr{F}_t$-measurable. Likewise,
$(\chi_{(0,\infty)}(s)\eul^{-s\omega-\ii K\cdot X_{j,s-}}v)_{s\ge0}$ is adapted with
left-continuous paths, hence progressively measurable (predictable in fact). 
Furthermore, for every  $\vo\in\Omega$ the left-continuous path $s\mapsto X_{j,s-}(\vo)$ differs
from $s\mapsto X_{j,s}(\vo)$ at most at countably many points, which permits to write
$U_{t}^-[X_{j,\bullet}]=\int_{(0,t]}\eul^{-s\omega-\ii K\cdot X_{j,s-}}v\Id s$ on $\Omega$.
As before we conclude that $U_{t}^-[X_{j,\bullet}]$ is $\fr{F}_t$-measurable. 
\end{proof}

\section{A useful exponential estimate}\label{app:exp}

\noindent
Here we infer an exponential moment bound on some L\'{e}vy type
stochastic integral from a corresponding exponential tail estimate from \cite{Applebaum.2009,Siakalli.2019}. 
We use this bound to treat the martingale contribution to the complex action in \cref{subsec:m}.

To that end, we assume that $\nu_*$ is some L\'evy measure on $\RR^d$ with $d\in\NN$ and 
that $\wt N_*$ is the martingale-valued measure on $[0,\infty)\times\RR^d$ corresponding to
some c\`{a}dl\`{a}g L\'{e}vy process with characteristics $(0,0,\nu_*)$ on the filtered probability space
$(\Omega,\fr{F},(\fr{F}_t)_{t\ge0},\PP)$ satisfying the usual hypothses. 
Further, we assume that $h:[0,\infty)\times \RR^d\times\Omega \to \RR$ is predictable such that, for all $t\ge0$,
\begin{align*}
\sup_{s\in[0,t]}\sup_{z\in\RR^d}|h(s,z)|<\infty\quad\text{and}\quad
\int_0^t\int_{\RR^d}h(s,z)^2\Id \nu_*(z)\,\Id s<\infty,\quad\mbox{$\PP$-a.s.,}
\end{align*}
where $h(s,z)\coloneq h(s,z,\cdot)$ as usual. These are the conditions required in \cite{Applebaum.2009,Siakalli.2019} 
to construct the corresponding c\`{a}dl\`{a}g stochastic integral process $Y$ and to apply the It\^{o} formula for $Y$, where
\begin{align*}
Y_t&\coloneq \int_{(0,t]\times\RR^d}h(s,z)\Id \wt{N}_*(s,z),\quad t\ge0.
\end{align*}
Moreover, these conditions and a bound analogous to \eqref{Taylorexph} ensure that the process $Z_\alpha$ is
well-defined $\PP$-a.s. by
\begin{align*}
Z_{\alpha,t}&\coloneq\frac{1}{\alpha}\int_0^t\int_{\RR^d}(\eul^{\alpha h(s,z)}-1-\alpha h(s,z))\,\Id\nu_*(z)\Id s,\quad t\ge0.
\end{align*}
\begin{lem}\label{lemexpbdtildeN}
In the situation described in the previous paragraph, let $\alpha>1$ and assume in addition that
\begin{align*}
\sup_{s\in[0,t]}Z_{\alpha,s}\le f_\alpha(t),\quad\text{$\PP$-a.s.},
\end{align*}
for all $t\ge0$ and some deterministic function $f_\alpha:[0,\infty)\to\RR$. Then
\begin{align*}
\EE\Big[\sup_{s\in[0,t]}\eul^{Y_s}\Big]&\le
\frac{\alpha}{\alpha-1} \eul^{f_\alpha(t)},\quad t\ge0.
\end{align*}
\end{lem}

\begin{proof}
We know from \cite[Theorem~3.2.5]{Siakalli.2019} (see also \cite[Theorem~5.2.9]{Applebaum.2009}) that
\begin{align}\label{AppSia}
\PP\bigg(\sup_{s\in[0,t]}(Y-Z_\alpha)_s>\tau\bigg)&\le\eul^{-\alpha \tau},\quad \tau,t>0.
\end{align}
Furthermore, it is clear that
\begin{align*}
\EE\Big[\sup_{s\in[0,t]}\eul^{Y_s}\Big]
&\le\EE\bigg[\exp\bigg(\sup_{s\in[0,t]}(Y-Z_{\alpha})_s\bigg)\bigg]\eul^{f_\alpha(t)}.
\end{align*}
For every measurable function $q:\Omega\to\RR$, the layer cake representation
(see, e.g., \cite[Theorem 1.13]{LiebLoss.2001}) entails, however,
\begin{align*}
\int_{\Omega}\eul^{q}\Id\PP&
\le1+\int_{\{q>0\}}(\eul^{q}-1)\Id\PP
=1+\int_0^\infty \PP(q>\tau)\,\eul^{\tau}\Id\tau.
\end{align*}
The exponential tail estimate \eqref{AppSia} thus implies
\begin{align*}
\EE\bigg[\exp\bigg(\sup_{s\in[0,t]}(Y-Z_\alpha)_s\bigg)\bigg]
&\le 1+\int_0^\infty\eul^{-(\alpha-1)\tau}\Id\tau=1+\frac{1}{\alpha-1}.
\end{align*}
Putting all these remarks together we arrive at the asserted bound.
\end{proof}

\section{Bounds on the relativistic semigroup}\label{app:lplq}

\noindent
In this appendix, we consider L\'{e}vy processes in dimension $d\in\NN$ with $d\ge2$, whose
L\'{e}vy symbols are given by some negative constant times the relativistic dispersion relation
\begin{align}\label{defpsid}
	\psi_d(\xi)\coloneq(|\xi|^2+m_{\p}^2)^{1/2}-m_{\p}, \quad \xi\in\RR^d.
\end{align} 
As in the main text, $m_{\p}\ge0$. After briefly introducing the Bessel functions also playing a role in the main part of the article,  we use the second subsection to derive some weighted
$L^p$ to $L^q$ bounds for the semigroup associated with these L\'{e}vy processes, in a straightforward fashion.
In the third subsection, we extend some of these bounds to $N$ independent copies of 
the L\'{e}vy processes.
As we briefly observe in the last subsection, the $L^p$ to $L^\infty$ bounds obtained in the second one
can be combined with Carmona's derivation of his bounds on Kac averages for Brownian
motion \cite{Carmona.1979}, thus yielding bounds on Kac averages for the considered L\'{e}vy processes.
Without doubt all results of this appendix are essentially well-known, but we couldn't find a reference
containing precisely the bounds we need.

In the whole \cref{app:lplq}, unexplained symbols like $c_\nu$, $c_{d,p}$, $c_{d,p,q}$, $c_{d,p}'$ and so on denote positive
constants solely depending on the quantities displayed in their subscripts. Their values might change from
one estimate to another.

\subsection{Brief remarks on some Bessel functions}\label{app:Bessel}

In the main text we repeatedly employed Bessel functions and standard upper bounds on their behavior. For the convenience of the reader, we collect the essential formulas here. Standard literature includes the extensive monography \cite{Watson.1944}.

The Bessel function of the first kind and order $0$ is given by
\begin{align}\label{def:J0-int}
	J_0(s) = \int_0^{2\pi} e^{\ii s \sin t}\frac{\Id t}{2\pi}=J_0(-s),\quad s\in\RR.
\end{align}
In view of this formula and the asymptotic expansion of $J_0(s)$ at $s\to \infty$ we find
some constant $c>0$ such that
\begin{align}\label{eq:J0bound}
	|J_0(r)| \le 1 \wedge \frac{c}{\sqrt{r}},\quad r>0.
\end{align}
Let $\nu\in(0,\infty)$. 
Then the modified Bessel function of the third kind and order $\nu$
is given by the integral formula
\begin{align}\label{forKnur}
	K_\nu(r) = \frac 12 \int_0^\infty t^{\nu-1}\eul^{-r(t+t^{-1})/2}\Id t,\quad r>0.
\end{align}
Estimating $t^{-1}\ge0$ and computing the resulting integral we find the first upper bound in 
\begin{align}\label{bdBesselK}
	0<K_{\nu}(r)&\le
	\begin{cases}
		\displaystyle \frac{c_\nu}{r^\nu}, & r\in(0,1],\\[1em]
		\displaystyle \frac{c_\nu}{r^{1/2}}\eul^{-r}, & r>1.
	\end{cases}
\end{align}
The second one follows from the well-known asymtotics of $K_\nu(r)$ as $r\to\infty$.

\subsection{Weighted \texorpdfstring{$L^p$}{Lp} to \texorpdfstring{$L^q$}{Lq} bounds}

\noindent
In this subsection, we always assume that $X$ is a $d$-dimensional L\'{e}vy process with
symbol $-\psi_d$ given by \eqref{defpsid}. Hence, the law of $X_t$ has a density $\rho_{m_{\p},t}$,
which is equal to $(2\pi)^{-d/2}$ times the inverse unitary Fourier transform of $\eul^{-t\psi_d}$. 
The latter can be computed using the subordination identity, Fourier transforms of Gaussians, elementary
substitutions and \cref{forKnur}. The well-known result is
\begin{align}\label{lawXd0}
	\rho_{0,t}(y)&=\frac{\Gamma((d+1)/2)}{\pi^{(d+1)/2}}\cdot\frac{t}{(t^2+|y|^2)^{(d+1)/2}},\quad y\in\RR^d,
\end{align}
when $m_{\p}=0$, and 
\begin{align}\label{forrhomp}
	\rho_{m_{\p},t}(y)&=\frac{2m_{\p}^{(d+1)/2}}{(2\pi)^{(d+1)/2}}\cdot\frac{t\eul^{m_{\p}t}}{(t^2+|y|^2)^{(d+1)/4}}\cdot
	K_{(d+1)/2}\left(m_{\p}(t^2+|y|^2)^{1/2}\right),
\end{align}
for $m_{\p}>0$; 
cf. also, e.g., \cite[\textsection 7.1]{LiebLoss.2001}. 
For all $t>0$ and $L\ge0$, we abbreviate
\begin{align}\label{defUpsilont}
	\Upsilon_t&\coloneq \{y\in\RR^d|\,m_{\p}(t^2+|y|^2)^{1/2}\le 1\},
	\\\label{defrhoLmt}
	\rho_{L,m_{\p},t}(y)&\coloneq\eul^{L|y|}\rho_{m_{\p},t}(y),\quad y\in\RR^d,
	\\\label{defrho12}
	\rho_{L,m_{\p},t}^{(1)}&\coloneq\chi_{\Upsilon_t}\rho_{L,m_{\p},t},
	\quad
	\rho_{L,m_{\p},t}^{(2)}\coloneq\chi_{\Upsilon_t^c}\rho_{L,m_{\p},t}.
\end{align}

\begin{lem}\label{lemrhoexp1}
	Assume that $m_{\p}>0$. Let $t>0$ and $L\in[0,m_{\p})$. 
	Then the following bounds hold for all $p\in[1,\infty]$,
	\begin{align}\label{LpLinfty1}
		\|\rho_{L,m_{\p},t}^{(1)}\|_p&\le c_{d,p}t^{-d/p'},
		\\\label{LpLinfty2}
		\|\rho_{L,m_{\p},t}^{(2)}\|_p&\le c_{d,p} \eul^{Lt}m_{\p}^{d/2}(m_{\p}-L)^{-d/2p}t^{-d/2p'}.
	\end{align}
	In the case $p=\infty$, \eqref{LpLinfty1} and \eqref{LpLinfty2} also hold for $L=m_{\p}$.
	In particular,
	\begin{align}\label{LpLinftyX}
		\EE[\eul^{L|X_t|}]&=\int_{\RR^d}\rho_{L,m_{\p},t}(y)\Id y\le
		c_{d}\big(1+\eul^{Lt}m_{\p}^{d/2}(m_{\p}-L)^{-d/2}\big).
	\end{align}
\end{lem}

\begin{proof}
	If $m_{\p}(t^2+|y|^2)^{1/2}\le 1$, then $m_{\p}t\le1$, $L|y|\le L/m_{\p}\le 1$. Hence, 
	\eqref{forrhomp} and \eqref{bdBesselK} entail $\rho_{L,m_{\p},t}(y)\le c_d/t^{d}$, which is the case $p=\infty$ of
	\eqref{LpLinfty1}.
	If otherwise $m_{\p}(t^2+|y|^2)^{1/2}>1$,  then \eqref{forrhomp}, \eqref{bdBesselK} and $L\le m_{\p}$ yield
	\begin{align*}
		\rho_{L,m_{\p},t}(y)\le\frac{c_dm_{\p}^{d/2}}{t^{d/2}}\cdot\sup_{r\ge0}\eul^{-m_{p}t(1+r^2)^{1/2}+m_{\p}t+Lrt}.
	\end{align*}
	To obtain \eqref{LpLinfty2} with $p=\infty$, we apply the bound
	\begin{align}\label{Lmpbd}
		-m_{\p}\sqrt{1+r^2}+m_{\p}+Lr&\le-(m_{\p}-L)\big(\sqrt{1+r^2}-1\big)+L,
	\end{align}
	whose right hand side is $\le L$. For finite $p\ge1$, the bound \eqref{LpLinfty1} follows from
	\begin{align*}
		\int_{\Upsilon_t}\rho_{L,m_{\p},t}(y)^{p}\Id y
		&\le c_{d,p}\int_{\RR^d}\frac{t^{p}\eul^{p+pL/m_{\p}}}{(t^2+|y|^2)^{(d+1)p/2}}\Id y,
	\end{align*}
	upon substituting $y=tz$ on the right hand side. Furthermore,
	\begin{align*}
		\int_{\Upsilon_t^c}\rho_{L,m_{\p},t}(y)^{p}\Id y
		&\le c_{d,p}m_{\p}^{dp/2}\int_{\RR^d}
		\frac{t^{p}\eul^{pm_{\p}t+pL|y|}}{(t^2+|y|^2)^{(d+2)p/4}}\cdot
		\eul^{-pm_{\p}(t^2+|y|^2)^{1/2}}\Id y
		\\
		&=c_{d,p}'m_{\p}^{dp/2}t^{d-dp/2}\int_{0}^\infty
		\frac{r^{d-1}\eul^{-pm_{\p}t(1+r^2)^{1/2}+pm_{\p}t+pLrt}}{(1+r^2)^{(d+2)p/4}}\Id r.
	\end{align*}
	To dominate the integral in the last line above we first apply \eqref{Lmpbd}.
	Afterwards we use the bound
	\begin{align*}
		\sqrt{1+r^2}-1&=\frac{r^2}{\sqrt{1+r^2}+1}\ge\frac{r^2}{2\sqrt{1+r^2}}\ge \frac{1}{2\sqrt{2}} (r^2\wedge r),\quad r\ge0.
	\end{align*}
	Together with $r^{d/2}/(1+r^2)^{(d+2)p/4}\le1$ for $r\ge1$, this permits to get
	\begin{align*}
		&\int_{0}^\infty
		\frac{r^{d-1}}{(1+r^2)^{(d+2)p/4}}\cdot
		\eul^{-pm_{\p}t(1+r^2)^{1/2}+pm_{\p}t+pLrt}\Id r
		\\
		&\le \eul^{pLt}\int_{0}^1r^{d-1}\eul^{-p(m_{\p}-L)tr^2/2\sqrt{2}}\Id r
		+ \eul^{pLt}\int_{1}^\infty r^{d/2-1}\eul^{-p(m_{\p}-L)tr/2\sqrt{2}}\Id r
		\\
		&= \eul^{pLt}((m_{\p}-L)t)^{-d/2}\int_0^{\sqrt{(m_{\p}-L)t}}s^{d-1}\eul^{-ps^2/2\sqrt{2}}\Id s
		\\
		&\quad+ \eul^{pLt}((m_{\p}-L)t)^{-d/2}\int_{(m_{\p}-L)t}^\infty u^{d/2-1}\eul^{-pu/2\sqrt{2}}\Id u
		\\
		&\le c_{d,p} \eul^{pLt}((m_{\p}-L)t)^{-d/2},
	\end{align*}
	where we estimated both integrals by integrals over $(0,\infty)$ in the last step. 
\end{proof}

In the following we shall repeatedly employ the bound
\begin{align}\label{convbd1}
	\|h*f\|_q&\le\|h*f\|_\infty^{1-p/q}\|h*f\|_p^{p/q}\le \|h\|_{p'}^{1-p/q}\|h\|_1^{p/q}\|f\|_p,
\end{align}
which holds for all $h\in L^{1}(\RR^\nu)\cap L^{p'}(\RR^\nu)$, $f\in L^p(\RR^\nu)$,
$p\in[1,\infty]$ and $q\in[p,\infty]$ in any dimension $\nu\in\NN$. If $q=\infty$,
then the $L^\infty$-norms on the left hand sides of \eqref{convbd1} as well as of
\cref{pqbdm0,relLpLinfty} below can be replaced by actual suprema.

\begin{lem}\label{lemrelLpLinftym0}
	Suppose that $m_{\p}=0$ and let $p\in[1,\infty]$ and $q\in[p,\infty]$. Then
	\begin{align}\label{pqbdm0}
		\|\rho_{0,t}*f\|_q&\le c_{d,p,q}t^{-d(p^{-1}-q^{-1})}\|f\|_p,\quad f\in L^p(\RR^d),\, t>0.
	\end{align}
\end{lem}

\begin{proof}
	In view of \eqref{convbd1} and $\|\rho_{0,t}\|_1=1$ it suffices to verify that
	$\|\rho_{0,t}\|_{p'}=c_{d,p}t^{-d/p}$ for $p\in[1,\infty)$, which follows by scaling ($y=tz$).
\end{proof}

\begin{lem}\label{lemrelLpLinfty}
	Let $m_{\p}>0$, $p,\tilde{p}\in[1,\infty]$ and $q\in[p\vee\tilde{p},\infty]$. Then 
	\begin{align}\nonumber
		\|\rho_{L,m_{\p},t}*f\|_q&\le c_{d,\tilde{p},q}t^{-d(\tilde{p}^{-1}-q^{-1})}\|f\|_{\tilde{p}}
		\\\label{relLpLinfty}
		&\quad+c_{d,p,q}\eul^{Lt}\bigg(\frac{m_{\p}}{m_{\p}-L}\bigg)^{d/2}
		\bigg(\frac{m_{\p}-L}{t}\bigg)^{d(p^{-1}-q^{-1})/2}\|f\|_p,
	\end{align}
	for all $f\in L^p(\RR^d)\cap L^{\tilde{p}}(\RR^d)$, $L\in[0,m_{\p})$ and $t>0$.
\end{lem}

\begin{proof}
	Again using the notation \eqref{defrhoLmt} and \eqref{defrho12} and applying \eqref{convbd1}
	for two choices of exponents $p$ and $\tilde{p}$ we find
	\begin{align*}
		\|\rho_{L,m_{\p},t}*f\|_q
		&\le\|\rho_{L,m_{\p},t}^{(1)}\|_{\tilde{p}'}^{1-\tilde{p}/q}
		\|\rho_{L,m_{\p},t}^{(1)}\|_1^{\tilde{p}/q}\|f\|_{\tilde{p}}
		+\|\rho_{L,m_{\p},t}^{(2)}\|_{p'}^{1-p/q}\|\rho_{L,m_{\p},t}^{(2)}\|_1^{p/q}\|f\|_{p}.
	\end{align*}
	Combining this bound with \eqref{LpLinfty1} and \eqref{LpLinfty2} 
	(with $1,\tilde{p}'$ and $1,p'$ put in place of $p$, respectively)
	we arrive at \eqref{relLpLinfty} for $m_{\p}>0$.
\end{proof}

By the previous two lemmas and the rotation invariance of $\rho_{L,m_{\p},t}$, 
\begin{align}\label{eq:convolution}
	\EE[\eul^{L|X_t|}f(X_t^x)]&=(\rho_{L,m_{\p},t}*f)(x)=
	(\rho_{L,m_{\p},t}^{(1)}*f)(x)+(\rho_{L,m_{\p},t}^{(2)}*f)(x),
\end{align}
for all $t>0$, $x\in\RR^{d}$, $f\in L^p(\RR^d)$ and $p\in[1,\infty]$, if $0\le L<m_{\p}$ or $L=m_{\p}=0$.
The previous two lemmas combined with the case $L=0$ of \cref{eq:convolution} further imply the following result.

\begin{cor}\label{correlLpLinfty}
	Considering arbitrary $m_{\p}\ge0$ and $p\in[1,\infty]$, we always have
	\begin{align}\label{relLpLinfty2}
		\sup_{x\in\RR^d}\EE[|f(X_t^x)|]
		&\le c_{d,p}t^{-d/2p}(\|f\|_{2p}+m_{\p}^{d/2p}\|f\|_p),\quad t>0,
	\end{align}
	provided that $f\in L^p(\RR^d)\cap L^{2p}(\RR^d)$, if $m_{\p}>0$, and $f\in L^{2p}(\RR^d)$, if $m_{\p}=0$.
\end{cor}

\subsection{Weighted \texorpdfstring{$L^p$}{Lp} to \texorpdfstring{$L^q$}{Lq} bounds for \texorpdfstring{$N$}{N} relativistic particles}

\noindent
We also need a variant of our weighted $L^p$ to $L^q$ bounds applying to a vector 
$X=(X_1,\ldots,X_N)$ comprising $N$ independent L\'{e}vy processes $X_j$,
each having L\'{e}vy symbol $-\psi_d$ with $m_{\p}\ge0$.  To state the next result we write
\begin{align*}
	\rho_{L,m_{\p},t}^{\otimes_N}(y)\coloneq\rho_{L,m_{\p},t}(y_1)\cdots\rho_{L,m_{\p},t}(y_N),
	\quad y=(y_1,\ldots,y_N)\in(\RR^d)^N.
\end{align*}
The independence of $X_1,\ldots,X_N$ entails (recall our definition \eqref{1norm2} of $|\cdot|_1$)
\begin{align}\label{relLpLqN1}
	\EE[\eul^{L|X_t^x|_1}|f(X_t^x)|]&= (\rho_{L,m_{\p},t}^{\otimes_N}*|f|)(x),\quad x\in\RR^{dN},
	\,f\in L^p(\RR^{dN}).
\end{align}

\begin{lem}\label{lemrelLpLinftyN}
	In the situation described in the preceding paragraph let $t>0$, $L\in[0,m_{\p})$ for $m_{\p}>0$ and
	$L=0$ for $m_{\p}=0$,
	$p\in[1,\infty]$ and $q\in[p,\infty]$. Then
	\begin{align}\label{relLpLinftyN}
		\|\rho_{L,m_{\p},t}^{\otimes_N}*f\|_q&\le c_{d,p,q}^N\Theta_{d,p,q}(m_{\p},L,t)^N\|f\|_p,
	\end{align}
	for all $f\in L^p(\RR^{dN})$, where
	\begin{align*}
		\Theta_{d,p,q}(m_{\p},L,t)&\coloneq
		\frac{1}{t^{d(p^{-1}-q^{-1})}}+\eul^{Lt}\bigg(\frac{m_{\p}}{m_{\p}-L}\bigg)^{d/2}
		\bigg(\frac{m_{\p}-L}{t}\bigg)^{d(p^{-1}-q^{-1})/2},
	\end{align*}
	for $m_{\p}>0$, and $\Theta_{d,p,q}(0,0,t)\coloneq t^{-d(p^{-1}-q^{-1})}$.
\end{lem}

\begin{proof}
	We put
	$\rho_{L,m_{\p},t}^{(i_1,...,i_N)}(y)\coloneq \rho_{L,m_{\p},t}^{(i_1)}(y_1)\cdots\rho_{L,m_{\p},t}^{(i_N)}(y_N)$
	for $i_1,\ldots,i_N\in\{1,2\}$. On account of \eqref{convbd1},
	\begin{align*}
		\|\rho_{L,m_{\p},t}^{\otimes_N}*f\|_q&\le \sum_{i_1,...,i_N=1}^2
		\|\rho_{L,m_{\p},t}^{(i_1,...,i_N)}*f\|_q
		\\
		&\le \sum_{i_1,...,i_N=1}^2
		\|\rho_{L,m_{\p},t}^{(i_1,...,i_N)}\|_{p'}^{1-p/q}\|\rho_{L,m_{\p},t}^{(i_1,...,i_N)}\|_1^{p/q}\|f\|_p
		\\
		&=\bigg(\sum_{i=1}^2\|\rho_{L,m_{\p},t}^{(i)}\|_{p'}^{1-p/q}\|\rho_{L,m_{\p},t}^{(i)}\|_1^{p/q}\bigg)^N\|f\|_p.
	\end{align*}
	Here we used that $\|\rho_{L,m_{\p},t}^{(i_1,...,i_N)}\|_{\tilde{p}}=
	\|\rho_{L,m_{\p},t}^{(i_1)}\|_{\tilde{p}}\cdots\|\rho_{L,m_{\p},t}^{(i_N)}\|_{\tilde{p}}$
	for all $\tilde{p}\in[1,\infty]$
	in the last equality. We conclude exactly as in the proofs of Lemma~\ref{lemrelLpLinftym0} 
	and Lemma~\ref{lemrelLpLinfty}, again using \eqref{LpLinfty1} and  \eqref{LpLinfty2} for $m_{\p}>0$.
\end{proof}

\begin{cor}\label{corwLpLqNW}
Let $\scr{H}$ be a separable Hilbert space and let $Y:\RR^{dN}\times\Omega\to\LO(\scr{H})$
be strongly product measurable. Write $Y(x)\coloneq Y(x,\cdot)$ and assume that
\begin{align*}
C(\theta)&\coloneq  \sup_{x\in\RR^{dN}}\EE[\|Y(x)\|^{\theta}]^{1/\theta}<\infty,\quad \theta\in[1,\infty).
\end{align*}
Let $t\ge0$, $p\in(1,\infty]$, $q\in[p,\infty]$ and $\Psi\in L^p(\RR^{dN},\scr{H})$. Finally, let
$G$, $L$ and $r$ be given as in \cref{prop:LpLq}. Then $Y(x)(\eul^{-G}\Psi)(X_t^x)$ is $\PP$-integrable
for all $x\in\RR^{dN}$ and the function $\Phi:\RR^{dN}\to\scr{H}$ defined by
\begin{align*}
\Phi(x)&\coloneq\eul^{G(x)}\EE[Y(x)(\eul^{-G}\Psi)(X_t^x)],\quad x\in\RR^{dN},
\end{align*}
is in $L^q(\RR^{dN},\scr{H})$ with
	\begin{align*}
\|\Phi\|_q
&\le c_{d,p,q}^Na(m_{\p},rL)^{\frac{dN}{2}}
\bigg[\frac{1}{t^2}\vee\frac{m_{\p}-L}{t}\bigg]^{\frac{dN}{2}(\frac{1}{p}-\frac{1}{q})}
\eul^{NLt}C(p'r')\|\Psi\|_p,
\end{align*}
where $p'$ and $r'$ denote the dual exponents of $p$ and $r$, respectively.
In the above the essential supremum in $\|\Phi\|_\infty$ can be replaced by a supremum.
Furthermore, $a(m_{\p},rL)\coloneq m_{\p}/(m_{\p}-rL)$, if $m_{\p}>0$, and
$a(0,rL)\coloneq1$.
\end{cor}

\begin{proof}
Successively applying the bound $|G(x)-G(X_t^x)|\le L|X_t|_1$ and
H\"{o}lder's inequality with $p^{-1}+(rp')^{-1}+(r'p')^{-1}=1$, we find
\begin{align*}
\EE\big[\eul^{G(x)}\|Y(x)(\eul^{-G}\Psi)(X_t^x)\|_{\scr{H}}\big]
&\le C(p'r')\EE[\eul^{rL|X_t|_1}]^{1/p'r}\big\{\EE[\eul^{L|X_t|_1}\|\Psi(X_t^x)\|_{\scr{H}}^p]^{1/p}\big\},
\end{align*}
for all $x\in\RR^{dN}$ when $p<\infty$. 
If $p=\infty$, the expression in curly brackets $\{\cdots\}$
should be replaced by $\|\Psi\|_\infty$.
Applying \eqref{relLpLqN1} and \eqref{relLpLinftyN} (with $f=1$ and $p=q=\infty$), we find
\begin{align}\label{expbdLXt}
\EE[\eul^{r L|X_t|_1}]^{1/p'r}&\le c_d^{N/p'r}\eul^{NLt/p'}\bigg(\frac{m_{\p}}{m_{\p}-r L}\bigg)^{dN/2p'r}.
\end{align}
This finishes the proof in the case $p=\infty$, whence we assume $p<\infty$ from now on.
Since $\|\Psi(\cdot)\|_{\scr{H}}^p$ is in $L^1(\RR^{2N})$
with $L^1$-norm $\|\Psi\|_p^p$, we can also apply
\eqref{relLpLqN1} and \eqref{relLpLinftyN} (with $1$ and $q/p$ put in place
of $p$ and $q$, respectively) to the function given by
$h(x)\coloneq \EE[\eul^{L|X_t|_1}\|\Psi(X_t^x)\|_{\scr{H}}^p]$. 
Since $\|h^{1/p}\|_q=\|h\|_{q/p}^{1/p}$, this leads to
\begin{align}\label{hqbd}
\|h^{1/p}\|_q&\le c_{d,p,q}^N\eul^{NLt/p}\bigg(\frac{m_{\p}}{m_{\p}-L}\bigg)^{\frac{dN}{2p}}
\bigg[\frac{1}{t^2}\vee\frac{m_{\p}-L}{t}\bigg]^{\frac{dN}{2}(\frac{1}{p}-\frac{1}{q})}\|\Psi\|_p.
\end{align}
Combining the above H\"{o}lder estimate with \cref{expbdLXt,hqbd} we arrive at the asserted bound for finite $p$.
\end{proof}

\subsection{Carmona's bound on Kac averages in the relativistic case}

\noindent
Next, we observe the relativistic variant \eqref{relLpLinfty3} of a bound derived 
by Carmona for Brownian motions. The crucial aspect of \eqref{relLpLinfty3} is the
comparatively explicit dependence on the potential $v$ of the exponent on the right hand side.
This becomes important in the derivation of our bounds on the minimal energy of the
relativistic Nelson operator.

\begin{thm}
	Let $a>0$ and $Y$ be a $d$-dimensional L\'{e}vy process with L\'{e}vy symbol $-a\psi_d$,
	where $\psi_d$ is given by \eqref{defpsid}.
	Furthermore, let $p>d/2$ and assume that the Borel measurable function $v:\RR^d\to[0,\infty]$ is 
	$2p$-integrable, if $m_{\p}=0$, and
	both $p$- and $2p$-integrable, if $m_{\p}>0$. Then
	\begin{align}\nonumber
		&\sup_{x\in\RR^d}\EE\bigg[\exp\bigg(\int_0^tv(x+Y_s)\Id s\bigg)\bigg]
		\\\label{relLpLinfty3}
		&\le c_{d,p}'\exp\big(c_{d,p}a^{-d/(2p-d)}(m_{\p}^{d/2p}\|v\|_p+\|v\|_{2p})^{1/(1-d/2p)}t\big).
	\end{align}
\end{thm}

\begin{proof}
	Thanks to \eqref{relLpLinfty2} we can almost literally follow Carmona's proofs in
	\cite[Theorem~2.1 and Remark~3.1]{Carmona.1979} to obtain \eqref{relLpLinfty3} with $a=1$.
	The general case then follows from the equivalence of $Y$ and $(X_{at})_{t\ge0}$ and
	elementary substitutions; here $X$ again has L\'{e}vy symbol $-\psi_d$.
\end{proof}


\subsection*{Acknowledgements}
BH wants to thank Aalborg Universitet, where this work was initiated, for their kind hospitality.
OM is grateful for support by the Independent Research Fund Denmark via the 
project grant ``Mathematical Aspects of Ultraviolet Renormalization'' (8021-00242B).
The authors thank an anonymous referee for the constructive comments.


\bibliographystyle{halpha-abbrv}
\bibliography{../../Literature/00lit}

\begin{thebibliography}{GMM17}
\expandafter\ifx\csname url\endcsname\relax
  \def\url#1{\texttt{#1}}\fi
\expandafter\ifx\csname doi\endcsname\relax
  \def\doi#1{\burlalt{doi:#1}{http://dx.doi.org/#1}}\fi
\expandafter\ifx\csname urlprefix\endcsname\relax\def\urlprefix{URL }\fi
\expandafter\ifx\csname href\endcsname\relax
  \def\href#1#2{#2}\fi
\expandafter\ifx\csname burlalt\endcsname\relax
  \def\burlalt#1#2{\href{#2}{#1}}\fi

\bibitem[AM21]{AlvarezMoller.2022}
B.~Alvarez and J.~S. M{\o}ller.
\newblock Ultraviolet Renormalisation of a quantum field toy model I.
\newblock Preprint, 2021,
  \burlalt{arXiv:2103.13770}{http://arxiv.org/abs/2103.13770}.

\bibitem[App09]{Applebaum.2009}
D.~Applebaum.
\newblock {\em {L}{\'{e}}vy Processes and Stochastic Calculus}, volume 116 of
  {\em Cambridge Studies in Advanced Mathematics}.
\newblock Cambridge University Press, Cambridge, 2nd edition, 2009.
\newblock \doi{10.1017/CBO9780511809781}.

\bibitem[Ara18]{Arai.2018}
A.~Arai.
\newblock {\em Analysis on {F}ock Spaces and Mathematical Theory of Quantum
  Fields}.
\newblock World Scientific, New Jersey, 2018.
\newblock \doi{10.1142/10367}.

\bibitem[BHL00]{BroderixHundertmarkLeschke.2000}
K.~Broderix, D.~Hundertmark, and H.~Leschke.
\newblock Continuity properties of {S}chr\"{o}dinger semigroups with magnetic
  fields.
\newblock {\em Rev. Math. Phys.}, 12(2):181--225, 2000,
  \burlalt{arXiv:math-ph/9808004}{http://arxiv.org/abs/math-ph/9808004}.
\newblock \doi{10.1142/S0129055X00000083}.

\bibitem[Ble18]{Bley.2018}
G.~A. Bley.
\newblock A lower bound on the renormalized {N}elson model.
\newblock {\em J. Math. Phys.}, 59(6):061901, 2018,
  \burlalt{arXiv:1609.08590}{http://arxiv.org/abs/1609.08590}.
\newblock \doi{10.1063/1.5008831}.

\bibitem[Car79]{Carmona.1979}
R.~Carmona.
\newblock Regularity properties of {S}chr\"{o}dinger and {D}irichlet
  semigroups.
\newblock {\em J. Funct. Anal.}, 33(3):259--296, 1979.
\newblock \doi{10.1016/0022-1236(79)90068-5}.

\bibitem[CMS90]{CarmonaMastersSimon.1990}
R.~Carmona, W.~C. Masters, and B.~Simon.
\newblock Relativistic {S}chr\"{o}dinger Operators: Asymptotic Behavior of the
  Eigenfunctions.
\newblock {\em J. Funct. Anal.}, 91(1):117--142, 1990.
\newblock \doi{10.1016/0022-1236(90)90049-Q}.

\bibitem[CZ95]{ChungZhao.1995}
K.~L. Chung and Z.~X. Zhao.
\newblock {\em From {B}rownian Motion to {S}chr\"{o}dinger's Equation}, volume
  312 of {\em Grundlehren der mathematischen Wissenschaften}.
\newblock Springer, Berlin, 1995.
\newblock \doi{10.1007/978-3-642-57856-4}.

\bibitem[DP14]{DeckertPizzo.2014}
D.-A. Deckert and A.~Pizzo.
\newblock Ultraviolet Properties of the Spinless, One-Particle {Y}ukawa Model.
\newblock {\em Commun. Math. Phys.}, 327(3):887--920, 2014,
  \burlalt{arXiv:1208.2646}{http://arxiv.org/abs/1208.2646}.
\newblock \doi{10.1007/s00220-013-1877-9}.

\bibitem[GHL14]{GubinelliHiroshimaLorinczi.2014}
M.~Gubinelli, F.~Hiroshima, and J.~L\H{o}rinczi.
\newblock Ultraviolet renormalization of the {N}elson {H}amiltonian through
  functional integration.
\newblock {\em J. Funct. Anal.}, 267(9):3125--3153, 2014,
  \burlalt{arXiv:1304.6662}{http://arxiv.org/abs/1304.6662}.
\newblock \doi{10.1016/j.jfa.2014.08.002}.

\bibitem[GMM17]{GueneysuMatteMoller.2017}
B.~G{\"{u}}neysu, O.~Matte, and J.~S. M{\o }ller.
\newblock Stochastic differential equations for models of non-relativistic
  matter interacting with quantized radiation fields.
\newblock {\em Probab. Theory Relat. Fields}, 167(3-4):817--915, 2017,
  \burlalt{arXiv:1402.2242}{http://arxiv.org/abs/1402.2242}.
\newblock \doi{10.1007/s00440-016-0694-4}.

\bibitem[Gro73]{Gross.1973}
L.~Gross.
\newblock The relativistic polaron without cutoffs.
\newblock {\em Commun. Math. Phys.}, 31(1):25--73, 1973.
\newblock \doi{10.1007/BF01645589}.

\bibitem[HJ94]{HoffmannJorgensen.1994a}
J.~Hoffmann-J{\o }rgensen.
\newblock {\em Probability with a View Toward Statistics}, volume~1.
\newblock Chapman \& Hall, New York, 1994.
\newblock \doi{10.1007/978-1-4899-3019-4}.

\bibitem[HM22]{HiroshimaMatte.2019}
F.~Hiroshima and O.~Matte.
\newblock Ground states and their associated path measures in the renormalized
  {N}elson model.
\newblock {\em Rev. Math. Phys.}, 34(2):2250002, 2022,
  \burlalt{arXiv:1903.12024}{http://arxiv.org/abs/1903.12024}.
\newblock \doi{10.1142/S0129055X22500027}.

\bibitem[IW81]{IkedaWatanabe.1981}
N.~Ikeda and S.~Watanabe.
\newblock {\em Stochastic Differential Equations and Diffusion Processes},
  volume~24 of {\em North-Holland Mathematical Library}.
\newblock North-Holland Publishing, Amsterdam, 1981.

\bibitem[Kha59]{Khasminskii.1959}
R.~Z. Khas'minski{\u{\i}}.
\newblock On Positive Solutions of the Equation
  $\ifx\mathfrak\undefined{A}\else{\mathfrak A}\fi u + Vu=0$.
\newblock {\em Theory Probab. Appl.}, 4(3):309--318, 1959.
\newblock \doi{10.1137/1104030}.

\bibitem[Lam21]{Lampart.2020}
J.~Lampart.
\newblock The Resolvent of the {N}elson {H}amiltonian Improves Positivity.
\newblock {\em Math. Phys. Anal. Geom.}, 24(1):2, 2021,
  \burlalt{arXiv:2010.03235}{http://arxiv.org/abs/2010.03235}.
\newblock \doi{10.1007/s11040-021-09374-6}.

\bibitem[LL01]{LiebLoss.2001}
E.~H. Lieb and M.~Loss.
\newblock {\em Analysis}, volume~14 of {\em Graduate Studies in Mathematics}.
\newblock AMS, Providence, 2nd edition, 2001.
\newblock \doi{10.1090/gsm/014}.

\bibitem[Mat16]{Matte.2016}
O.~Matte.
\newblock Continuity properties of the semi-group and its integral kernel in
  non-relativistic {QED}.
\newblock {\em Rev. Math. Phys.}, 28(05):1650011, 2016,
  \burlalt{arXiv:1512.04494}{http://arxiv.org/abs/1512.04494}.
\newblock \doi{10.1142/S0129055X16500112}.

\bibitem[M{\'{e}}t82]{Metivier.1982}
M.~M{\'{e}}tivier.
\newblock {\em Semimartingales: A Course on Stochastic Processes}, volume~2 of
  {\em De Gruyter Studies in Mathematics}.
\newblock De Gruyter, Berlin, 1982.
\newblock \doi{10.1515/9783110845563}.

\bibitem[Miy19]{Miyao.2018}
T.~Miyao.
\newblock On the semigroup generated by the renormalized {N}elson
  {H}amiltonian.
\newblock {\em J. Funct. Anal.}, 276(6):1948--1977, 2019,
  \burlalt{arXiv:1803.08659}{http://arxiv.org/abs/1803.08659}.
\newblock \doi{10.1016/j.jfa.2018.11.001}.

\bibitem[MM18]{MatteMoller.2018}
O.~Matte and J.~S. M\o{}ller.
\newblock {F}eynman-{K}ac Formulas for the Ultra-Violet Renormalized {N}elson
  Model.
\newblock {\em Ast\'{e}risque}, 404, 2018,
  \burlalt{arXiv:1701.02600}{http://arxiv.org/abs/1701.02600}.
\newblock \doi{10.24033/ast.1054}.

\bibitem[Nel64a]{Nelson.1964}
E.~Nelson.
\newblock Interaction of Nonrelativistic Particles with a Quantized Scalar
  Field.
\newblock {\em J. Math. Phys.}, 5(9):1190--1197, 1964.
\newblock \doi{10.1063/1.1704225}.

\bibitem[Nel64b]{Nelson.1964c}
E.~Nelson.
\newblock Schr\"{o}dinger particles interacting with a quantized scalar field.
\newblock In W.~T. Martin and I.~Segal, editors, {\em Analysis in Function
  Space: Proceedings of a Conference on the Theory and Application of Analysis
  in Function Space Held At Endicott House in Dedham, Massachusetts, June 9-13,
  1963}, pages 87--120, Cambridge, 1964. MIT Press.

\bibitem[Par92]{Parthasarathy.1992}
K.~R. Parthasarathy.
\newblock {\em An Introduction to Quantum Stochastic Calculus}, volume~85 of
  {\em Monographs in Mathematics}.
\newblock Birkh\"{a}user, Basel, 1992.
\newblock \doi{10.1007/978-3-0348-0566-7}.

\bibitem[RS75]{ReedSimon.1975}
M.~Reed and B.~Simon.
\newblock {\em {F}ourier Analysis, Self-Adjointness}, volume~2 of {\em Methods
  of Modern Mathematical Physics}.
\newblock Academic Press, San Diego, 1975.

\bibitem[RS80]{ReedSimon.1980}
M.~Reed and B.~Simon.
\newblock {\em Functional Analysis}, volume~1 of {\em Methods of Modern
  Mathematical Physics}.
\newblock Academic Press, San Diego, revised and enlarged edition, 1980.

\bibitem[Sch19]{Schmidt.2019}
J.~Schmidt.
\newblock On a direct description of pseudorelativistic {N}elson
  {H}amiltonians.
\newblock {\em J. Math. Phys.}, 60(10):102303, 2019,
  \burlalt{arXiv:1810.03313}{http://arxiv.org/abs/1810.03313}.
\newblock \doi{10.1063/1.5109640}.

\bibitem[Sia09]{Siakalli.2019}
M.~Siakalli.
\newblock {\em Stability properties of stochastic differential equations driven
  by L{\'e}vy noise}.
\newblock PhD thesis, University of Sheffield, 2009.
\newblock \urlprefix\url{https://etheses.whiterose.ac.uk/15019/}.

\bibitem[Sim79]{Simon.1979}
B.~Simon.
\newblock {\em Functional Integration and Quantum Physics}, volume~86 of {\em
  Pure and Applied Mathematics}.
\newblock Academic Press, New York, 1979.

\bibitem[Slo74a]{Sloan.1974b}
A.~D. Sloan.
\newblock Analytic domination with quadratic form type estimates and
  nondegeneracy of ground states in quantum field theory.
\newblock {\em Trans. Amer. Math. Soc.}, 194:325--336, 1974.
\newblock \doi{10.2307/1996809}.

\bibitem[Slo74b]{Sloan.1974}
A.~D. Sloan.
\newblock The polaron without cutoffs in two space dimensions.
\newblock {\em J. Math. Phys.}, 15:190--201, 1974.
\newblock \doi{10.1063/1.1666620}.

\bibitem[Wat44]{Watson.1944}
G.~N. Watson.
\newblock {\em A Treatise on the Theory of Bessel Functions}.
\newblock Cambridge University Press, Cambridge, 1944.

\end{thebibliography}

\end{document}